%% file: rotation.tex
\documentclass[a4paper]{article}

\usepackage{times}

\usepackage[pdftex]{graphicx}
\usepackage[pdftex]{hyperref}
\usepackage{amsmath,amssymb,amsthm}

\theoremstyle{plain}
\newtheorem{theorem}{Theorem}
\newtheorem{lemma}[theorem]{Lemma}

\theoremstyle{definition}

\theoremstyle{remark}

\usepackage{triotex}
\usepackage{xcolor}
\usepackage{afterpage}
\usepackage{changepage}

\newif\iffoot\footfalse
\let\originalfootnote\footnote
\renewcommand{\footnote}[1]{\foottrue\originalfootnote{#1}\footfalse}

\usepackage{listings} 
\usepackage{eiffel-pseudo}
\usepackage{boogie}
\usepackage{dafny}
\usepackage[T1]{fontenc}
\usepackage[scaled=0.81]{beramono} 
\def\<#1>{\iffoot\mbox{\lstinline[basicstyle=\footnotesize\ttfamily]{#1}}\else\mbox{\lstinline!#1!}\fi}
\def\@#1@{\text{\lstinline|#1|}}

\usepackage{array}

\newcommand{\limpl}{\Longrightarrow}
\newcommand{\liff}{\Longleftrightarrow}

\usepackage{tikz}
\usetikzlibrary{shapes,arrows,automata,shapes.symbols,positioning,calc,fit}
\definecolor{left-color}{RGB}{241,163,64}
\definecolor{right-color}{RGB}{153,142,195}
\tikzstyle{left-item}=
[
  rectangle, minimum size=5mm, very thick,
  rounded corners=1mm, draw=left-color!80!black!80, font=\normalsize, 
  bottom color=left-color, top color=left-color!10!white
]
\tikzstyle{right-item}=
[
  rectangle, minimum size=5mm, very thick,
  rounded corners=1mm, draw=right-color!80!black!80, font=\normalsize, 
  fill=right-color, text=white
]
\tikzstyle{generic-item}=
[
  rectangle, minimum size=5mm,very thick,rounded corners=1mm,draw=black,font=\normalsize
]
\tikzstyle{fake-item}=
[
  rectangle, minimum size=5mm,draw=none,font=\normalsize
]
\tikzstyle{array-slice}=
[
  rectangle, minimum height=10mm,thick,draw=black,font=\footnotesize,align=center
]

\usepackage{subcaption}

\newcommand{\sq}{\ensuremath{\mathsf{S}}}
\newcommand{\tq}{\ensuremath{\mathsf{T}}}
\newcommand{\xq}{\ensuremath{\mathsf{X}}}
\newcommand{\yq}{\ensuremath{\mathsf{Y}}}
\newcommand{\zq}{\ensuremath{\mathsf{Z}}}

\newcommand{\cat}[2]{{#1}\mathop{\circ}{#2}}
\newcommand{\rot}[2]{\rho^{#2}{#1}}
\newcommand{\rev}[1]{{#1}^{-1}}
\DeclareRobustCommand{\mpos}[4]{\pi^{#2}_{#1}{\ifthenelse{\not \equal{#4}{}}  {\left({#3, #4}\right)} {}}}
\DeclareRobustCommand{\ic}[2]{\tau\ifthenelse{\not \equal{#2}{}}  {\left({#1, #2}\right)} {}}

\DeclareMathOperator{\bwr}{wrap}
\DeclareRobustCommand{\wrap}[2]{\ifthenelse{\equal{#1}{} \and \equal{#2}{}}{\bwr}{{#1}\bwr{#2}}}
\DeclareRobustCommand{\len}[1]{\left|{#1}\right|}

\newcommand{\setn}[1]{\langle #1 \rangle}

\usepackage{xspace}
\newcommand{\escj}{{ESC/Ja\-va2}\xspace}

\begin{document}

\title{Rotation of Sequences: Algorithms and Proofs}

\date{20 June 2014\\Revised 27 November 2014, 5 February 2015}

\author{Carlo A. Furia \\
ETH Zurich, Switzerland \\
\url{http://bugcounting.net}
}

\maketitle

\begin{abstract}
Sequence rotation consists of a circular shift of the sequence's elements by a given number of positions.
We present the four classic algorithms to rotate a sequence; the loop invariants underlying their correctness; detailed correctness proofs; and fully annotated versions for the verifiers Boogie, Dafny, and \escj.
The presentation illustrates in detail both how the algorithms work and what it takes to carry out mechanized proofs of their correctness.
\end{abstract}

\newpage
\tableofcontents
\newpage

\section{Introduction}

Rotating a sequence, typically represented by an array, is the problem of shifting all its elements in one direction by a fixed number of positions while cyclically wrapping over the sequence's bounds.
Rotations have various practical applications in programming; a classic example is in editors, to arrange the lines of text stored as sequences of characters.

There is an interesting choice of algorithms to rotate a sequence, exploring different trade offs between performance, code complexity, and resource usage.
The algorithms themselves have been part of computer science folklore knowledge for a long time~\cite[Sec.~2.3]{pearls-book}; they are an excellent subject to discuss algorithm design and analysis~\cite{Bornat-slides}.

Proving correctness of these algorithms is a challenging task too, even more so if we aim for (automated) mechanized proofs.
The only document formally discussing how to establish the correctness of the rotation algorithms is a technical report by Gries and Mills~\cite{swapping-TR}.
They sketch the loop invariants underlying the correctness arguments; the level of detail of their presentation is, however, clearly insufficient to perform complete mechanized proofs, which require to exhaustively deal with issues such as framing and intermediate properties and assertions.

This paper presents each rotation algorithm for rotation first informally, then with detailed loop invariants and lemmas, and finally with all the gory low-level details necessary to carry out automated proofs using the Boogie prover~\cite{BoogieManual}, the Dafny verifier~\cite{Dafny}, and the \escj verifier~\cite{Burdy:2005:OJT:1070908.1070911} for Java code annotated with JML.
This presentation can serve as a useful guide to carry out similar correctness proofs of the same algorithms using other automated tools, as well as a tutorial introduction to some idioms (lemma procedures, framing, ghost code, and so on) frequently used in automated verification of full functional correctness.

The code presented in the paper is available online (under directory \<rotation>):
\begin{center}
\url{https://bitbucket.org/caf/verified/}
\end{center}

\section{Modular arithmetic and sequences}
\label{sec:modular-arithmetic}

A precise definition of the notion of rotation relies on some mathematical concepts that we introduce in this section.

\paragraph{Modular arithmetic.}
Modular arithmetic makes extensive usage of the `$\bmod$' binary operation, which is normally defined~\cite[Sec.~3.4]{GKP94} in terms of integer division and floor as
\begin{equation}
x \bmod y \quad=\quad x - y \lfloor x / y \rfloor\,, \qquad \text{for } y \neq 0.
\label{eq:mod-def}
\end{equation}
This definition is not, however, always the best choice to specify and reason about programs.
First, programming languages may implement definitions of `$\bmod$' that differ in the sign of the result; for example, \lstinline[language=Java]|x % y| 
has the same sign as \<x> in Java but as \<y> in Python.
Second, \eqref{eq:mod-def} relies on two other operations, whereas it would be more convenient to have a direct definition only in terms of basic operators that are universally available.

These considerations justify the introduction of the binary operation `$\wrap{}{}$' with recursive definition
\begin{equation}
\wrap{x}{y} \quad=\quad
\begin{cases}
x  
&  0 \leq x < y\,,  \\
\wrap{(x - y)}{y}
&  0 < y \leq x\,.
\end{cases}
\label{eq:wrap-def}
\end{equation}
We only need to define `$\wrap{}{}$' for nonnegative arguments, although we could easily generalize \eqref{eq:wrap-def} to handle negative arguments too.

When it is defined, it is easy to see that $\wrap{x}{y}$ is the same as $x \bmod y$.
Specifically, we will make use of the property that
\begin{equation}
0 \leq (\wrap{x}{y}) < y\,,   \qquad \text{for } x \geq 0 \text{ and } y > 0\,,
\label{eq:wrap_bounds}
\end{equation}
which can be proved by induction on $x$ (base case: $x < y$).

\paragraph{Sequences.}
Sequences are finite ordered collections of elements, which we denote by sans-serif letters such as $\sq$.
The \emph{length} of a sequence $\sq$ is denoted by $\len{\sq}$.
An \emph{element} of $\sq$ is denoted by $\sq_k$ (also: $\sq[k]$), with $0 \leq k < \len{\sq}$ denoting the position (or index) of the element in the sequence starting from $0$:
\[
\sq \quad=\quad \sq_0\, \sq_1\, \cdots\, \sq_{\len{\sq} - 1}\,.
\]
For sequences $\sq$ and $\tq$ of equal length, $\sq = \tq$ denotes that the two sequences consist of the same elements in the same order.
The empty sequence $\epsilon$ is such that $\len{\epsilon} = 0$.

The \emph{concatenation} $\cat{\sq}{\tq}$ of sequences $\sq$ and $\tq$ is the sequence of length $\len{\sq} + \len{\tq}$ obtained by juxtaposing $\sq$ and $\tq$.
We define $\cat{\sq}{\tq}$ element-wisely as
\begin{equation}
(\cat{\sq}{\tq})_k \quad=\;
\begin{cases}
\sq_k  &  0 \leq k < \len{\sq}\,,  \\
\tq_{k - \len{\sq}}  &  \len{\sq} \leq k < \len{\sq} + \len{\tq}\,,
\end{cases}
\label{eq:cat-def}
\end{equation}
for  $0 \leq k < \len{\sq} + \len{\tq}$.
For example, the concatenation of $\text{A}\,\text{B}$ and $\text{C}\,\text{D}\,\text{E}\,\text{F}$ is $\text{A}\,\text{B}\,\text{C}\,\text{D}\,\text{E}\,\text{F}$.

The \emph{reverse} $\rev{\sq}$ of a sequence $\sq$ is the sequence defined by $\len{\rev{\sq}} = \len{\sq} = N$ and
\begin{equation}
\rev{\sq}_k \ =\ \sq_{N - 1 - k}\,,  \qquad \text{for } 0 \leq k < N.
\label{eq:rev-def}
\end{equation}
For example, the reverse of sequence $\text{A}\,\text{B}\,\text{C}\,\text{D}\,\text{E}\,\text{F}$ is sequence $\text{F}\,\text{E}\,\text{D}\,\text{C}\,\text{B}\,\text{A}$.

\paragraph{Permutations and cycles.}
Rotations---introduced in the next section---are a special kind of \emph{permutations}, that is bijections of a set onto itself; henceforth, $\{0, 1, \ldots, N - 1\}$ is the set in question, which we denote $\setn{N}$.

A \emph{cycle} $(a_0\,a_1\,\cdots\, a_m)$, for $a_0, a_1, \ldots, a_m$ distinct values in $\setn{N}$, is the permutation $\lambda: \setn{N} \to \setn{N}$ such that:
\[
\lambda(a_k) = 
\begin{cases}
a_{k + 1}   & 0 \leq k < m\,, \\
a_0        & k = m\,, \\
a_k        & \text{otherwise.}
\end{cases}
\]
In other words, a cycle $(a_0\,a_1\,\cdots\, a_m)$ is a permutation that sends $a_0$ into $a_1$, $a_1$ into $a_2$, and so on until $a_m$, which it sends back to $a_0$, while leaving all other elements in $\setn{N} \setminus \{a_0, \ldots, a_m\}$ unchanged.
Two cycles $(a_0\,\cdots\, a_m)$ and $(b_0\,\cdots\, b_n)$ are disjoint if the intersection of the sets $\{a_0, \ldots, a_m\}$ and $\{b_0, \ldots, b_n\}$ is empty.
A fundamental result of the theory of permutations~\cite[Th.~1.6]{AMA} is that every permutation is uniquely expressible as the composition of disjoint cycles (the order of composition does not matter).

\section{Rotation: the problem}  \label{sec:problem-def}

The rotation $\rot{\sq}{r}$ of sequence $\sq$ by $r$ (also called ``$r$-rotation'', ``$r$-circular shift'', or ``$r$-cyclic shift'') is the sequence obtained by shifting all elements in $\sq$ by $r$ positions while wrapping over the sequence's bounds.
We assume that positive values of $r$ denote shifts to the left; hence the definition 
\begin{equation}
\rot{\sq}{r}  \quad=\;  
\begin{cases}
\underbrace{\sq_{r}\, \sq_{r + 1}\, \cdots\, \sq_{\len{\sq} - 1}}_{\len{\sq} - r\,\text{ elements}}\,
\underbrace{\sq_0\, \sq_1\, \cdots\, \sq_{r - 1}}_{r \text{ elements}}
    &  0 \leq r < \len{\sq}\,, \\
\underbrace{\sq_{\len{\sq} + r}\, \sq_{\len{\sq} + r + 1}\, \cdots\, \sq_{\len{\sq} - 1}}_{-r \text{ elements}}\,
\underbrace{\sq_0\, \sq_1\, \cdots\, \sq_{\len{\sq} + r - 1}}_{\len{\sq} - (-r) \,\text{ elements}}
    &  -\len{\sq} < r \leq 0\,.
\end{cases}
\label{eq:rot-def}
\end{equation}
For simplicity, we ignore the case of rotations by more than $\len{\sq}$ in absolute value (although it is clear they correspond to applications of definition \eqref{eq:rot-def}).
Notice that $\rot{}{0}$ is the identify mapping.
Figure~\ref{fig:rotation-example} demonstrates the definition on the running example: rotating the $6$-element sequence $\text{A}\,\text{B}\,\text{C}\,\text{D}\,\text{E}\,\text{F}$ by $2$ yields sequence $\text{C}\,\text{D}\,\text{E}\,\text{F}\,\text{A}\,\text{B}$.

\begin{figure}[!ht]
\centering
\begin{tikzpicture}
[node distance=1pt]
\lstset{basicstyle=\footnotesize}

  \node (A) [left-item] {A};
  \node (B) [left-item,right=of A] {B};
  \node (C) [right-item,right=of B] {C};
  \node (D) [right-item,right=of C] {D};
  \node (E) [right-item,right=of D] {E};
  \node (F) [right-item,right=of E] {F};

  \node (C2) [right-item,right=15mm of F] {C};
  \node (D2) [right-item,right=of C2] {D};
  \node (E2) [right-item,right=of D2] {E};
  \node (F2) [right-item,right=of E2] {F};
  \node (A2) [left-item,right=of F2] {A};
  \node (B2) [left-item,right=of A2] {B};

  \begin{scope}[-latex,thick]
  \path (F) edge[bend left] node[above] {$\rot{}{2} = \rot{}{-4}$} (C2);
  \path (C2) edge[bend left] node[below] {$\rot{}{-2} = \rot{}{4}$} (F);
  \end{scope}
\end{tikzpicture}
\caption{The rotation of sequence $\text{A}\,\text{B}\,\text{C}\,\text{D}\,\text{E}\,\text{F}$ of length $6$ by $2$ (or, equivalently, by $-(6 - 2) = -4$) yields sequence $\text{C}\,\text{D}\,\text{E}\,\text{F}\,\text{A}\,\text{B}\,$. Different colors highlight elements of the two subsequences that are swapped.}
\label{fig:rotation-example}
\end{figure}
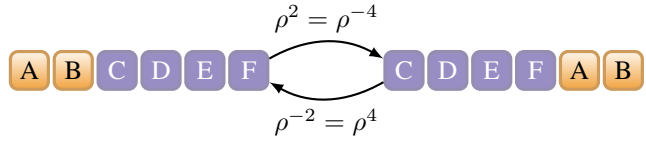

We can derive from \eqref{eq:rot-def} an equivalent element-wise representation of $\rot{\sq}{r}$ as the sequence such that $\len{\rot{\sq}{r}} = \len{\sq} = N$ and, for $0 \leq k < N$, 
\begin{equation}
(\rot{\sq}{r})_k
\quad=\;
\begin{cases}
\sq_{\wrap{(k + r)}{N}}   &   0 \leq r < N\,, \\
\sq_{\wrap{(k + N + r)}{N}}   &   - N < r \leq 0\,.
\end{cases}
\label{eq:lemma-rot}
\end{equation}
% When convenient, we remove the parentheses in the left-hand side of \eqref{eq:lemma-rot} without ambiguity, so that $\rot{\sq}{r}_k$ also denotes a generic element in the $r$-rotation of $\sq$.

The duality between left rotation and $r$ on one side, and right rotation and $N + r$ on the other side suggests the inverse mapping of \eqref{eq:lemma-rot}
\begin{equation}
\sq_k
\quad=\;
\begin{cases}
(\rot{\sq}{r})_{\wrap{(k + N - r)}{N}}   &   0 \leq r < N\,, \\
(\rot{\sq}{r})_{\wrap{(k - r)}{N}}   &   - N < r \leq 0\,.
\end{cases}
\label{eq:lemma-rot-2}
\end{equation}

\input{algos.tex}

\input{boogie.tex}

\input{dafny.tex}

\input{jml.tex}

\section{Discussion}   \label{sec:discussion}

We have seen that the four main algorithms require increasingly more complex correctness proofs---both on paper and, even more so, when mechanizing them.
This verification complexity does have relevant practical implications: the more complex the correctness argument, the more complex it is to get an implementation correct or to modify an existing implementation (for example to work on a different kind of data structure) without introducing subtle errors.
In my own experience of implementing the algorithms in Java, getting rotation by reversal right is straightforward, rotation by swapping requires more attention mostly to corner cases, rotation by modular visit can be tricky to implement correctly at the first attempt.
The likelihood of introducing errors is yet another feature to be traded off against others such as performance when choosing which algorithm to implement.

Auto-active verifiers such as Boogie generate verification conditions in multiple steps, often involving heuristics that may be hard to express and opaque to users.
We have seen several cases in the mechanized proofs, especially of the more complex algorithms, were changing seemingly irrelevant details of the input (such as the order or names of declarations) transformed a successful proof into a very slow or even nonterminating one.
Similarly but in a different dimension, changes in how Boogie (or the underlying Z3 SMT solver) encodes verification conditions and simplifies programs may break previously successful proofs, for example because the new heuristics to generate triggers do not work with the same assertions.\footnote{In fact, some of the Boogie examples in Section~\ref{sec:all-boogie-proofs} do not seem to work with more recent (with respect to when this report was first written) versions of Boogie (v.~2.2.3, October 2014) and Z3 (v.~4.3.2).}
The best defense against such brittleness is to put great care into modularizing proofs into files and procedures that are as separate as possible: provided each individual input to the verifier is sufficiently small (or large but structurally simple) the sensitivity on low-level details is kept at bay and does not interfere with the high-level goals of the prover.

By providing a higher level of abstraction and idiomatic specification features (ghost state, immutable sets and sequences, and so on), Dafny supports proofs much closer to those done on paper.
It also provides a more stable tool, where changing little details in how assertions and instructions and formulated normally does not have significant impact on the prover's behavior.
Nonetheless, some brittleness remains depending on how variables are named, or on the order in which variables and instructions  are introduced.
Deciding which definitions to expose and which to factor out and prove separately is also crucial for good performance as soon as one tackles nontrivial programs and specifications.
Thanks to its low-level nature, the preliminary experience with Boogie has been extremely useful to outline the criticalities of verifying the rotation algorithms, and to guide an efficient construction of verified implementations in Dafny or other similar provers.

\bibliographystyle{plain}
\bibliography{rotation}

\end{document}

%% file: algos.tex
\section{Rotation: the algorithms}  \label{sec:algos}

\lstset{language={PseudoEiffel}}

There are four main algorithms that compute the rotation of a sequence.
We present them in increasing level of complexity, where ``complexity'' simultaneously refers to complexity of implementation and to complexity of understanding, reasoning about, and proving the correctness of the algorithms---but not \emph{computational} complexity.
% In contrast, space complexity roughly \emph{decreases} as the other complexity increase (although a detailed computational complexity analysis is out of the scope of this paper).

We present the algorithms in pseudo-code.  Sequences are
represented by arrays of a generic type \<G>, indexed from 0.  
That is, we identify an array \<a> with the sequence $\@a[0]@\,\@a[1]@\,\ldots\,\@a[a.count - 1]@$ of its elements, where \<a.count> denotes the length of \<a>.
An array \emph{slice} \<a[low..high)> denotes the sequence $\@a[low]@\,\ldots\,\@a[high - 1]@$ if \<0 <= low <= high <= a.count>, and the empty sequence in all other cases.
Correspondingly, equality between slices of equal length corresponds to element-wise equality of their sequences of elements; this notational convention makes it possible to elide explicit quantification without ambiguity:
\[
\@x@[a..b) = \@y@[a..b)
\quad\liff\quad
\forall i: a \leq i < b \limpl \@x@[i] = \@y@[i]\,.
\]

Each algorithm
operates on an array \<a> of length \<N> (an alias of \<a.count>), and on an integer \<r>; it modifies \<a> in place so that it
is rotated by \<r> when the algorithm terminates.
For simplicity, we assume $\@0@ < \@r@ < \@N@$ as precondition.
This is without loss of generality as the \emph{right} rotation of a sequence $\sq$ by some $r$ such that $\len{\sq} > r > 0$ coincides with its left rotation by $\len{\sq} - r$, and rotation by $0$ is the identity.
We also assume that \<a>'s size does not change, so that \<N> denotes \<a.count> at any point during the computation.

Using this notation, Figure~\ref{fig:rotate-specification} shows the input/output \textbf{specification} of the rotation algorithms, where \<old a> denotes the content of \<a> upon calling \<rotate>.

\begin{figure}[!htb]
\begin{lstlisting}
rotate (a: ARRAY[G]; r: INTEGER)
   require 0 < r < N
   ensure a = $\rot{\!}{\text{r}}$(old a)
\end{lstlisting}
  \caption{Rotate array \<a> by \<r> to the left: specification.}
  \label{fig:rotate-specification}
\end{figure}

\subsection{Rotation by copy} \label{sec:algos:copy}

A straightforward application of the definition of rotation, the first algorithm (shown in Figure~\ref{fig:rotation-by-copy-algo}) uses a second array \<b> as scratch space.
With two pointers \<s> (source) and \<d> (destination), it copies each element \<a[s]> from \<a> into \<b> at its position in $\rot{\@a@}{\@r@}$.
Matching the mapping in~\eqref{eq:lemma-rot-2}, initially \<s> is \<0> and \<d> is \<N - r>; each loop iteration increments both, and resets \<d> to \<0> when it reaches \<N> (the first non-valid position in \<a>).\footnote{Alternatively, it could initialize \<s> to \<r> and \<d> to \<0> and decrement indexes following~\eqref{eq:lemma-rot}.} 
When the loop terminates, \<b> contains $\rot{\@a@}{\@r@}$, and its content is copied back into \<a>.

\begin{figure}[!htb]
\begin{lstlisting}
rotate_copy (a: ARRAY[G]; r: INTEGER)
   require 0 < r < N
   ensure a = $\rot{\!}{\text{r}}$(old a)
   local b: ARRAY[G]; s, d: INTEGER
   do
      b := [N]  // initialize to size (*N*)
      s := 0; d := N - r
      while s < N do
      invariant 0 <= s <= N
                   $\,$d = (s + N - r) $\!\wrap{}{}\!$ N
                   $\,$forall $i$: $0$ <= $i$ < s limplies a[$i$] = b[($i$ + N - r) $\!\wrap{}{}\!$ N]
         b[d] := a[s]
         s, d := s + 1, d + 1
         // wrap over (*a*)'s bounds
         if d = N then d := 0 end
      end
      // copy (*b*)'s content back into (*a*)
      a.copy (b)
   end
\end{lstlisting}
  \caption{Rotate array \<a> by \<r> to the left through copy, using \<b> as scratch space.}
  \label{fig:rotation-by-copy-algo}
\end{figure}

\begin{figure}[!ht]
\centering
\begin{tikzpicture}
[
  node distance=1pt,
  pin distance=2.5mm,
  every pin edge/.style={<-, shorten <=1pt, thick,black}
]

  \node (A) [left-item,pin=above:{\<s>},label=left:\<a:>] {A};
  \node (B) [left-item,right=of A] {B};
  \node (C) [right-item,right=of B] {C};
  \node (D) [right-item,right=of C] {D};
  \node (E) [right-item,right=of D] {E};
  \node (F) [right-item,right=of E] {F};

  \node (bA) [generic-item,below=8mm of A,label=left:\<b:>] {};
  \node (bB) [generic-item,right=of bA] {};
  \node (bC) [generic-item,right=of bB] {};
  \node (bD) [generic-item,right=of bC] {};
  \node (bE) [generic-item,right=of bD,pin=above:{\<d>}] {};
  \node (bF) [generic-item,right=of bE] {};

  \node (A2) [left-item,right=25mm of F,label=left:\<a:>] {A};
  \node (B2) [left-item,right=of A2,pin=above:{\<s>}] {B};
  \node (C2) [right-item,right=of B2] {C};
  \node (D2) [right-item,right=of C2] {D};
  \node (E2) [right-item,right=of D2] {E};
  \node (F2) [right-item,right=of E2] {F};

  \node (bA2) [generic-item,below=8mm of A2,label=left:\<b:>] {};
  \node (bB2) [generic-item,right=of bA2] {};
  \node (bC2) [generic-item,right=of bB2] {};
  \node (bD2) [generic-item,right=of bC2] {};
  \node (bE2) [left-item,right=of bD2] {A};
  \node (bF2) [generic-item,right=of bE2,pin=above:{\<d>}] {};

  \node (A3) [left-item,below=15mm of bA2,label=left:\<a:>] {A};
  \node (B3) [left-item,right=of A3] {B};
  \node (C3) [right-item,right=of B3,pin=above:{\<s>}] {C};
  \node (D3) [right-item,right=of C3] {D};
  \node (E3) [right-item,right=of D3] {E};
  \node (F3) [right-item,right=of E3] {F};

  \node (bA3) [generic-item,below=8mm of A3,pin=above:{\<d>},label=left:\<b:>] {};
  \node (bB3) [generic-item,right=of bA3] {};
  \node (bC3) [generic-item,right=of bB3] {};
  \node (bD3) [generic-item,right=of bC3] {};
  \node (bE3) [left-item,right=of bD3] {A};
  \node (bF3) [left-item,right=of bE3] {B};

  \node (A4) [left-item,below=15mm of bA,label=left:\<a:>] {A};
  \node (B4) [left-item,right=of A4] {B};
  \node (C4) [right-item,right=of B4] {C};
  \node (D4) [right-item,right=of C4] {D};
  \node (E4) [right-item,right=of D4] {E};
  \node (F4) [right-item,right=of E4] {F};
  \node (FF4) [fake-item,right=of F4,pin=above:{\<s>}] {};

  \node (bA4) [right-item,below=8mm of A4,label=left:\<b:>] {C};
  \node (bB4) [right-item,right=of bA4] {D};
  \node (bC4) [right-item,right=of bB4] {E};
  \node (bD4) [right-item,right=of bC4] {F};
  \node (bE4) [left-item,right=of bD4,pin=above:{\<d>}] {A};
  \node (bF4) [left-item,right=of bE4] {B};

  \begin{scope}[-latex,thick]
  \path ($(F.east)+(0,-7mm)$) edge ($(A2.west)+(0,-7mm)$);
  \path ($(F2.east)+(0,-7mm)$) edge[bend left] ($(F3.east)+(0,-7mm)$);
  \path ($(A3.west)+(0,-7mm)$) edge[dashed] ($(F4.east)+(0,-7mm)$) ;
  \path (A) edge[bend right=45,dotted] node[left] {$\rot{}{2}$} (bA4);
  \end{scope}
\end{tikzpicture}
\caption{Rotating sequence $\text{A}\,\text{B}\,\text{C}\,\text{D}\,\text{E}\,\text{F}$ by $2$ through copy.}
\label{fig:rotation-by-copy-example}
\end{figure}

Figure~\ref{fig:rotation-by-copy-example} demonstrates some steps of the algorithm \<rotate_copy> on an example.
The top left figure represents the state upon first entering the loop, followed by the state after one (top right) and two (bottom right) iterations; the bottom right figure is the state upon exiting the loop.

\begin{table}[!t]
\begin{tabular}{l ll |rrrr}
\multicolumn{1}{c}{\textsc{algorithm}} &  \multicolumn{1}{c}{\textsc{time}}  
& \multicolumn{1}{c}{\textsc{space}} & % \multicolumn{1}{c}{\textsc{swaps}} & 
 \multicolumn{4}{c}{\textsc{Java implementation on:}} \\
% & 
& & & {\small 3$\,${\small k}} 
      & {\small 10$\,${\small k}} 
      & {\small 100$\,${\small k}} 
      & {\small 1000$\,${\small k}} \\
\hline
\<copy>  &  $\Theta(N)$  &  $\Theta(N)$ &  % $2N\,${\small array assignments}  & 
         78$\,${\small ms}  &  540$\,${\small ms}  &  47$\,${\small s}  &  4324$\,${\small s} \\
\<copy> (native)  &  $\Theta(N + d)$  &  $\Theta(d)$ & % $2(N + d)${\small array assignments} & 
                43$\,${\small ms}  &  188$\,${\small ms}  &  12$\,${\small s}  &  1159$\,${\small s} \\
\<reverse>  &  $\Theta(N)$  &  $\Theta(1)$ &  %  $N$ & 
            44$\,${\small ms}  &  323$\,${\small ms}  &  17$\,${\small s}  &  2354$\,${\small s} \\
\<swap> (iterative)  &  $\Theta(N)$  &  $\Theta(1)$ &  % $N - \gcd(r, N - r)$ &
            42$\,${\small ms}  &  221$\,${\small ms}  &  12$\,${\small s}  &  1138$\,${\small s} \\
\<modulo>  &  $\Theta(N)$  &  $\Theta(1)$ & % $N$ swaps or $N + \gcd(N, N - r)$ array assignments & 
            48$\,${\small ms}  &  333$\,${\small ms}  &  37$\,${\small s}  &  5307$\,${\small s} \\
\end{tabular}
\caption{Time and space complexity of the various algorithms for rotating an array of size $N$ by $r$ (with $d = \min(r, N - r) \leq N/2$). The righ-hand columns show the times spent by Java 1.7 implementations of the algorithms over arrays of sizes from 3 thousand to 1000 thousand elements (for each $N$, an algorithm runs once for every $0 \leq r < N$). The experiments ran on an Intel Xeon 2.13~GHz server with 10~GB of physical RAM.}
\label{tab:complexity-overall}
\end{table}

\subsubsection{Rotation by copy: computational complexity} \label{sec:copy:complexity}

Algorithm \<rotate_copy> takes $\Theta(N)$ time and $\Theta(N)$ space.\footnote{In the complexity analyses, $N$ and $r$ denote generic values of input length \<N> and rotation \<r>.}
For large values of $N$, the space complexity may be prohibitive.
We can save some space by noticing that we only need scratch space for $d = \min(r, N - r)$ elements, while we can swap the other $N - d$ elements in place.
For example, assuming $\@r@ \leq \@N@ - \@r@$ as in the running example, rotation of \<a> by \<r> reduces to:

\begin{center}
\begin{lstlisting}[numbers=none]
b[0..r) := a[0..r)     $\qquad\ $// copy (*a[0..r)*) into (*b*)
a[0..N - r) := a[r, N) // copy (*a[r..N)*) to the left by (*r*) in place
a[N - r..N) := b[0..r) // copy (*b*) back into (*a[N - r..N)*)
\end{lstlisting}
\end{center}

This takes time $\Theta(N + d)$ and space $\Theta(d)$, with $d \leq N/2$.
If we use native memory copy methods (such as Java's \<System.arraycopy>) this is quite fast in practice but only if enough memory is available.
This is the case of the benchmarks reported in Table~\ref{tab:complexity-overall}, which ran on a server with a lot of physical RAM: the implementation of \<reverse_copy> using native memory copy methods is consistently the fastest (or very close to the fastest).

The algorithms presented in the following sections improve over the space requirements of the ``rotation by copy'' algorithm by trading time for space.

\subsubsection{Rotation by copy: correctness} \label{sec:copy:proof}

A correctness proof for \<rotate_copy> relies on a suitable ``essential'' loop invariant~\cite{FMV-CSUR14} that characterizes the state of \<b> as reflecting definition \eqref{eq:lemma-rot-2}.
For the essential invariant to be well-defined, we first need a bounding invariant that constrains the variability of index \<s> to be within \<a>'s bounds:
\begin{equation}
0 \leq \@s@ \leq \@N@ \,.
\label{eq:copy:bounding-1}
\end{equation}
Since \<d - s = N - r> initially, and both \<s> and \<d> are incremented in every iteration (while wrapping over \<N>), a corresponding bounding loop invariant about \<d> is 
\begin{equation}
\@d@ = \wrap{(\@s@ + \@N@ - \@r@)}{\@N@} \,,
\label{eq:copy:bounding-2}
\end{equation}
whose inductiveness directly follows from the definition \eqref{eq:wrap-def} of `$\wrap{}{}$' by case discussion.

The relation between \<d> and \<s> also suggests the essential invariant that relates the content of \<b> to that of \<a>:
\begin{equation}
\forall i: 0 \leq i < \@s@  \limpl \@a@[i] = \@b@[\wrap{(i + \@N@ - r)}{\@N@}]\,.
\label{eq:copy:essential}
\end{equation}
Its inductiveness is a consequence of the other invariant \eqref{eq:copy:bounding-1} and of the assignment \<b[d] := a[s]> performed in the loop.

Upon exiting the loop, \<s> equals \<N>; hence, \eqref{eq:copy:essential} asserts that
\[
\forall i: 0 \leq i < \@N@  \limpl \@a@[i] = \@b@[\wrap{(i + \@N@ - r)}{\@N@}]\,;
\]
that is, \<b> is $\rot{\@a@}{\@r@}$ according to \eqref{eq:lemma-rot-2}, which establishes the postcondition after copying \<b>'s content into \<a>.

\subsection{Rotation by reversal}  \label{sec:algos:reversal}

The rotation by reversal algorithm conjugates simplicity and efficiency in a way that makes it a very effective solution in practice.
In fact, it has been used in numerous text editors to reshuffle lines of text; Bentley reports usage as early as 1971---according to Ken Thompson, it was folklore even then~\cite[Sec.~2.3]{pearls-book}.

To rotate \<a[0..N)> to the left by \<r>, the algorithm performs three in-place reversals.
The first two reversals are partial, in that they reverse the slices \<a[0..r)> and \<a[r..N)>.
The last reversal targets the whole \<a[0..N)>.
Figure~\ref{fig:rotation-reverse-algo} shows the resulting straightforward implementation, which calls a routine \<reverse> to reverse \<a> in place.

The algorithm works thanks to a fundamental property of reversal with respect to concatenation (which we prove in Section~\ref{sec:reversal:proof} below): the reversal $\rev{(\cat{\xq}{\yq})}$ of the concatenation of two sequences $\xq$ and $\yq$ is the concatenation $\cat{\rev{\yq}}{\rev{\xq}}$ of $\yq$'s reversal and $\xq$'s reversal.
Then, consider a sequence $\sq$ of length $N$ as the concatenation $\cat{\xq}{\yq}$, where $\len{\xq} = r$ and $\len{\yq} = N - r$.
% Note that $\rot{\sq}{r}$ is $\cat{\yq}{\xq}$.
As demonstrated in Figure~\ref{fig:rotation-by-reversal-example} on the running example, where $N = 6$ and $r = 2$, the rotation by reversal algorithm applies the following transformations to $\sq$:
\[
\sq = \cat{\xq}{\yq}
\xrightarrow{\text{reverse }\xq}
\cat{\rev{\xq}}{\yq}
\xrightarrow{\text{reverse }\yq}
\cat{\rev{\xq}}{\rev{\yq}}
\xrightarrow{\text{reverse all}}
\cat{\yq}{\xq} = \rot{\sq}{r}\,,
\]
where the fundamental property justifies the last reversal of the whole $\cat{\rev{\xq}}{\rev{\yq}}
$.

\begin{figure}[!ht]
\centering
\begin{tikzpicture}
[node distance=1pt]
\lstset{basicstyle=\footnotesize}

  \node (A) [left-item] {A};
  \node (B) [left-item,right=of A] {B};
  \node (C) [right-item,right=of B] {C};
  \node (D) [right-item,right=of C] {D};
  \node (E) [right-item,right=of D] {E};
  \node (F) [right-item,right=of E] {F};

  \node (X) [label=above:{$\xq$},draw,rectangle,rounded corners,dotted,fit={(A)(B)},inner xsep=1pt,inner ysep=3pt] {};
  \node (Y) [label=above:{$\yq$},draw,rectangle,rounded corners,dotted,fit={(C)(D)(E)(F)},inner xsep=1pt,inner ysep=3pt] {};

  \node (Bx1) [left-item,right=40mm of F] {B};
  \node (Ax1) [left-item,right=of Bx1] {A};
  \node (Cx1) [right-item,right=of Ax1] {C};
  \node (Dx1) [right-item,right=of Cx1] {D};
  \node (Ex1) [right-item,right=of Dx1] {E};
  \node (Fx1) [right-item,right=of Ex1] {F};

  \node (Xx1) [label=above:{$\rev{\xq}$},draw,rectangle,rounded corners,dotted,fit={(Ax1)(Bx1)},inner xsep=1pt,inner ysep=3pt] {};
  \node (Yx1) [label=above:{$\yq$},draw,rectangle,rounded corners,dotted,fit={(Cx1)(Dx1)(Ex1)(Fx1)},inner xsep=1pt,inner ysep=3pt] {};

  \node (Bx2) [left-item,below=20mm of Bx1] {B};
  \node (Ax2) [left-item,right=of Bx2] {A};
  \node (Fx2) [right-item,right=of Ax2] {F};
  \node (Ex2) [right-item,right=of Fx2] {E};
  \node (Dx2) [right-item,right=of Ex2] {D};
  \node (Cx2) [right-item,right=of Dx2] {C};

  \node (Xx2) [label=above:{$\rev{\xq}$},draw,rectangle,rounded corners,dotted,fit={(Ax2)(Bx2)},inner xsep=1pt,inner ysep=3pt] {};
  \node (Yx2) [label=above:{$\rev{\yq}$},draw,rectangle,rounded corners,dotted,fit={(Cx2)(Dx2)(Ex2)(Fx2)},inner xsep=1pt,inner ysep=3pt] {};

  \node (B2) [left-item,left=40mm of Bx2] {B};
  \node (A2) [left-item,left=of B2] {A};
  \node (F2) [right-item,left=of A2] {F};
  \node (E2) [right-item,left=of F2] {E};
  \node (D2) [right-item,left=of E2] {D};
  \node (C2) [right-item,left=of D2] {C};

  \node (X2) [label=above:{$\xq$},draw,rectangle,rounded corners,dotted,fit={(A2)(B2)},inner xsep=1pt,inner ysep=3pt] {};
  \node (Y2) [label=above:{$\yq$},draw,rectangle,rounded corners,dotted,fit={(C2)(D2)(E2)(F2)},inner xsep=1pt,inner ysep=3pt] {};

  \begin{scope}[-latex,thick]
  \path (Y.east) edge node[above] {reverse $\xq$} (Xx1.west);
  \path ($(Fx1.south)+(0,-1mm)$) edge node[left] {reverse $\yq$} ($(Cx2.north)+(0,1mm)$);
  \path (Xx2.west) edge node[above] {reverse $\cat{\rev{\xq}}{\rev{\yq}}$} (X2.east);
  \path (X.west) edge[dotted,bend right] node[left] {$\rot{}{2}$} (Y2.west);
  \end{scope}
\end{tikzpicture}
\caption{Rotating sequence $\text{A}\,\text{B}\,\text{C}\,\text{D}\,\text{E}\,\text{F}$ by $2$ through three reversals.}
\label{fig:rotation-by-reversal-example}
\end{figure}

Completing the picture, Figure~\ref{fig:reversal-algo} provides an implementation of \<reverse> that works by switching elements at opposite ends of \<a[low..high)> while working its way inward: each iteration of the main loop swaps \<a[p]> and \<a[q]> on line~\ref{cl:reverse:swap}, and then increments \<p> and decrements \<q> (initialized to \<low> and \<high - 1>) on line~\ref{cl:reverse:increment}.

\begin{figure}[!htb]
\begin{lstlisting}
rotate_reverse (a: ARRAY[G]; r: INTEGER)
   require 0 < r < N
   ensure a = $\rot{\!}{\text{r}}$(old a)
   do
      reverse (a, 0, r)  (*\label{cl:reverse-call-1}*)
      reverse (a, r, N)  (*\label{cl:reverse-call-2}*)
      reverse (a, 0, N)  (*\label{cl:reverse-call-3}*)
   end
\end{lstlisting}
  \caption{Rotate array \<a> by \<r> to the left through three reversals.}
  \label{fig:rotation-reverse-algo}
\end{figure}

\begin{figure}[!hbt]
\begin{lstlisting}
reverse (a: ARRAY[G]; low, high: INTEGER)
   require 0 <= low <= high <= N
   ensure a[low..high) = (old a[low..high))$^{-1}$
   local p, q: INTEGER
   do
      p, q := low, high - 1   (*\label{cl:reverse:init}*)
      while p < q + 1  (*\label{cl:reverse:loopcondition}*)
      invariant 
            low <= p <= q + 2 <= high + 1
            q = high + low - 1 - p
            forall $i$: low <= $i$ < p limplies (old a)[$i$] = a[high + low - 1 - $i$]
            forall $i$: q < $i$ < high limplies (old a)[$i$] = a[high + low - 1 - $i$]
      do
         // swap (*a[p]*) and (*a[q]*)
         a[p], a[q] := a[q], a[p]  (*\label{cl:reverse:swap}*)
         p, q := p + 1, q - 1      (*\label{cl:reverse:increment}*)
      end
   end
\end{lstlisting}
  \caption{In-place reversal of \<a[low..high)> by swapping elements at opposite ends.}
  \label{fig:reversal-algo}
\end{figure}

\subsubsection{Rotation by reversal: computational complexity} \label{sec:reversal:complexity}

Since \<rotate_reverse> just calls \<reverse> three times, the asymptotic complexities of the two algorithms are the same.
The implementation of \<reverse> shown in Figure~\ref{fig:reversal-algo} has space complexity $\Theta(1)$ (since it only needs one variable to swap) and time complexity $\Theta(\@high@ - \@low@)$.
Hence, \<rotate_reverse> has time complexity $\Theta(N)$ and space complexity $\Theta(1)$.
In terms of swaps of array elements, \<reverse> performs $r/2 + (N - r)/2 + N/2 = N$ of them.
Table~\ref{tab:complexity-overall} shows that rotation by reversal also scales gracefully in practice and, while it is not the fastest overall, it normally is in the ballpark of the fastest.

Another appealing feature of rotation by reversal is its flexibility with respect to the data structure it operates on.
As long as we can implement in-place reversal in constant space and linear time on the structure, \<rotate_reverse> will still work with the same complexity.
In particular, we can have linear-time in-place reversal on linked lists; rotation by reversal works there as well as it works on arrays.

\subsubsection{Rotation by reversal: correctness} \label{sec:reversal:proof}

We sketched a correctness argument for \<rotate_reverse> in Section~\ref{sec:algos:reversal}; now we provide a complete proof, beginning with loop invariants sufficient to verify \<reverse>.

\paragraph{Reversal: correctness.}
A basic bounding invariant requires the indexes \<p> and \<q> to be valid positions within \<a[low..high)>:
\begin{equation}
\@low@ \leq \@p@ \leq \@q@ + 2 \leq \@high@ + 1\,.
\label{eq:reverse:bounding-1}
\end{equation}
The loop exits when \<p = q + 1> if \<high - low> is an even number; and when \<p = q + 2> if \<high - low> is an odd number.
This reveals that there is a bit of redundancy in \<reverse>: when \<high - low> is odd, $\@p@ = \@q@ = \lfloor (\@high@ - \@low@) / 2 \rfloor + 1$ at the beginning of the last loop iteration, which consequently swaps \<a[low..high)>'s central element with itself.
To avoid this unnecessary swap, relax the loop staying condition (line~\ref{cl:reverse:loopcondition}) to $\@p@ < \@q@$.
Here, however, we prefer the slightly redundant formulation because it makes for simpler loop invariants and correctness arguments, as we do not have to separately discuss what happens to the central element.

Another bounding invariant relates \<q> to \<p>:
\begin{equation}
\@q@ = \@high@ + \@low@ - 1 - \@p@ \,,
\label{eq:reverse:bounding-2}
\end{equation}
which implies, together with \eqref{eq:reverse:bounding-1}, that \<q> is also within bounds in the loop.

The essential loop invariant is two-fold, as it has to relate elements in the upper half \<a(q..high)> to the corresponding elements in the lower half of \<old a> that have been swapped; and vice versa for the lower half \<a[low..p)> with respect to the upper half of \<old a>:
\begin{equation}
\begin{split}
\forall i: \@low@ \leq i < \@p@  &\quad \limpl\quad  (\@old a@)[i] = \@a@[\@high@ + \@low@ - 1 - i] \,, \\
\forall i: \@q@ < i < \@high@  &\quad \limpl\quad  (\@old a@)[i] = \@a@[\@high@ + \@low@ - 1 - i] \,.
\end{split}
\label{eq:reverse:essential}
\end{equation}

Initiation and consecution are trivial to prove for the bounding invariants, based on how \<p> and \<q> are initialized (line~\ref{cl:reverse:init}) and modified by every loop iteration (line~\ref{cl:reverse:increment}).
The bounding invariants are also the basis to prove inductiveness of the essential invariant: each iteration swaps the elements at positions \<p> and \<q>, thus preserving \eqref{eq:reverse:essential} thanks to \eqref{eq:reverse:bounding-2}.
Finally, one can check that \eqref{eq:reverse:essential} implies \<reverse>'s postcondition when the loop exits with $\@p@ \geq \@q@ + 1$.

\paragraph{Rotation by reversal: correctness.}
We prove the lemma relating reversal and concatenation that underpins the correctness of \<rotate_reverse>.

\begin{lemma}[Reverse of concatenation] \label{lm:rev-cat}
$\rev{(\cat{\sq}{\tq})} = \cat{(\rev{\tq})}{(\rev{\sq})}$, for any two sequences $\sq$ and $\tq$.
\end{lemma}
\begin{proof}
Let $0 \leq k < \len{\sq} + \len{\tq}$ be a generic position in $\rev{(\cat{\sq}{\tq})}$.
By \eqref{eq:rev-def}, $\rev{(\cat{\sq}{\tq})}_k$ equals $(\cat{\sq}{\tq})_{k'}$, where $k' = \len{\sq} + \len{\tq} - 1 - k$.
We discuss two cases.
First case: (a) $0 \leq k < \len{\tq}$, and hence $\len{\sq} \leq k' < \len{\sq} + \len{\tq}$.
Thus $(\cat{\sq}{\tq})_{k'}$ equals $\tq_{k' - \len{\sq}}$ by \eqref{eq:cat-def} and $k' - \len{\sq} = \len{\tq} - 1 - k$; hence $\tq_{k' - \len{\sq}}$ is the element at position $k$ in $\rev{\tq}$ according to \eqref{eq:rev-def}.
Otherwise, second case: (b) $\len{\tq} \leq k < \len{\sq} + \len{\tq}$, and hence $0 \leq k' < \len{\sq}$.
Thus $(\cat{\sq}{\tq})_{k'}$ equals $\sq_{k'}$ by \eqref{eq:cat-def}; hence $\sq_{k'}$ is the element at position $k - \len{\tq}$ in $\rev{\sq}$ according to \eqref{eq:rev-def}.
(a) and (b) show that $\rev{(\cat{\sq}{\tq})}$ follows the definition of $\cat{(\rev{\tq})}{(\rev{\sq})}$.
\qedhere
\end{proof}

Let $0 \leq x < \@r@$ be an index in the lower half of \<a[0..N)>.
The first reversal of \<rotate_reverse> maps $x$ to $\@r@ - 1 - x$ according to \eqref{eq:rev-def}; the second reversal leaves it unchanged; the third reversal maps it to $\@N@ - 1 - (\@r@ - 1 - x) = x + \@N@ - \@r@$ still according to \eqref{eq:rev-def} and to Lemma~\ref{lm:rev-cat}.
Following~\eqref{eq:lemma-rot-2}, the latter is the position in \<old a>'s rotation of the element originally at $x$ (note that $0 \leq x < \@r@$ implies $x + \@N@ - \@r@ < \@N@$).
The dual argument, for $\@r@ \leq y < \@N@$, establishes that the element originally at $y$ ends up in the position in \<old a>'s rotation.
Since such generic $x$ and $y$ span the whole interval \<[0..N)>, \<rotate_reverse>'s postcondition holds.

\subsection{Rotation by swapping}  \label{sec:algos:swap}

The rotation by swapping algorithm applies a divide-and-conquer strategy to improve over the space requirements of the rotation by copy algorithm (Section~\ref{sec:algos:copy}).

The algorithm builds upon two key observations.
First, as it is apparent from definition \eqref{eq:rot-def} of rotation, rotating \<a[0..N)> to the left by \<r> can be seen as swapping the adjacent array slices \<a[0..r)> and \<a[r..N)>---which have different length in general.
Second, swapping two non-overlapping array slices of \emph{equal} length can be done in place in linear time, as shown in Figure~\ref{fig:swap-algo}: we simply maintain two index variables \<x> and \<z> pointing to the corresponding elements in each section, and swap the corresponding element elements in each iteration.

\begin{figure}[!htb]
\begin{lstlisting}
// swap (*a[low..low $+$ d)*) and (*a[high $-$ d..high)*)
swap_sections (a: ARRAY[G]; low, high: INTEGER; d: INTEGER)
                        // non overlapping slices
   require 0 <= low <= low + d <= high - d <= high <= N
   ensure  a[low..low + d) = old a[high - d..high)
            $\,$a[low + d..high - d) = old a[low + d..high - d)
            $\,$a[high - d..high) = old a[low..low + d)
                    // pointers to left (x) and right (z) slices
   local x, z: INTEGER  
   do
      x, z := low, high - d
      until x = low + d
      invariant
         low <= x <= low + d  $\ $land high - d <= z <= high
         x - low = z - (high - d)
         a[low..x) = (old a)[high - d..z)
         a[x..high - d) = (old a)[x..high - d)
         a[high - d..z) = (old a)[low..x)
         a[z..high) = (old a)[z..high)
      do
         // swap (*a[x]*) and (*a[z]*)
         a[x], a[z] := a[z], a[x]   (*\label{cl:swap:swap}*)
         x, z := x + 1, z + 1
      end
   end
\end{lstlisting}
  \caption{In-place swap of \<a[low..low + d)> and \<a[low..low + d)>.}
  \label{fig:swap-algo}
\end{figure}

The divide-and-conquer strategy implemented by the rotation by swapping algorithm calls \<swap_sections> to compute part of the rotation, and then repeats on the smaller unrotated section until completion.
To illustrate, consider the running example in Figure~\ref{fig:rotation-by-swap-example}, where the goal is to swap the subsequence denoted by $\xq$ with the rest.
Since the size $r = 2$ of $\xq$ is less than $N - r = 4$, we can select another subsequence of size $r$ (denoted by $\zq$ in Figure~\ref{fig:rotation-by-swap-example}), at the other end of the whole sequence, such that it does not overlap $\xq$.
After swapping $\xq$ and $\zq$ by calling \<swap_sections>, $\xq$ acquires its final position in the rotation of the whole sequence.
Then, we recursively apply the algorithm to the subsequence $\cat{\zq}{\yq}$, which we rotate also by $r$.
In this case, the two subsequences $\zq$ and $\yq$ have equal length; hence swapping them concludes the overall rotation.

\begin{figure}[!htb]
\centering
\begin{tikzpicture}
[node distance=3pt]
\lstset{basicstyle=\footnotesize}

  \node (A) [left-item] {A};
  \node (B) [left-item,right=of A] {B};
  \node (C) [right-item,right=of B] {C};
  \node (D) [right-item,right=of C] {D};
  \node (E) [right-item,right=of D] {E};
  \node (F) [right-item,right=of E] {F};

  \node (E2) [right-item,right=30 mm of F] {E};
  \node (F2) [right-item,right=of E2] {F};
  \node (C2) [right-item,right=of F2] {C};
  \node (D2) [right-item,right=of C2] {D};
  \node (A2) [left-item,right=of D2] {A};
  \node (B2) [left-item,right=of A2] {B};

  \node (C3) [right-item,below=20 mm of E2] {C};
  \node (D3) [right-item,right=of C3] {D};
  \node (E3) [right-item,right=of D3] {E};
  \node (F3) [right-item,right=of E3] {F};
  \node (A3) [left-item,right=of F3] {A};
  \node (B3) [left-item,right=of A3] {B};

  \node (X) [label=above:{$\xq$},draw,rectangle,rounded corners,dotted,fit={(A)(B)},inner xsep=1pt,inner ysep=3pt] {};
  \node (Y) [label=above:{$\yq$},draw,rectangle,rounded corners,dotted,fit={(C)(D)},inner xsep=1pt,inner ysep=3pt] {};
  \node (Z) [label=above:{$\zq$},draw,rectangle,rounded corners,dotted,fit={(E)(F)},inner xsep=1pt,inner ysep=3pt] {};

  \node (X2) [label=above:{$\xq$},draw,rectangle,rounded corners,dotted,fit={(A2)(B2)},inner xsep=1pt,inner ysep=3pt] {};
  \node (Y2) [label=above:{$\yq$},draw,rectangle,rounded corners,dotted,fit={(C2)(D2)},inner xsep=1pt,inner ysep=3pt] {};
  \node (Z2) [label=above:{$\zq$},draw,rectangle,rounded corners,dotted,fit={(E2)(F2)},inner xsep=1pt,inner ysep=3pt] {};

  \node (X3) [label=above:{$\xq$},draw,rectangle,rounded corners,dotted,fit={(A3)(B3)},inner xsep=1pt,inner ysep=3pt] {};
  \node (Y3) [label=above:{$\yq$},draw,rectangle,rounded corners,dotted,fit={(C3)(D3)},inner xsep=1pt,inner ysep=3pt] {};
  \node (Z3) [label=above:{$\zq$},draw,rectangle,rounded corners,dotted,fit={(E3)(F3)},inner xsep=1pt,inner ysep=3pt] {};

  \begin{scope}[-latex,thick]
  \path (Z) edge node[above] {swap $\xq$ and $\zq$} (Z2);
  \path (Z2.south east) edge[dashed] node[right] {$\rot{}{2}(\zq\,\yq)$} (Z3.north west);
  \path ($(Y.south)+(0,-1mm)$) edge[dotted,bend right=20] node[below] {$\rot{}{2}$} ($(Y3.west)+(-1mm,0)$);
  \end{scope}

\end{tikzpicture}
\caption{Rotating sequence $\text{A}\,\text{B}\,\text{C}\,\text{D}\,\text{E}\,\text{F}$ by $2$ through swapping sections.}
\label{fig:rotation-by-swap-example}
\end{figure}

It is natural to generalize this approach using a recursive formulation.
As Figure~\ref{fig:rotation-swap-algo} shows, we rely on a helper procedure \<rotate_swap_helper> that swaps the slices \<a[low..p)> and \<a[p..high)>.
If the two slices have equal length (case on line~\ref{cl:rotswap:case-equal}), then calling \<swap_sections> with \<d = p - low = high - p> does the job.
Otherwise, suppose the first slice is smaller (case on line~\ref{cl:rotswap:case-left}, such as in the running example of Figure~\ref{fig:rotation-by-swap-example}, where \<p = r = 2>); that is, \<p - low \< high - p>.
Then, swap \<a[low..p)> with \<a[high - (p - low)..high)>; as a result, the latter slice is in place, and we repeat on \<a[low..high - (p - low))>.
Conversely, if \<p - low \> high - p> (case on line~\ref{cl:rotswap:case-right}, such as in the other example of Figure~\ref{fig:rotation-by-swap-example-right}, where \<p = r = 4>), swap \<a[low..low + (high - p))> with \<a[high - p..high)>; as a result, the former slice is in place, and we repeat on \<a[low + (high - p)..high)>.

\begin{figure}[!htb]
\begin{lstlisting}
rotate_swap (a: ARRAY[G]; r: INTEGER)
   require 0 < r < N
   ensure a = $\rot{\!}{\text{r}}$(old a)
   do rotate_swap_helper (a, 0, r, N) end

// Rotate (*a[low..high)*) at (*p*) by swapping (*a[low..p)*) and (*a[p..high)*)
rotate_swap_helper (a: ARRAY[G]; low, p, high: INTEGER)
   require 0 <= low <= p < high <= N
   ensure
      a[low..high) = $\rot{\!}{(\text{p} - \text{low})}$(old a)[low..high)
      a[0..low) = (old a)[0..low)  land$\:$a[high..N) = (old a)[high..N)
   do
      if low < p < high then
         if p - low = high - p then     (*\label{cl:rotswap:case-equal}*)
            // swap (*a[low..p)*) and (*a[p..high)*)
            swap_sections (a, low, high, p - low)
            // now the whole (*a[low..high)*) is in place
         elseif p - low < high - p then (*\label{cl:rotswap:case-left}*)
            // swap (*a[low..p)*) and   (*a[high $-$ (p $-$ low)..high)*)
            swap_sections (a, low, high, p - low)
            // now (*a[high $-$ (p $-$ low)..high)*) is in place
            rotate_swap_helper (a, low, p, high - (p - low))
         elseif p - low > high - p then (*\label{cl:rotswap:case-right}*)
            // swap (*a[low..low $+$ (high $-$ p))*) and (*a[p..high)*)
            swap_sections (a, low, high, high - p)
            // now (*a[low..low $+$ (high $-$ p))*) is in place
            rotate_swap_helper (a, low + (high - p), p, high)
         end
      end
   end
\end{lstlisting}
  \caption{Rotate array \<a> by \<r> to the left by swapping sections of equal length: recursive algorithm.}
  \label{fig:rotation-swap-algo}
\end{figure}

\begin{figure}[!htb]
\centering
\begin{tikzpicture}
[node distance=3pt]
\lstset{basicstyle=\footnotesize}

  \node (A) [left-item] {L};
  \node (B) [left-item,right=of A] {M};
  \node (C) [left-item,right=of B] {N};
  \node (D) [left-item,right=of C] {O};
  \node (E) [right-item,right=of D] {P};
  \node (F) [right-item,right=of E] {Q};

  \node (E2) [right-item,right=30 mm of F] {P};
  \node (F2) [right-item,right=of E2] {Q};
  \node (C2) [left-item,right=of F2] {N};
  \node (D2) [left-item,right=of C2] {O};
  \node (A2) [left-item,right=of D2] {L};
  \node (B2) [left-item,right=of A2] {M};

  \node (E3) [right-item,below=20 mm of E2] {P};
  \node (F3) [right-item,right=of E3] {Q};
  \node (A3) [left-item,right=of F3] {L};
  \node (B3) [left-item,right=of A3] {M};
  \node (C3) [left-item,right=of B3] {N};
  \node (D3) [left-item,right=of C3] {O};

  \node (X) [label=above:{$\xq$},draw,rectangle,rounded corners,dotted,fit={(A)(B)},inner xsep=1pt,inner ysep=3pt] {};
  \node (Y) [label=above:{$\yq$},draw,rectangle,rounded corners,dotted,fit={(C)(D)},inner xsep=1pt,inner ysep=3pt] {};
  \node (Z) [label=above:{$\zq$},draw,rectangle,rounded corners,dotted,fit={(E)(F)},inner xsep=1pt,inner ysep=3pt] {};

  \node (X2) [label=above:{$\xq$},draw,rectangle,rounded corners,dotted,fit={(A2)(B2)},inner xsep=1pt,inner ysep=3pt] {};
  \node (Y2) [label=above:{$\yq$},draw,rectangle,rounded corners,dotted,fit={(C2)(D2)},inner xsep=1pt,inner ysep=3pt] {};
  \node (Z2) [label=above:{$\zq$},draw,rectangle,rounded corners,dotted,fit={(E2)(F2)},inner xsep=1pt,inner ysep=3pt] {};

  \node (X3) [label=above:{$\xq$},draw,rectangle,rounded corners,dotted,fit={(A3)(B3)},inner xsep=1pt,inner ysep=3pt] {};
  \node (Y3) [label=above:{$\yq$},draw,rectangle,rounded corners,dotted,fit={(C3)(D3)},inner xsep=1pt,inner ysep=3pt] {};
  \node (Z3) [label=above:{$\zq$},draw,rectangle,rounded corners,dotted,fit={(E3)(F3)},inner xsep=1pt,inner ysep=3pt] {};

  \begin{scope}[-latex,thick]
  \path (Z) edge node[above] {swap $\xq$ and $\zq$} (Z2);
  \path (Y2.south east) edge[dashed] node[right] {$\rot{}{2}(\yq\,\xq)$} (Y3.north west);
  \path ($(Y.south)+(0,-1mm)$) edge[dotted,bend right=20] node[below] {$\rot{}{4}$} ($(Z3.west)+(-1mm,0)$);
  \end{scope}

\end{tikzpicture}
\caption{Rotating sequence $\text{L}\,\text{M}\,\text{N}\,\text{O}\,\text{P}\,\text{Q}$ by $4$ through swapping sections.}
\label{fig:rotation-by-swap-example-right}
\end{figure}

In the remainder, we refer to the two recursive cases as the ``left is smaller'' case, for \<p - low \< high - p>, and the ``right is smaller'' case, for \<p - low \> high - p>.
As we justify rigorously in Section~\ref{sec:swap:proof}, the correctness of the algorithm relies on two dual properties of rotation with respect to concatenation, one for each recursive case.
In both cases, we represent \<a[low..high)> as the concatenation $\cat{\cat{\xq}{\yq}}{\zq}$ of three sequences.
In the \emph{left is smaller} case, $\len{\xq} = \len{\zq} = \@p@ - \@low@$, and the property that $\rot{(\cat{\cat{\xq}{\yq}}{\zq})}{\len{\xq}} = \cat{\rot{(\cat{\zq}{\yq})}{\len{\zq}}}{\xq}$ justifies the recursive call.
In the \emph{right is smaller} case, $\len{\xq} = \len{\zq} = \@high@ - \@p@$ (hence $\len{\xq} + \len{\yq} = \@p@ - \@low@$), and the property that $\rot{(\cat{\cat{\xq}{\yq}}{\zq})}{\len{\xq} + \len{\yq}} = \cat{\zq}{\rot{(\cat{\yq}{\xq})}{\len{\yq}}}$ justifies the recursive call.

\subsubsection{Rotation by swapping: computational complexity} \label{sec:swap:complexity}

Here is a back-of-the-envelope complexity analysis of \<rotate_swap> via its helper function.
Overall, \<rotate_swap_helper> makes some $n$ recursive calls to \<swap_sections>; let $d_k$ denote the value of argument \<d> in the $k$th call, for $1 \leq k \leq n$ (for example, $d_1 = \min(\@r@, \@N - r@)$).
Every such call to \<swap_sections> takes time $\Theta(\@d@)$ and reduces the problem size by \<d>.
Since recursion terminates when the yet-to-be-rotated array slice becomes empty, it must be $d_1 + \cdots + d_n = N$.
The overall time complexity is then $\Theta(d_1) + \cdots + \Theta(d_n) = \Theta(N)$.

Gries and Mills~\cite[Sec.~5]{swapping-TR} provide a more rigorous analysis of the complexity of \<rotate_swap> in terms of number of swaps between array elements.
First, note the elegant property that \<rotate_swap_helper> reduces to Euclid's algorithm for greatest common divisor by successive subtractions~\cite[Sec.~1.3]{FMV-CSUR14} if we omit the calls to \<swap_sections>: it computes $\gcd(\@r@, \@N@ - \@r@)$.
Hence, the last call to \<swap_sections> takes place when $\@p@ - \@low@ = \@high@ - \@p@ = \gcd(\@r@, \@N@ - \@r@)$; it places the remaining $2\cdot\gcd(\@r@, \@N@ - \@r@)$ elements in their final rotated position through exactly $\gcd(\@r@, \@N@ - \@r@)$ swaps.
The previous calls to \<swap_sections> perform another $\@N@ - 2\cdot\gcd(\@r@, \@N@ - \@r@)$ swaps: each swap places one element in its final rotated position.
Overall \<rotate_swap> performs $\@N@ - \gcd(\@r@, \@N@ - \@r@)$ swaps.

The space complexity is $\Theta(N)$ due to recursion: the worst case is $\@r@ = 1$, when the maximum recursion depth is \<N>.
However, it is straightforward to produce an equivalent iterative version of \<rotate_swap>, as shown in Figure~\ref{fig:rotation-swap-algo-iter}.
The condition \<lnot(low \< p \< high)> that terminates recursion becomes the exit condition for a loop (line~\ref{cl:rswap:iterative:loop} in Figure~\ref{fig:rotation-swap-algo-iter}) that calls \<swap_sections> and moves \<low> or \<high> closer to each other accordingly.
The iterative version clearly has space complexity $\Theta(1)$.
Practical implementations will use iteration even if enough memory is available, since limits on recursion stack size would become a bottleneck.
As Table~\ref{tab:complexity-overall} shows, such iterative version is quite fast in practice, often nearly as fast as rotation by copy using native methods, but with only constant memory usage.

\subsubsection{Rotation by swapping: correctness} \label{sec:swap:proof}

We first quickly illustrate the invariants for a correctness proof of \<swap_sections>.
Then, we discuss the key steps of a correctness proof for \<rotate_swap> in its recursive and iterative versions.

\paragraph{Swap sections: correctness.}
Variables \<x> and \<z> span the intervals \<[low..low + d)> and \<[high - d.. high)>; hence the bounding invariant 
\begin{equation}
\begin{split}
\@low@ &\leq \@x@ \leq \@low@ + \@d@\,, \\
\@high@ - \@d@ &\leq \@z@ \leq \@high@\,.
\end{split}
\label{eq:swap-bounding-1}
\end{equation}
At the beginning of every loop iteration, they point to the pair of elements that are about to be swapped; hence the other bounding invariant 
\begin{equation}
\@x@ - \@low@ = \@z@ - (\@high@ - \@d@)\,.
\label{eq:swap-bounding-2}
\end{equation}

Based on the bounding invariants \eqref{eq:swap-bounding-1} and \eqref{eq:swap-bounding-2}, we characterize the content of \<a> during \<swap_sections>'s execution as partitioned into six sections:
\begin{center}
\begin{tikzpicture}
[node distance=-0.5pt]

\node (a) {\<a>$\colon$};
\node (left-unchanged) [array-slice,right=2mm of a] {untouched};
\node (left-swapped) [array-slice,right=of left-unchanged] {region X:\\swapped with (Z)};
\node (mid-unchanged) [array-slice,right=of left-swapped] {unchanged};
\node (right-swapped) [array-slice,right=of mid-unchanged] {region Z:\\swapped with (X)};
\node (right-unchanged) [array-slice,right=of right-swapped] {unchanged};
\node (rightmost-unchanged) [array-slice,right=of right-unchanged] {untouched};

\node [xshift=2pt,yshift=6pt] at (left-unchanged.north west) {$0$};
\node [xshift=7pt,yshift=6pt] at (left-swapped.north west) {\<low>};
\node [xshift=3pt,yshift=6pt] at (mid-unchanged.north west) {\<x>};
\node [xshift=20pt,yshift=6pt] at (right-swapped.north west) {\<high - d>};
\node [xshift=3pt,yshift=6pt] at (right-unchanged.north west) {\<z>};
\node [xshift=10pt,yshift=6pt] at (rightmost-unchanged.north west) {\<high>};
\node [xshift=3pt,yshift=6pt] at (rightmost-unchanged.north east) {\<N>};
\end{tikzpicture}
\end{center}
Thus, we have the essential invariants:
\begin{equation}
\begin{split}
\@a@[\@low@..\@x@)  & =  (\@old a@)[\@high@ - \@d@..\@z@)\,,  \\
\@a@[\@x@..\@high@ - \@d@)  & =  (\@old a@)[\@x@..\@high@ - \@d@)\,,  \\
\@a@[\@high@ - \@d@..\@z@)  & =  (\@old a@)[\@low@..\@x@)\,,  \\
\@a@[\@z@..\@high@)  & =  (\@old a@)[\@z@..\@high@)\,.  \\
\end{split}
\label{eq:swap-essential}
\end{equation}
It is not difficult to prove initiation and consecution of \eqref{eq:swap-essential}.
In particular, swapping the elements \<a[x]> and \<a[z]> on line~\ref{cl:swap:swap} in Figure~\ref{fig:swap-algo} maintains invariance of the swapped slices \<a[low..x)> and \<a[high - d..z)>.

\paragraph{Rotation by swapping: correctness of recursive version.}
As mentioned in the overview, the proof makes usage of a fundamental lemma, which we now prove.

\begin{lemma}[Rotation and swap] \label{lm:rot-swap}
For any three sequences $\xq$, $\yq$, $\zq$, with $\len{\xq} = \len{\zq} = d$ and $\len{\xq} + \len{\yq} + \len{\zq} = N$:
\begin{subequations}
\label{eq:lemma-rot-swap-both}
\begin{align}
\rot{(\cat{\cat{\xq}{\yq}}{\zq})}{d} &\ = \cat{\rot{(\cat{\zq}{\yq})}{d}}{\xq}\,,
\label{eq:lemma-rot-swap-left}
\\
\rot{(\cat{\cat{\xq}{\yq}}{\zq})}{N - d} &\ = \cat{\zq}{\rot{(\cat{\yq}{\xq})}{N - 2d}}\,.
\label{eq:lemma-rot-swap-right}
\end{align}
\end{subequations}
\end{lemma}
\begin{proof}
We prove \eqref{eq:lemma-rot-swap-left}; the proof of \eqref{eq:lemma-rot-swap-right} can be constructed by similar means.
Let $0 \leq k < N$ be a generic position in $\cat{\cat{\xq}{\yq}}{\zq}$; the goal is showing that the mapping $k \mapsto_1 k_1$ determined by the left-hand side transformation $\rot{(\cat{\cat{\xq}{\yq}}{\zq})}{d}$ and the mapping $k \mapsto_2 k_2$ determined by the right-hand side transformation $\cat{\rot{(\cat{\zq}{\yq})}{d}}{\xq}$ are such that $k_1 = k_2$.
We discuss two cases: (a) $0 \leq k < d$ and (b) $d \leq k < N$.

In case (a), $k + N - d < N$; hence $k_1$ is $k + N - d$ according to \eqref{eq:lemma-rot-2}.
Also in case (a), $k$ denotes a position of $\xq$; hence, mapping $\mapsto_2$ shifts $k$ by $\len{\zq} + \len{\yq}$ to the right; that is, $k_2 = k + (N - 2d) + d = k_1$, which concludes case (a).

In case (b), $N \leq k + N - d < 2N$; hence $k_1$ is $(k + N - d) - N = k - d$ according to \eqref{eq:lemma-rot-2} (and the definition \eqref{eq:wrap-def} of `$\wrap{}{}$').
To determine the value of $k_2$, we describe $\mapsto_2$ as the application of $\mapsto_{2.1}$ followed by $\mapsto_{2.2}$: $\mapsto_{2.1}$ accounts for the swapping of $\zq$ and $\xq$ in $\cat{\cat{\xq}{\yq}}{\zq}$, and $\mapsto_{2.2}$ accounts for the rotation $\rot{(\cat{\zq}{\yq})}{d}$.
Accordingly, we further split case (b) into: (b.1) $d \leq k < N - d$ and (b.2) $N - d \leq k < N$.
In case (b.1), $k$ denotes a position within $\yq$ in sequence $\cat{\cat{\xq}{\yq}}{\zq}$.
Hence, mapping $\mapsto_{2.1}$ leaves $k$ unchanged (since $\len{\zq} = \len{\xq}$), and then $\mapsto_{2.2}$ maps it to $\wrap{(k + \len{\cat{\zq}{\yq}} - d)}{\len{\cat{\zq}{\yq}}} = k - d$ according to \eqref{eq:lemma-rot-2} (since $\len{\cat{\zq}{\yq}} = N - d$ and $N - d \leq k + (N - d) - d < 2(N - d)$ in this case).
In case (b.2), $k$ denotes a position within $\zq$ in sequence $\cat{\cat{\xq}{\yq}}{\zq}$.
Hence, mapping $\mapsto_{2.1}$ shifts $k$ by $(N - d)$ to the left (i.e., $k \mapsto_{2.1} k - (N - d)$), and then $\mapsto_{2.2}$ maps it to $(k - (N - d)) + (N - 2d) = k - d$ according to \eqref{eq:rot-def} (since $\len{\yq} = N - 2d$).
This concludes the proof that $k_1 = k_2$ in all cases.
\qedhere
\end{proof}

The proof of \<rotate_swap_helper> now discusses the three main cases, for the three conditional branches on lines~\ref{cl:rotswap:case-equal}, \ref{cl:rotswap:case-left}, and \ref{cl:rotswap:case-right} in Figure~\ref{fig:rotation-swap-algo}.
In the first case, \<p - low = high - p>, note that the postcondition of \<swap_sections> called on\linebreak \<a[low..high)> with \<d = p - low> satisfies definition \eqref{eq:cat-def} of rotation for $r = \@p@ - \@low@$ and $N = \@high@ - \@low@$.
In the \emph{left is smaller} case, \<p - low \< high - p>, after the call to \<swap_sections> the content of \<a[low..high)> is the concatenation
\[
\underbrace{
(\@old a@)[\@high@ - (\@p@ - \@low@)..\@high@)
}_{Z}
\cat{}{}
\underbrace{
(\@old a@)[\@p@..\@high@ - (\@p@ - \@low@))
\cat{}{}
}_{Y}
\underbrace{
(\@old a@)[\@low@..\@p@)
}_{X}
\,.
\]
Using the names assigned to each slice, the call to \<swap_sections> turns sequence $\cat{X}{\cat{Y}{Z}}$ into $\cat{Z}{\cat{Y}{X}}$; the following recursive call rotates $\cat{Z}{Y}$ by $r = \@p@ - \@low@$.
According to~\eqref{eq:lemma-rot-swap-left}, for $\xq = X$, $\yq = Y$, $\zq = Z$, and $d = r$, this achieves a rotation of the original sequence $\@(old a)@[\@low@..\@high@) = \cat{X}{\cat{Y}{Z}}$ by the same $r$, which establishes \<rotate_swap_helper>'s postcondition.
The \emph{right is smaller} case is symmetric and crucially relies on \eqref{eq:lemma-rot-swap-right} for $d = \@high@ - \@p@$ and $N = \@high@ - \@low@$.

\begin{figure}[!htb]
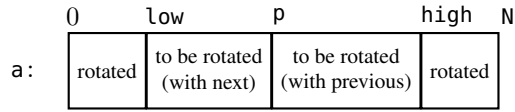

\begin{lstlisting}
rotate_swap_iterative(a: ARRAY[G]; r: INTEGER)
   require 0 < r < N
   ensure a = $\rot{\!}{\text{r}}$(old a)
   local low, high, p: INTEGER
   do
      low, p, high := 0, r, N
      while low < p < high   (*\label{cl:rswap:iterative:loop}*)
      invariant
         0 <= low <= p <= high <= N
         low = p liff p = high
            // rotated on the left
         forall $i$: 0 <= $i$ < low limplies a[$i$] = $\rot{\!}{\text{r}}$(old a)[$i$]
            // to be rotated
         p - low < high - low limplies forall $i$: low <= $i$ < high limplies 
                  $\rot{\!}{\text{p} - \text{low}}$(a[low..high))[$i$ - low] = $\rot{\!}{\text{r}}$(old a)[$i$]
            // rotated on the right
         forall $i$: high <= $i$ < N limplies a[$i$] = $\rot{\!}{\text{r}}$(old a)[$i$]
      do
         if p - low = high - p then
            // swap (*a[low..p)*) and (*a[p..high)*)
            swap_sections (a, low, high, p - low)
            // now the whole (*a[low..high)*) is in place
            low, high := low + (p - low), high - (high - p)
         elseif p - low < high - p then
            // swap (*a[low..p)*) and (*a[high $-$ (p $-$ low)..high)*)
            swap_sections (a, low, high, p - low)
            // now (*a[high $-$ (p $-$ low)..high)*) is in place
            high := high - (p - low)
         elseif p - low > high - p then
            // swap (*a[low..low $+$ (high $-$ p))*) and (*a[p..high)*)
            swap_sections (a, low, high, high - p)   (*\label{cl:rswap:iter:right}*)
            // now (*a[low..low $+$ (high $-$ p))*) is in place
            low := low + (high - p)  (*\label{cl:rswap:iter:right-index}*)
         end
      end
   end
\end{lstlisting}
  \caption{Rotate array \<a> by \<r> to the left by swapping sections of equal length: iterative algorithm.}
  \label{fig:rotation-swap-algo-iter}
\end{figure}

\paragraph{Rotation by swapping: correctness of iterative version.}
As usual, we start by identifying the straightforward bounding invariants.
Variables \<low> and \<high> mark a shrinking slice of \<a> as they get closer to \<p>, hence the obvious invariant
\begin{equation}
0 \leq \@low@ \leq \@p@ \leq \@high@ \leq \@N@\,.
\label{eq:rotation-swap-iter-bounding-1}
\end{equation}
As we prove inductiveness of this invariant based on how \<low> and \<high> are updated in every iteration, we notice that when the loop exits both \<p = low> and \<p = high> hold.
We record this fact with another bounding invariant
\begin{equation}
\@low@ = \@p@ \quad\liff\quad \@p@ = \@high@ \,,
\label{eq:rotation-swap-iter-bounding-2}
\end{equation}
which lets us establish that the interval \<[low..high)> is empty when the loop exits.

The bounding invariants suggest an essential invariant that predicates about three slices of \<a>:
\begin{center}
\begin{tikzpicture}
[node distance=-0.5pt]

\node (a) {\<a>$\colon$};
\node (left-unchanged) [array-slice,right=2mm of a] {rotated};
\node (tbr-1) [array-slice,right=of left-unchanged] {to be rotated\\(with next)};
\node (tbr-2) [array-slice,right=of tbr-1] {to be rotated\\(with previous)};
\node (tbr-3) [array-slice,right=of tbr-2] {rotated};

\node [xshift=2pt,yshift=6pt] at (left-unchanged.north west) {$0$};
\node [xshift=7pt,yshift=6pt] at (tbr-1.north west) {\<low>};
\node [xshift=3pt,yshift=6pt] at (tbr-2.north west) {\<p>};
\node [xshift=10pt,yshift=6pt] at (tbr-3.north west) {\<high>};
\node [xshift=3pt,yshift=6pt] at (tbr-3.north east) {\<N>};
\end{tikzpicture}
\end{center}
The leftmost and rightmost slices are initially empty and invariably in place as the loop iterates:
\begin{equation}
\begin{split}
\@a@[0..\@low@) & = \rot{(\@old a@)}{\@r@}[0..\@low@)\,,  \\
\@a@[\@high@..\@N@) & = \rot{(\@old a@)}{\@r@}[\@high@..\@N@)\,.
\end{split}
\label{eq:rotation-swap-iter-essential-1}
\end{equation}
The mid slice has to be rotated; precisely, slices \<a[low..p)> and \<a[p..high)> have to be swapped relative to each other:
\begin{equation}
\@p@ - \@low@ < \@high@ - \@low@ \limpl
\left(
\begin{array}{l}
\forall i\colon
\@low@ \leq i \leq \@high@
\quad \limpl \\
\rot{(\@a@[\@low@..\@high@))}{\@p@ - \@low@}[i - \@low@]
=
\rot{(\@old a@)}{\@r@}[i]
\end{array}
\right)\!.
\label{eq:rotation-swap-iter-essential-2}
\end{equation}
Invariant \eqref{eq:rotation-swap-iter-essential-2} is required only to make the other essential invariant inductive, not to establish the postcondition which follows from \eqref{eq:rotation-swap-iter-essential-1} alone upon exiting the loop.
Notice the index shift in the left-hand side of the equality in \eqref{eq:rotation-swap-iter-essential-2}: the first element (index $0$) of the rotation of sequence \<a[low..high)> by \<p - low> corresponds to the element at position \<low> in the rotation of the whole \<a> by \<r = p>.

Without the antecedent $\@p@ - \@low@ < \@high@ - \@low@$, \eqref{eq:rotation-swap-iter-essential-2} would not be inductive in the case \<p - low = high - p>: in this case, \<$r =\ $ p - low = high - low = 0 = $N$> after updating \<low> and \<high>, but \eqref{eq:rot-def} is undefined if $r = N$.
Then, proving inductiveness of the essential invariants in the ``left is smaller'' and ``right is smaller'' cases crucially relies on Lemma~\ref{lm:rot-swap}, following the same overall argument as the proof of the recursive version \<rotate_swap>.
Consider the \emph{right is smaller} case: \<p - low \> high - p>.
The call to \<swap_sections> in the corresponding branch of \<rotate_swap_iterative>'s loop (line~\ref{cl:rswap:iter:right} in Figure~\ref{fig:rotation-swap-algo-iter}) swaps $X = \@(old a)[low..low + (high - p))@$ with $Z = \@(old a)[p..high)@$, while leaving $Y =  \@(old a)[low + (high - p)..p)@$ untouched.
Thus, \<a[low..high)> consists of $\cat{Z}{\cat{Y}{Z}}$ after the swap.
For $N = \@high@ - \@low@$, $d = \@high@ - \@p@$, $\xq = X$, $\yq = Y$, and $\zq = Z$, \eqref{eq:lemma-rot-swap-right} shows that $Z$ is in place, whereas \<a[low + (high - p)..high)> must be rotated by $N - 2d = 2\@p@ - \@high@ - \@low@$.
After incrementing \<low> by \<high - p> (on line~\ref{cl:rswap:iter:right-index} in Figure~\ref{fig:rotation-swap-algo-iter}), this corresponds to a rotation by \<p - low> of \<a[low..high)>, thus establishing that \eqref{eq:rotation-swap-iter-essential-2} is inductive in this case.

\subsection{Rotation by modular visit}  \label{sec:algos:modular}
The rotation by modular visit algorithm has the property that it directly moves elements into their final position. 
To understand how it works, we see $\rot{}{r}$ as a permutation of the set $\setn{N}$---that is as the mapping $k \mapsto \wrap{(k + (N - r))}{N}$ defined in \eqref{eq:lemma-rot-2}.
(Remind that we only deal with left rotations: $0 \leq r < N$.)

\paragraph{Cycle decomposition of rotations.}
As recalled in Section~\ref{sec:modular-arithmetic}, $\rot{}{r}$ has a unique decomposition in disjoint cycles.
The first cycle starts from the element at index $0$, goes through the elements at indexes 
\begin{alignat*}{3}
      0 
& \quad\to\quad \wrap{(0 + (N - r))}{N} 
& \quad\to\quad \wrap{(0 + 2(N - r))}{N}
& \quad\to\quad \cdots
\end{alignat*}
until it reaches index $0$ again.
Similarly, the second cycle goes through  
\begin{alignat*}{3}
      1 
& \quad\to\quad \wrap{(1 + (N - r))}{N} 
& \quad\to\quad \wrap{(1 + 2(N - r))}{N}
& \quad\to\quad \cdots
\end{alignat*}
until it reaches $1$ again.
And a generic cycle that starts from $s$ is
\begin{alignat}{3}
      s 
& \quad\to\quad \wrap{(s + (N - r))}{N} 
& \quad\to\quad \wrap{(s + 2(N - r))}{N}
& \quad\to\quad \cdots
\label{eq:cycle-at-s}
\end{alignat}
until $s$.

The number of elements in each cycle is the smallest positive integer $t$ such that $\wrap{s + (t(N - r))}{N} = s$, which we equivalently express as the modular equation 
\begin{equation}
t(N - r) \equiv 0 \pmod N\,.
\label{eq:mod-period}
\end{equation}
The Linear Congruence Theorem~\cite[Th.~1.6.14]{NumberTheory}\footnote{Also:~\url{http://en.wikipedia.org/wiki/Linear_congruence_theorem}.} says that \eqref{eq:mod-period} has solutions for $t \in \{ k N / \gcd(N, N - r) \mid k \in \integers \}$.
The smallest positive integer in this set is obviously $N / \gcd(N, N - r)$, which is then the length of each cycle.

\begin{figure}[!htb]
\begin{lstlisting}
rotate_modulo (a: ARRAY[G]; r: INTEGER)
  require 0 < r < N
  ensure a = $\rot{\!}{\text{r}}$(old a)
  local start, moved, v: INTEGER; displaced: G
  do
     start := 0                 (*\label{cl:rmodulo:start-init}*)
     moved := 0
     while moved /= N           (*\label{cl:rmodulo:outerloop}*)
     invariant
        $0$ <= moved <= N
        $0$ <= start <= $\gcd$(N, N - r)
        moved = start $\!\!\cdot\!\!$ xxic
        forall $i, s, p$: $0$ <= $i$ < xxic$\,\!$ land $\,0$ <= $s$ < start land $p$ = xxmp($s, i$)
                  $\,$limplies a[$p$] = $\rot{\!}{\text{r}}$(old a)[$p$]
     do
        displaced := a[start]
        v := start   (*\label{cl:rmodulo:v-init}*)
        repeat      (*\label{cl:rmodulo:innerloop}*)
           v := v + N - r                     (*\label{cl:rmodulo:next}*)
           // wrap over (*a*)'s bounds
           if v >= N then v := v - N end      (*\label{cl:rmodulo:wrap}*)
           // swap (*a[v]*) and (*displaced*)
           a[v], displaced := displaced, a[v] (*\label{cl:rmodulo:swap}*)
           moved := moved + 1                 (*\label{cl:rmodulo:incmoved}*)
        invariant
          0 < moved - start $\!\!\cdot\!\!$ xxic $\,\leq$ xxic
          v = xxmp(start, moved - start $\!\!\cdot\!\!$ xxic)
          displaced = (old a)[v]
          forall $i, s, p$: $0$ <= $i$ < xxic$\,\!$ land $\,0$ <= $s$ < start land $p$ = xxmp($s, i$)
                     $\,$limplies a[$p$] = $\rot{\!}{\text{r}}$(old a)[$p$]
          forall $j, q$: $0$ < $j$ <= moved - start $\!\!\cdot\!\!$ xxic
                  $\,$land $q$ = xxmp(start, $j$) $\,$limplies $\,$ a[$q$] = $\rot{\!}{\text{r}}$(old a)[$q$]
        until v = start end    (*\label{cl:rmodulo:inner-exit}*)
        start := start + 1   (*\label{cl:rmodulo:incstart}*)
     end
  end
\end{lstlisting}
  \caption{Rotate array \<a> by \<r> to the left through modular visit.}
  \label{fig:rotation-by-modulo-algo}
\end{figure}

\paragraph{Rotation by visiting cycles.}
We finally have all elements to describe the rotation by modular visit algorithm, presented in Figure~\ref{fig:rotation-by-modulo-algo} and demonstrated on the running example in Figure~\ref{fig:rotation-by-modulo-example}.
The basic idea is to go through elements in the order given by the decomposition, one cycle at a time until all elements are moved.
During the visit, the element originally at position $0$ moves to position $\wrap{(0 + (N - r))}{N}$; the element originally at $\wrap{(0 + (N - r))}{N}$ moves to $\wrap{(0 + 2(N - r))}{N}$, and so on for all elements in the cycle.
Thanks to the unique decomposition property of permutations, this procedure eventually reaches all elements in the sequence; when they are all moved, the whole sequence has been rotated in place.

In the implementation of Figure~\ref{fig:rotation-by-modulo-algo}, 
the outermost loop (line~\ref{cl:rmodulo:outerloop}) performs a series of cyclic visits starting with the element at index \<start>---which is $0$ initially (line~\ref{cl:rmodulo:start-init}).
Variable \<moved> is a counter that records the number of number of elements that are in place; correspondingly, the outermost loop exits when \<moved = N> and the rotation is complete.

The inner loop (line~\ref{cl:rmodulo:innerloop}) actually performs the visits of the cycles; precisely, each iteration of the outer loop executes the inner loop to completion for the current value of \<start>, which visits all elements in the cycle beginning at \<start> as follows.
% In the running example (Figure~\ref{fig:rotation-by-modulo-example}) the first cycle has \<start = 0>.
With every iteration of the inner loop, a local variable \<v> takes on the values in the cycle beginning at \<start>: \<start>, \<start + N - r>, and so on, where each new value of \<v> (line~\ref{cl:rmodulo:next}) is wrapped over when it overflows \<N> (line~\ref{cl:rmodulo:wrap}).
After updating \<v>, the inner loop exchanges \<a[v]> with the element at the \emph{previous} position in the cycle (line~\ref{cl:rmodulo:swap}), which is stored in variable \<displaced> (initially equal to \<a[start]> and successively updated after updating \<v>). It then continues with the next iteration.
In the running example, the first iteration of the inner loop begins with $\@displaced@ = \@a@[\@0@] = \text{A}$ (top left in Figure~\ref{fig:rotation-by-modulo-example});  which it writes to position \<0 + N - r = 4> (its position in the rotation by \<r>) while saving \<(old a)[4]> into \<displaced> for the next iteration (mid right in Figure~\ref{fig:rotation-by-modulo-example}).

Earlier in this section, we established that each cycle has $\@N@ / \gcd(\@N@, \@N@ - \@r@)$ elements.
Hence, the inner loop iterates this many times before reaching the exit condition \<v = start> (line~\ref{cl:rmodulo:inner-exit}).
In the running example, $\@N@ = 6$, $\@r@ = 2$, and $\gcd(6, 4) = 2$, and in fact the inner loop has put $6/2 = 3$ elements in place when it reaches \<start> again (bottom right in Figure~\ref{fig:rotation-by-modulo-example}).
The outer loop correspondingly performs exactly $\@N@ / (\@N@ / \gcd(\@N@, \@N@ - \@r@)) = \gcd(\@N@, \@N@ - \@r@)$ iterations, which is when the last cycle in the decomposition is visited (mid left in Figure~\ref{fig:rotation-by-modulo-example}, where the outer loop iterates twice).

\begin{figure}[!ht]
\centering
\begin{tikzpicture}
[
  node distance=1pt,
  pin distance=2.5mm,
  every pin edge/.style={<-, shorten <=1pt, thick,black}
]

  \node (A) [generic-item,pin=above:{\<start>},pin=below:{\<v>}] {A};
  \node (B) [generic-item,right=of A] {B};
  \node (C) [generic-item,right=of B] {C};
  \node (D) [generic-item,right=of C] {D};
  \node (E) [generic-item,right=of D] {E};
  \node (F) [generic-item,right=of E] {F};

  \node (dis) [left-item,below=8mm of D,label=left:\<displaced:>] {A};

  \node (A2) [generic-item,right=25mm of F,pin=above:{\<start>}] {A};
  \node (B2) [generic-item,right=of A2] {B};
  \node (C2) [generic-item,right=of B2] {C};
  \node (D2) [generic-item,right=of C2] {D};
  \node (E2) [left-item,right=of D2,pin=below:{\<v>}] {A};
  \node (F2) [generic-item,right=of E2] {F};

  \node (dis2) [right-item,below=8mm of D2,label=left:\<displaced:>] {E};

  \node (A3) [generic-item,below=25mm of A2,pin=above:{\<start>}] {A};
  \node (B3) [generic-item,right=of A3] {B};
  \node (C3) [right-item,right=of B3,pin=below:{\<v>}] {E};
  \node (D3) [generic-item,right=of C3] {D};
  \node (E3) [left-item,right=of D3] {A};
  \node (F3) [generic-item,right=of E3] {F};

  \node (dis3) [right-item,below=8mm of D3,label=left:\<displaced:>] {C};

  \node (A4) [right-item,below=25mm of A3,pin=above:{\<start>},pin=below:{\<v>}] {C};
  \node (B4) [generic-item,right=of A4] {B};
  \node (C4) [right-item,right=of B4] {E};
  \node (D4) [generic-item,right=of C4] {D};
  \node (E4) [left-item,right=of D4] {A};
  \node (F4) [generic-item,right=of E4] {F};

  \node (dis4) [left-item,below=8mm of D4,label=left:\<displaced:>] {A};

  \node (F5) [generic-item,left=25mm of A4] {F};
  \node (E5) [left-item,left=of F5] {A};
  \node (D5) [generic-item,left=of E5] {D};
  \node (C5) [right-item,left=of D5] {E};
  \node (B5) [generic-item,left=of C5,pin=above:{\<start>},pin=below:{\<v>}] {B};
  \node (A5) [right-item,left=of B5] {C};

  \node (dis5) [left-item,below=8mm of D5,label=left:\<displaced:>] {B};

  \node (F6) [left-item,left=25mm of A3] {B};
  \node (E6) [left-item,left=of F6] {A};
  \node (D6) [right-item,left=of E6] {F};
  \node (C6) [right-item,left=of D6,pin=above:{\<start>}] {E};
  \node (B6) [right-item,left=of C6,pin=below:{\<v>}] {D};
  \node (A6) [right-item,left=of B6] {C};

  \node (dis6) [generic-item,below=8mm of D6,label=left:\<displaced:>] {D};

  \begin{scope}[-latex,thick]
  \path ($(F.east)+(0,-7mm)$) edge ($(A2.west)+(0,-7mm)$);
  \path ($(F2.east)+(0,-6mm)$) edge[bend left] ($(F3.east)+(0,-6mm)$);
  \path ($(F3.east)+(0,-8mm)$) edge[bend left] ($(F4.east)+(0,-7mm)$);
  \path ($(A4.west)+(0,-7mm)$) edge ($(F5.east)+(0,-7mm)$) ;
  \path (A.west) edge[bend right=45,dotted] node[left] {$\rot{}{2}$} (A6.west);
  \path ($(A5.west)+(0,-7mm)$) edge[bend left,dashed] ($(A6.west)+(0,-7mm)$);
  \end{scope}
\end{tikzpicture}
\caption{Rotating sequence $\text{A}\,\text{B}\,\text{C}\,\text{D}\,\text{E}\,\text{F}$ by $2$ through modular visit.}
\label{fig:rotation-by-modulo-example}
\end{figure}

According to Gries and Mills~\cite{swapping-TR}, the ``jumping around'' pattern of the cyclic visits suggested the name ``dolphin algorithm'' by which it is sometimes referred to---like a dolphin that leaps out of water and plunges back into it someplace forth.

\subsubsection{Rotation by modular visit: complexity} \label{sec:modulo:complexity}
The illustration of the algorithm suggests the complexity of rotation by modular visit.
Clearly, only a finite amount of scratch memory is needed; hence the space complexity is $\Theta(1)$.
The outer loop iterates $\gcd(N - r, N)$ times, each of which sees the inner loop iterate $N / \gcd(N - r, N)$; hence the time complexity is $\Theta(N)$.

This corresponds to $N$ array writes (one per element put in place).
Gries and Mills~\cite{swapping-TR} present a variant of the algorithm that puts the elements in place in each cycle backwards, using $N + \gcd(N - r, N)$ array \emph{accesses}: one for each element plus $\gcd(N - r, N)$ to temporarily save the array value put in place last and overwritten first (i.e., \<a[start]>).
If we count swapping a pair of array elements as three array accesses (using a temporary variable for the swap), this variant of the modular visit algorithm performs the fewest number of array writes among the rotation algorithms.
Even in the form of Figure~\ref{fig:rotation-by-modulo-algo}, rotation by modular visit has the property that it swaps elements directly into their final position (using \<displaced> as pivot).

Nevertheless, Table~\ref{tab:complexity-overall} suggests that the algorithm does not scale well in practice.
While we have not thoroughly investigated the reasons for this lackluster practical performance, it might have to do with (lack of) locality in access: the ``jumping around'' of modular visits accesses non-adjacent elements which may generate many cache misses when a large array cannot be stored in the fastest level of the memory hierarchy.\footnote{It is somewhat surprising that method \<rotate> of \<java.util.Collections> in OpenJDK 6 uses rotation by modular visit (with an implementation very similar to the one used for the experiments reported in Table\ref{tab:complexity-overall}) not only for small collections but also whenever a collection supports constant-time random access---as in arrayed lists.}

\subsubsection{Rotation by modular visit: correctness} \label{sec:modulo:proof}

\begin{minipage}[c][4em][c]{0.8\textwidth}
\begin{quote}
\emph{A proof that it works is remarkably difficult\ldots} \\
--- Richard Bornat about rotation by modular visit~\cite{Bornat-slides}
\end{quote}
\end{minipage}

We introduce the loop invariants necessary to prove correctness; we then discuss how to prove their inductiveness.

\paragraph{Loop invariants.}
The formal analysis starts with the outer loop: each iteration visits (and puts in place) all elements whose index is in the cycle that begins at \<start>.
The bounding invariants
\begin{align}
0 &\leq \@moved@ \leq \@N@\,,
\label{eq:modular:bounding-moved} \\
0 &\leq \@start@ \leq \gcd(\@N, N@ - \@r@)\,,
\label{eq:modular:bounding-start}
\end{align}
are then easy to justify (but not so easy to prove!).
\eqref{eq:modular:bounding-moved} follows from the number of moved elements being initially zero; and the outer loop exiting when all \<N> elements have been moved.
\eqref{eq:modular:bounding-start} captures the fact that the outer loop executes exactly $\gcd(\@N@, \@N@ - \@r@)$ times---as each cycle visits $\@N@ / \gcd(\@N@, \@N@ - \@r@)$ elements.
Since the quantity $\@N@ / \gcd(\@N@, \@N@ - \@r@)$ will feature often in the invariants and proof of the algorithm, we give it an abbreviation:
\begin{equation}
\ic{N}{N - r} \quad=\quad 
\frac{N}{\gcd(N, N - r)}\,.
\label{eq:ic-def}
\end{equation}

The first usage of this definition is to express an exact relation between \<moved> and \<start>.
As we repeatedly discussed, every cycle consists of exactly \<xxic> elements, and \<start> is incremented on line~\ref{cl:rmodulo:incstart} after every cycle is completed; hence the invariant 
\begin{equation}
\@moved@ = \@start@ \cdot \@xxic@\,.
\label{eq:modular:bounding-moved-from-start}
\end{equation}

The outer loop's essential loop invariant precisely characterizes the elements put in place by each iteration of the loop.
Given an initial value $\overline{s}$ of \<start>, such an iteration visits all elements whose indexes are in the cycle of \<xxic> elements defined by \eqref{eq:cycle-at-s} for $s = \overline{s}$.
Using the abbreviation
\begin{equation}
\mpos{N}{M}{s}{k}
\quad=\quad
\wrap{(s + kM)}{N}
\label{eq:mp-def}
\end{equation}
to denote the $k$th index in the cycle of step $M$ modulo $N$ that starts at $s$, the essential loop invariant is
\begin{equation}
\forall i, s, p\colon
\left(
\begin{array}{cl}
& 0 \leq i < \@xxic@ \\
\land
& 0 \leq s < \@start@ \\
\land
& p = \@xxmp@(s, i)
\end{array}
\right)
\limpl
\@a@[p] = \rot{(\@old a@)}{\@r@}[p]\,.
\label{eq:modular:outer:essential}
\end{equation}
That is, all elements of all cycles originating in values of \<start> less than the current one have been put in place.

Moving on to the inner loop,\footnote{Notice it is a \<repeat...until> loop, and hence initiation for its invariants means that they have to hold after one iteration.}
variable \<v> takes on all the values in the currently visited cycle (beginning at \<start> on line~\ref{cl:rmodulo:v-init}).
Since all cycles previously visited have the same length \<xxic>, we can express the value of \<v> as a function of \<start> and \<moved>:
\begin{equation}
\@v@ =
\@xxmp@ (\@start, moved@ - \@start@\cdot\@xxic@)\,.
\label{eq:modular:inner:var-v}
\end{equation}
Expression $\@moved@ - \@start@\cdot\@xxic@$ is $1$ initially (that is, after the first unconditionally executed loop iteration); it is \<xxic> when the inner loop exits with \<v = start>.
Hence the bounding loop invariant
\begin{equation}
\@0@ < \@moved@ - \@start@\cdot\@xxic@ \leq \@xxic@\,.
\label{eq:modular:inner:bounding-expr}
\end{equation}
Given the current value of \<v>, \<displaced> is simply the value in \<a> initially at index \<v>:
\begin{equation}
\@displaced@ =
(\@old a@)[\@v@]\,.
\label{eq:modular:inner:displaced}
\end{equation}

The outer loop's essential invariant \eqref{eq:modular:outer:essential} is also maintained by the inner loop: \eqref{eq:modular:outer:essential} only involves elements whose indexes are in a cycle starting at some $s < \@start@$, but these cycles are disjoint from the currently visited cycle (which begins at \<start>).
To describe partial progress made by the inner loop in visiting the current cycle, we introduce another essential invariant:
\begin{equation}
\forall j, q\colon
\left(
\begin{array}{cl}
& 0 < j \leq \@moved@ - \@start@\cdot\@xxic@ \\
\land
& q = \@xxmp@(\@start,@\, j)
\end{array}
\right)
\limpl
\@a@[q] = \rot{(\@old a@)}{\@r@}[q]\,.
\label{eq:modular:inner:essential}
\end{equation}
Quantified variable $j$ determines the position in the current cycle; correspondingly, $j$'s range of quantification excludes $0$, since $\@a@[\@xxmp@(\@start, 0@)]$ is set last, and includes $\@moved@ - \@start@\cdot\@xxic@$, corresponding to the element set in the latest loop iteration.

\paragraph{Proving initiation.}
Initiation is trivial for the outer loop invariants, so let's focus on initiation for the inner loop invariants.
The outer loop's \eqref{eq:modular:bounding-moved-from-start} still holds at the beginning of the inner loop body, since neither \<start> nor \<moved> has changed.
Thus, incrementing \<moved> at the end of the inner loop body makes $\@moved@ - \@start@\cdot\@xxic@ = 1$, which proves initiation of \eqref{eq:modular:inner:bounding-expr}.
Since $\@xxmp@(\@start, 1@)$ is $\wrap{(\@start@ + \@N@ - \@r@)}{\@N@}$, \eqref{eq:modular:inner:var-v} also holds initially.
For similar reasons, \<displaced> stores the value originally at index \<v> that has just been assigned to; hence \eqref{eq:modular:inner:displaced} initially.

We already discussed that the outer loop essential invariant's validity is not affected by the inner loop's work; hence \eqref{eq:modular:outer:essential} satisfies initiation and consecution with respect to the inner loop.
Finally, initiation for \eqref{eq:modular:inner:essential} amounts to proving that the value assigned to $\@a@[\@xxmp@(\@start, @1)]$ in the first execution of line~\ref{cl:rmodulo:swap} is the one of $\rot{(\@old a@)}{\@r@}$; this can be done by matching the definitions of \<xxmp> and of rotation \eqref{eq:lemma-rot-2}.

\paragraph{Proving consecution.}
The outer loop's bounding invariant \eqref{eq:modular:bounding-start} is unaffected by the inner loop, which does not modify \<start>.
Its inductiveness follows from the bound \eqref{eq:modular:bounding-moved} on \<moved> and on the connection \eqref{eq:modular:bounding-moved-from-start} between the latter and \<start>.

For the consecution proofs of the remaining outer loop invariant, we rely on the inner loop invariants.
When the inner loop terminates, \<v = start>; through \eqref{eq:modular:inner:var-v}, it follows that $\@moved@ - \@start@\cdot\@xxic@ = \@xxic@$; hence \eqref{eq:modular:bounding-moved-from-start} is restored after incrementing \<start> by one on line~\ref{cl:rmodulo:incstart}.

We now move to proving inductiveness of the inner loop invariants.
Since the inner loop exits when \<v = start>, \eqref{eq:modular:inner:var-v} implies that the increment of \<moved> does not overflow \<xxic> relative to the initial value at the current outer loop iteration; hence \eqref{eq:modular:inner:bounding-expr} is maintained.
Conversely, \eqref{eq:modular:inner:var-v} is maintained because the inner loop body implements the definition of \<xxmp> with respect to the current \<moved> that is incremented by one.
Along the same lines one can prove that \eqref{eq:modular:inner:displaced} is inductive.

Disjointness of the cycles visited by the inner loop ensures that \eqref{eq:modular:outer:essential} is also maintained.
Proving consecution of \eqref{eq:modular:inner:essential} is more involved.
Thanks to the inductive hypothesis, we only have to establish progress for $q = \@v@$.
The inner loop, on line~\ref{cl:rmodulo:swap}, assigns to \<a[v]> the value that was assigned to \<displaced> in the \emph{previous} loop iteration; from \eqref{eq:modular:inner:displaced}, noting that \<moved> has just been incremented, this is the element in \<old a>  whose index is given by the previous value in the cycle, that is the previous value of \<v>.
This is in fact $\rot{(\@old a@)}{\@r@}[\@v@]$ because of how \<xxmp> is defined.

\paragraph{Final correctness.}
The final step in the correctness proof of \<rotate_modulo> is establishing the postcondition from the outer loop invariants.
When the outer loop terminates, \<moved = N>; \eqref{eq:modular:bounding-moved-from-start} implies that $\@start@ = \gcd(\@N@, \@N@ - \@r@)$.
Therefore, proving the postcondition boils down to the following lemma.
\begin{lemma} \label{lm:invert-mp}
Given $r$ and $N$ satisfying the precondition of \<rotate_modulo>: for every $0 \leq k < N$, there exist $0 \leq i < \ic{N}{N - r}$ and $0 \leq s < \gcd(N, N - r)$ such that $\mpos{N}{N - r}{s}{i} = k$.
\end{lemma}
\begin{proof}
The lemma ultimately follows from the property of decomposition in cycles of permutations (Section~\ref{sec:modular-arithmetic}).

Since $k$ is listed in some cycle \eqref{eq:cycle-at-s}, the three variables $k, s, i$ satisfy $s + i(N - r) \equiv k \pmod N$, which indicates that $s = k \bmod (N - r)$.
It follows, from standard properties of modular arithmetic~\cite[Sec.~1.6.2]{NumberTheory}, that $N - r$ divides $k - s$.
Hence, we can rewrite the expression that relates $k$, $s$, and $i$---where now only $i$ is unknown---as
\begin{equation}
i \cdot \frac{N - r}{g} \equiv \frac{k - s}{g} \left(\operatorname{mod} \frac{N}{g}\right),
\label{eq:solve-for-i}
\end{equation}
for $g = \gcd(N, N - r)$.
\eqref{eq:solve-for-i} has exactly one solution: use the extended Euclidean algorithm to find $x$ and $y$ such that $x (N - r)/g + y N / g = \gcd((N - r)/g, N/g) = 1$ satisfies B\'ezout's identity.
Then, $i = x(k - s)/g$ is the unique solution.
\end{proof}

%%% Local Variables: 
%%% mode: latex
%%% TeX-master: t
%%% End: 

%% file: boogie.tex
\section{Rotation: mechanized proofs in Boogie}
\label{sec:all-boogie-proofs}
\lstset{language=boogie}

Notwithstanding our efforts to be as rigorous as possible in the correctness arguments of Section~\ref{sec:algos}, there still is substantial ground to cover before we can have mechanized proofs.
Part of the remaining gap is due to the unforgiving level of precision that is required by mechanical proof tools; another part is more specific to the nature of a specific tool we may choose, such as its level of automation and limitations.
In this section, we provide a detailed account of what is necessary to turn the proof ideas of Section~\ref{sec:algos} into successful verification using Boogie~\cite{BoogieManual}.\footnote{The presentation assumes basic familiarity with the Boogie language and tool.}

Boogie is an \emph{auto-active} tool, providing a level of automation intermediate between completely automatic (such as in static analyzers) and interactive (such as in proof assistants).
In practice, users interact with the tool offline by providing \emph{annotations} (such as assertions and lemmas) that guide proof attempts.

Boogie is mainly used as an intermediate language for verification; hence we will have to provide annotations at a relatively low level of detail.
This will turn out to be instructive and will showcase several fundamental categories of annotations and annotation styles that are present, in one form or another, in practically every auto-active tool---and possibly in other kinds of tools as well.

The four rotation algorithms make for a gradual introduction to these features of automated verification, as each of them requires new specific annotation techniques:
\begin{description}
\item[Rotation by copy] is simple enough that it only requires basic definitions; we can replicate the proof essentially as done on paper.

\item[Rotation by reversal] requires intermediate assertions to guide the prover and explicit \emph{lemmas} proved separately and applied where appropriate in the main correctness proof.

\item[Rotation by swapping] requires lemmas with non-trivial proofs and modularization tailored to the proof outline; it also requires a little usage of \emph{triggers} to curb instantiation patterns of the underlying automatic theorem prover.

\item[Rotation by modular visit] requires clever axiomatization, as well as non-trivial \emph{ghost code} added to the implementation specifically to represent additional information about program state required to justify the correctness proof; a \emph{framing} specification is of the essence.
\end{description}

In each case, the features used (such as ghost code or framing annotations) are not necessarily the only way to carry out a proof of that algorithm using Boogie.
However, they support a natural approach, and one that is often idiomatic to using auto-active tools of the same family.

\subsection{Sequences and rotated sequences} \label{sec:boogie:theory}

Before delving into the details of the algorithms, Figure~\ref{fig:boogie:seq-rot-def} introduces some basic definitions that we will use in all the Boogie annotations and proofs.
As arrays, we use Boogie \emph{maps} from integers to integers (type \<[int]int>).
While we could use a generic type as codomain, sticking to plain integers generally works better as it has better support with the underlying SMT solver (in other words, it requires fewer explicit axioms).

\begin{figure}[!tb]
\begin{boogie}
// Sequence of (*a[low..high)*)
function seq(a: [int]int, low: int, high: int) returns([int]int);
axiom (forall a: [int]int, low: int, high: int, i: int :: 
   0 <= i && i < high - low ==> seq(a, low, high)[i] == a[low + i]);

// Definition (*\eqref{eq:wrap-def}*): (*i wrap N*)
function wrap(i: int, N: int) returns(int);
axiom (forall i, N: int :: 0 <= i && i < N ==> wrap(i, N) == i);
axiom (forall i, N: int :: (*\label{cl:axiom:wrap-2}*)
        0 < N && N <= i ==> wrap(i, N) == wrap(i - N, N));

// Left-rotated sequence of (*a[low..high)*) by (*r*)
function rot(a: [int]int, low: int, high: int, r: int)
         returns([int]int);
axiom (forall a: [int]int, low: int, high: int, i, r: int :: 
  0 <= r && r < high - low && 0 <= i && i < r ==> 
  rot(a, low, high, r)[i + high - low - r] == seq(a, low, high)[i]);
axiom (forall a: [int]int, low: int, high: int, i, r: int :: 
  0 <= r && r < high - low && r <= i && i < high - low ==> 
  rot(a, low, high, r)[i - r] == seq(a, low, high)[i]);
\end{boogie}
\caption{Boogie declarations and axiomatic definitions of array slice sequence \<seq>, \<wrap>, and rotated sequence \<rot>.}
\label{fig:boogie:seq-rot-def}
\end{figure}

First, it is convenient to have an explicit representation of array slices as \emph{sequences}.
Since we still using maps to represent sequences, this amounts to a scaling of indexes, so that the lower index of the sequence corresponding to \<a[low..high)> is zero.
Thus, function \<seq> takes a map \<a>, a lower index \<low>, and an upper index \<high> and returns another map whose content over indexes \<[0..high - low)> coincides with \<a[low..high)>.

The two axioms defining \<wrap> replicate definition \eqref{eq:wrap-def} verbatim.
And the definition of rotated sequence follows \eqref{eq:rot-def}: precisely, function \<rot> takes a slice \<a[low..high)> and a rotation coefficient \<r> and returns a sequence (that is, a zero-based map) representing the left rotation of \<a[low..high)> by \<r>.
Since we only consider nonnegative values of \<r>, the axiomatic definition of \<rot> consists of two axioms, in the same order as definition \eqref{eq:rot-def}: the first one describes the head of \<(high - low) - r> elements; and the second one describes the tail of \<r> elements of the rotated sequence.

Using functions \<seq> and \<rot>, Figure~\ref{fig:boogie:seq-rot-def} shows the generic signature and input/output specification of a Boogie procedure that performs rotation.
Since input arguments are read only in Boogie, \<rotate> returns another map \<b> whose content represents the input slice \<a[0..N)> after processing.
The rest is as in Figure~\ref{fig:rotate-specification}, but we have to make explicit, in the postcondition, the quantification over range that was implicit in the notation $\@a@ = \rot{(\@old a@)}{\@r@}$.

\begin{figure}[!hbt]
\begin{boogie}
procedure rotate(a: [int]int, N: int, r: int) returns(b: [int]int)
  requires 0 < r && r < N;
  ensures (forall i: int :: 
    0 <= i && i < N ==> seq(b, 0, N)[i] == rot(a, 0, N, r)[i]);
\end{boogie}
  \caption{Rotate array \<a> by \<r> to the left: signature and specification in Boogie.}
  \label{fig:boogie:rotate-specification}
\end{figure}

\subsection{Rotation by copy: mechanized proofs as on paper} \label{sec:boogie:copy}

Rotation by copy retains most of its simplicity in Boogie.
As Figure~\ref{fig:boogie:rotate-copy} shows, the same implementation and loop invariants of Figure~\ref{fig:rotation-by-copy-algo} work in Boogie.

\begin{figure}[!tb]
\begin{boogie}
// Rotate (*a[0..N)*) to the left by (*r*) by copying
procedure rotate_copy(a: [int]int, N: int, r: int) 
          returns(b: [int]int)
  requires 0 < r && r < N;
  ensures (forall i: int :: 
    0 <= i && i < N ==> seq(b, 0, N)[i] == rot(a, 0, N, r)[i]);
{
   var s, d: int;
   s, d := 0, N - r;
   while (s < N)
   invariant 0 <= s && s <= N;
   invariant d == wrap(s + N - r, N);
   invariant (forall i: int :: 0 <= i && i < s ==> (*\label{cl:boogie:copy:essential}*)
     seq(a, 0, N)[i] == seq(b, 0, N)[wrap(i + N - r, N)]);
   {
      b[d] := a[s];
      s := s + 1;
      d := d + 1;
      if (d == N) { d := 0; }
   }
   assert (forall i: int :: 0 <= i && i < N ==>  (*\label{cl:boogie:copy:assert}*)
     rot(a, 0, N, r)[i] == seq(a, 0, N)[wrap(i + r, N)]);
}
\end{boogie}
\caption{Verified Boogie annotated implementation of the rotation by copy algorithm of Figure~\ref{fig:rotation-by-copy-algo}.}
\label{fig:boogie:rotate-copy}
\end{figure}

The proof outline presented in Section~\ref{sec:copy:proof} mentioned that the essential loop invariant implies the postcondition thanks to the equivalent definition of rotation \eqref{eq:lemma-rot-2}.
In a similar way, Boogie has to prove that the representation of \<rot> in terms of \<wrap>, used in the loop invariant, and the axiomatic definition of \<rot> in Figure~\ref{fig:boogie:seq-rot-def}, used in the postcondition, are equivalent.
To this end, we introduce an assert expressing \eqref{eq:lemma-rot}, on line~\ref{cl:boogie:copy:assert} in Figure~\ref{fig:boogie:rotate-copy}.

Since \eqref{eq:lemma-rot} is \eqref{eq:lemma-rot-2}'s inverse, it is the former that translates from a representation based on the latter into one conforming to \<rot>'s axioms.
Namely, after proving the \<assert> from \<rot>'s definition, Boogie's reasoning follows this chain of equalities, for a generic index $0 \leq k < \@N@$:
\begin{center}
\begin{boogie}[numbers=none]
   rot(a, 0, N, r)[$k$] 
=  seq(a, 0, N)[$k$ + r]    (assert on line (*\ref{cl:boogie:copy:assert}*) in Figure (*\ref{fig:boogie:rotate-copy}*))
\end{boogie}
\end{center}

\begin{center}
\begin{boogie}[numbers=none]
=  seq(b, 0, N)[wrap(($k$ + r) + N - r, N)] 
                      (invariant on line (*\ref{cl:boogie:copy:essential}*) in Figure (*\ref{fig:boogie:rotate-copy}*))
=  seq(b, 0, N)[wrap($k$ +$\,$N, N)]                 (*(arithmetic)*)
=  seq(b, 0, N)[$k$]         (axiom on line (*\ref{cl:axiom:wrap-2}*) in Figure (*\ref{fig:boogie:seq-rot-def}*))
   QED                             (postcondition ensures).
\end{boogie}
\end{center}

\subsection{Rotation by reversal: mechanized lemmas} \label{sec:boogie:reverse}

Rotation by reversal requires expressing Lemma~\ref{lm:rev-cat} in Boogie, which is crucial to prove that the three reversals achieve a rotation of the original sequence.
In order to be able to do that, we first axiomatize reversal along the same lines as done for rotation.

\begin{figure}[!p]
\begin{boogie}[xleftmargin=6pt]
// Reversed sequence of (*a[low..high)*)
function rev(a: [int]int, low: int, high: int) returns([int]int);
axiom (forall a: [int]int, low: int, high: int, i: int ::   (*\label{cl:axiom:rev-def}*)
  0 <= i && i < high - low ==> 
  rev(a, low, high)[i] == seq(a, low, high)[high - low - 1 - i]);

// The position (*i*) maps to in a reversal of (*[low..high)*)
function rp(i: int, low: int, high: int) returns(int);
axiom (forall i, low, high: int :: rp(i, low, high) == high + low - 1 - i);
\end{boogie}
\caption{Boogie declarations and axiomatic definitions of reversed sequence \<rev> and inverse index mapping \<rp> in a reversal.}
\label{fig:boogie:rev:def}
\end{figure}

\subsubsection{Axioms and lemmas about reversal}
Figure~\ref{fig:boogie:rev:def} shows an axiomatization based on~\eqref{eq:rev-def}: \<rev(a, low, high)> is the sequence obtained by reversal of \<a[low..high)>.
We also introduce a function \<rp(i, low, high)> that represents the mapping used in the essential loop invariant of \<reverse>.
This is merely a convenience, since we could equivalently use the expanded expression \<high + low - 1 - i> wherever \<rp(i, low, high)> occurs.
However, this choice may have an impact in practice because Boogie introduces different triggers for integer expressions than for uninterpreted function applications.
For lack of space, we won't discuss every single alternative in detail; trying out some of them is a useful exercise, also to realize the sensitivity of Boogie's encoding to changes in annotation style.

\begin{figure}[!p]
\begin{boogie}
// Representation (*\eqref{eq:lemma-rot}*) is equivalent to (*rot*)'s definition
procedure lemma_rot(a: [int]int, low, high: int, r: int, p: int)
   requires low <= high;
   requires 0 <= r && r < high - low;
   requires 0 <= p && p < high - low;
   ensures rot(a, low, high, r)[p] == 
           seq(a, low, high)[wrap(p + r, high - low)];
{ }

// Lemma (*\ref{lm:rev-cat}*) (reverse of concatenation)
procedure lemma_rev_cat(t: [int]int, tl: int, th: int, 
                        s: [int]int, sl: int, sh: int, 
                        c: [int]int, p: int)
   requires tl <= th && sl <= sh;
   requires 0 <= p && p < sh - sl + th - tl;
   requires (forall i: int :: 0 <= i && i < th - tl ==> 
     seq(c, 0, th - tl)[i] == rev(t, tl, th)[i]);
   requires (forall i: int :: 0 <= i && i < sh - sl ==> 
     seq(c, th - tl, sh - sl + th - tl)[i] == rev(s, sl, sh)[i]);
   ensures 0 <= p && p < sh - sl ==>   (*\label{cl:boogie:revcat:p1}*)
     rev(c, th - tl, th - tl + sh - sl)[p] == seq(s, sl, sh)[p];
   ensures 0 <= p && p < th - tl ==>   (*\label{cl:boogie:revcat:p2}*)
     rev(c, 0, th - tl)[p] == seq(t, tl, th)[p];
   ensures 0 <= p && p < sh - sl ==>   (*\label{cl:boogie:revcat:p3}*)
     rev(c, 0, th - tl + sh - sl)[p] == seq(s, sl, sh)[p];
   ensures 0 <= p && p < th - tl ==>   (*\label{cl:boogie:revcat:p4}*)
     rev(c, 0, th - tl + sh - sl)[p + sh - sl] == seq(t, tl, th)[p];
{ }
\end{boogie}
\caption{Two Boogie lemma procedures about rotation and reversal. Procedure \<lemma_rot> establishes that \eqref{eq:lemma-rot} is equivalent to \<rot>'s axiomatic definition.
Procedure \<lemma_rev_cat> represents, for a generic index $0 \leq \@p@ < \len{\sq} + \len{\tq}$,  Lemma~\ref{lm:rev-cat} with $\sq = \@s@[\@sl@..\@sh@)$, $\tq = \@t@[\@tl@..\@th@)$, and $\@c@[0..\len{\sq}+\len{\tq}) = \cat{\rev{\tq}}{\rev{\sq}}$.}
\label{fig:boogie:rev:lemmas}
\end{figure}

\paragraph{Lemmas as procedures.}
In Boogie, lemmas are encoded as \<procedure>s without returned values: preconditions express the lemma's hypotheses; preconditions express the lemma's statement; and the procedure body outlines steps in the lemma's proof.
For example, this is a lemma procedure expressing the conclusion of a classic syllogism:
\begin{center}
\begin{boogie}[numbers=none,xleftmargin=16pt]
procedure syllogism(p: P) requires greek(p) ensures mortal(p)
{ assert human(p); /* (*p*) is human, and hence mortal */ }
\end{boogie}
\end{center}

Since lemmas are procedures in Boogie, asserting a lemma is done by calling the corresponding procedure.
For example, we can ask Boogie to derive the fact that Socrates is mortal by the instruction 
\<call syllogism(socrates)>, which checks that \<greek(socrates)> and derives that \<mortal(socrates)>.
To instantiate a lemma for a generic value of some of its arguments, there is the \<\\cforall> instruction: $\@\\cforall@\ \@syllogism(*)@$ makes the fact that every Greek is mortal available at the call site.

\paragraph{Two lemmas about rotation and reversal in Boogie.}
The first lemma we present is the Boogie version of \eqref{eq:lemma-rot} as an alternative definition of \<rot>.
We already used this fact in the proof of \<rotate_copy>, where we introduced it as an ad hoc \<assert>; now, we proceed systematically and formalize it as procedure \<lemma_rot>.

The second lemma is the fundamental Lemma~\ref{lm:rev-cat}.
Procedure \<lemma_rev_cat> formalizes it in a form that is readily usable with maps: it asserts that, given slices \<s[sl..sh)>, \<t[tl..th)>, and \<c[0..sh - sl + th - tl)> such that \<c>'s slice equals the concatenation of \<t[tl..th)>'s reversal and \<s[sl..sh)>'s reversal, reversing the whole \<c>'s slice gives the concatenation of \<s[sl..sh)> and \<t[tl..th)>.
Precisely, only the last two \<ensures>, lines~\ref{cl:boogie:revcat:p3} and \ref{cl:boogie:revcat:p4} in Figure~\ref{fig:boogie:rev:lemmas}, express the lemma's conclusion.
By contrast, using the more readily understandable slice notation, the first two \<ensures> express that
\[
\begin{split}
\rev{\@c@[\@ah - al..ah - al + bh - bl@)}
& = \@b@[\@bl..bh@)\,, \\
\rev{\@c@[\@0..ah - al@)}
& = \@a@[\@al..ah@)\,,
\end{split}
\]
which is part of the information used in proving the lemma.
If we wanted to directly reflect the lemma's structure on paper, we would move the formulas on lines~\ref{cl:boogie:revcat:p1} and \ref{cl:boogie:revcat:p2} in Figure~\ref{fig:boogie:rev:lemmas} as \<assert>s inside \<lemma_rev_cat>'s body.
This is another alternative that we do not explore in full.
It turns out, however, that having those formulas as \<ensures> rather than \<assert> makes for an overall faster verification---probably because the extra \<ensures> are useful facts where the lemma is employed: not having to derive them again from other available facts at the call site is advantageous.

Short of this, both lemma procedures have empty bodies: Boogie can round up the facts required to prove them without additional guidance.

\begin{figure}[!htb]
\begin{boogie}
procedure reverse(a: [int]int, low, high: int) returns(b: [int]int)
   ensures (forall i: int :: 0 <= i && i < high - low ==> 
     seq(b, low, high)[i] == rev(a, low, high)[i]);
   ensures (forall i: int :: i < low ==> b[i] == a[i]);
   ensures (forall i: int :: high <= i ==> b[i] == a[i]);   
{
   var s, d, t: int;  // Variable (*t*) is temporary for swapping
   b := a;
   if (low >= high) { return; }  (*\label{cl:boogie:reverse:empty}*)
   s, d := low, high - 1;
   assert d == rp(s, low, high);
   while (s < d + 1)
   invariant low <= s && s <= d + 2 && d + 1 <= high;
   invariant d == rp(s, low, high);
   invariant 
    (forall i: int :: low <= i && i < s ==> a[i] == b[rp(i, low, high)]);
   invariant 
    (forall i: int :: d < i && i < high ==> a[i] == b[rp(i, low, high)]);
   invariant (forall i: int :: s <= i && i <= d ==> b[i] == a[i]);
   invariant (forall i: int :: i < low ==> b[i] == a[i]);
   invariant (forall i: int :: high <= i ==> b[i] == a[i]);
   {  // swap (*b[s]*) and (*b[d]*)
      t := b[d]; b[d] := b[s]; b[s] := t;
      s, d := s + 1, d - 1;
   }
}                                                         
\end{boogie}
\caption{Verified Boogie annotated implementation of the in-place reversal algorithm of Figure~\ref{fig:reversal-algo}.}
\label{fig:boogie:rotate}
\end{figure}

\subsubsection{Mechanized proofs of reversal and rotation by reversal}
We now have all the ingredients to present the implementation and proof of the main algorithms.

\paragraph{Proof of in-place reversal.}
Figure~\ref{fig:boogie:rotate} shows Boogie procedure \<reverse>,\footnote{Boogie's performance in this example is affected by the names given to local variables \<s> and \<d>, as well as the temporary local \<t>. For instance, using \<p> and \<q> as in the original pseudo-code listing triggers a time out. We could not investigate this bizarre (and somewhat distressing) issue in depth, but it probably has to do with how the SMT solver's instantiation rewrite order depends on Boogie's translation of variable names in the encoding of verification conditions.} which renders the pseudo-code implementation of Figure~\ref{fig:reversal-algo} using the same convention on input and output used for \<rotate> in Figure~\ref{fig:boogie:rotate-specification}.

Boogie can convert between the representation of \<b>'s content with respect to \<a>'s given by the loop invariants and the one used in the definition of \<rev>, and hence in \<reverse>'s postcondition.
When the loop terminates, the essential invariant characterizes the program state in a way that can be expressed as follows, for $0 \leq k < \@high@ - \@low@\colon$
\begin{equation}
\@seq(a, low, high)@[k] = \@seq(b, low, high)@[\@high@ - \@low@ - 1 - k]\,.
\label{eq:boogie:reverse:invariant-exit-sequence}
\end{equation}
Hence, Boogie verifies the following chain of equalities:
\begin{center}
\begin{boogie}[numbers=none]
   rev(a, low, high)[$k$] 
=  seq(a, low, high)[high - low - 1 - $k$] 
                         (definition of rev: axiom on line (*\ref{cl:axiom:rev-def}*))
=  seq(b, low, high)[high - low - 1 - (high - low - 1 - $k$)] 
                                                  (equation (*\ref{eq:boogie:reverse:invariant-exit-sequence}*))
=  seq(b, low, high)[$k$]                             (*(arithmetic)*)
   QED                                  (postcondition ensures).
\end{boogie}
\end{center}

The code in Figures~\ref{fig:boogie:rotate} and~\ref{fig:reversal-algo} is structurally very similar.
The only major---yet unsurprising---difference is that Boogie procedure \<reverse> includes information (in the postcondition and, correspondingly, in the loop invariants) about what is \emph{not} changed by the body: \<b> is the same as \<a> for indexes smaller than \<low> and greater than or equal to \<higher>.
This is a simple form of \emph{framing} necessary because Boogie's reasoning is modular: the effects of calls to \<reverse> within any of its callers are limited to what is explicit in \<reverse>'s specification irrespective of its implementation; anything that is not explicitly defined in \<reverse>'s postcondition may have changed.

Another, minor, difference between the Boogie code in Figures~\ref{fig:boogie:rotate} and the pseudo-code in Figure~\ref{fig:reversal-algo} is that the former's \<reverse> has no precondition, and simply returns the input \<a> when \<[low..high)> is an empty range of indexes.
In fact, the conditional \<return> statement on line~\ref{cl:boogie:reverse:empty} is actually not needed, since the following loop exits immediately when $\@low@ \geq \@high@$ (but note that the invariant \<s \<= d + 2> may fail initially if \<s \>= d + 1>).
Boogie needs a little nudge to understand how to handles this case separately: providing a conditional \<return> is one way to do it with code; ways to do it with annotations are inserting a precondition \<requires low <= high>, or making the failing invariant conditional, so that it holds vacuously when \<low \>= high>.

\begin{figure}[!tb]
\begin{boogie}
// Rotate (*a[0..N)*) to the left by (*r*) by reversal
procedure rotate_reverse(a: [int]int, N: int, r: int) 
          returns(b: [int]int)
  requires 0 < r && r < N;
  ensures (forall i: int :: 
    0 <= i && i < N ==> seq(b, 0, N)[i] == rot(a, 0, N, r)[i]);
{
   b := a;
   call b := reverse(b, 0, r);
   call b := reverse(b, r, N);
   \cforall lemma_rev_cat(a, 0, r, a, r, N, b, *);
   call b := reverse(b, 0, N);
   \cforall lemma_rot(a, 0, N, r, *);
}
\end{boogie}
\caption{Verified Boogie annotated implementation of the rotation by reversal algorithm of Figure~\ref{fig:rotation-reverse-algo}.}
\label{fig:boogie:rotation-reverse}
\end{figure}

\paragraph{Proof of rotation by reversal.}
As we can see in Figure~\ref{fig:boogie:rotation-reverse}, the Boogie annotated implementation of the rotation by reversal algorithm closely follows its presentation in Section~\ref{sec:reversal:proof}.
After reversing in-place \<b[0..r)> and then \<b[r..N)>, \<lemma_rev_cat> ensures that reversing  \<b[0..N)> again yields a rotation of \<a[0..N)> by \<r>.
The last call to \<reverse> performs this final reversal; and \<lemma_rot> helps convert between the index representation in \<lemma_rev_cat>'s postcondition and the one used in the definition of \<rot>.
Even if \<wrap> is directly used in neither, it is applicable to ``invert'' the former for \<sh - sl = N - r>.
This reasoning is similar to the argument at the end of Section~\ref{sec:boogie:copy} that should have become familiar by now.

\subsection{Rotation by swapping: organizing code for proofs} \label{sec:boogie:swap}

Mechanizing the rotation by swapping algorithm requires more complex usage of lemma procedures; and a careful organization of the imperative code to help guide the proof search so that it terminates in reasonable time.
The first step is, however, straightforward: verifying the auxiliary routine \<swap_sections>---which we discuss next.

\begin{figure}[!htb]
\begin{boogie}
procedure swap_sections(a: [int]int, low, high: int, d: int) 
          returns(b: [int]int)
 requires low <= low + d && low + d <= high - d && high - d <= high;
 ensures (forall i: int :: 0 <= i && i < d ==> 
   seq(b, low, high)[i] == seq(a, low, high)[i + (high - low - d)]);
 ensures (forall i: int :: d <= i && i < high - low - d ==> 
   seq(b, low, high)[i] == seq(a, low, high)[i]);
 ensures (forall i: int :: high - low - d <= i && i < high - low ==> 
   seq(b, low, high)[i] == seq(a, low, high)[i - (high - low - d)]);
 ensures (forall i: int :: i < low ==> b[i] == a[i]);  (*\label{cl:boogie:swapinplace-frame1}*)
 ensures (forall i: int :: high <= i ==> b[i] == a[i]); (*\label{cl:boogie:swapinplace-frame2}*)
{
 var x, z: int;
 var tmp: int;   // Temporary variable for swap
 b := a;
 x, z := low, high - d;
 while (x < low + d)
 invariant low <= x && x <= low + d;
 invariant high - d <= z && z <= high;
 invariant x - low == z - (high - d);
 invariant (forall i: int :: 0 <= i && i < x - low ==>
   seq(b, low, high)[i] == seq(a, low, high)[i + (high - low - d)]);
 invariant (forall i: int :: x - low <= i && i < high - low - d ==>
   seq(b, low, high)[i] == seq(a, low, high)[i]);
 invariant (forall i: int :: high - low - d <= i && i < z - low ==>
   seq(b, low, high)[i] == seq(a, low, high)[i - (high - low - d)]);
 invariant (forall i: int :: i < low ==> b[i] == a[i]);
 invariant (forall i: int :: z <= i ==> b[i] == a[i]);
 {  // swap (*b[x]*) and (*b[z]*)
    tmp := b[z]; b[z] := b[x]; b[x] := tmp;
    x, z := x + 1, z + 1;
 }
}
\end{boogie}
\caption{Verified Boogie annotated implementation of the in-place slice swapping algorithm of Figure~\ref{fig:swap-algo}.}
\label{fig:boogie:swap-algo}
\end{figure}

\subsubsection{Mechanized proof of swapping sections}
Procedure \<swap_sections> in Figure~\ref{fig:boogie:swap-algo} directly translates the pseudo-code in Figure~\ref{fig:swap-algo}.
The are only few, unsurprising differences:
\begin{itemize}
\item The Boogie procedure uses an output map since input arguments are read only; maps have infinite domains, and hence there is no need to require that $0 \leq \@low@$ and $\@high@ \leq \@N@$.
\item The postcondition (and correspondingly the loop invariants) has two additional clauses about framing on lines~\ref{cl:boogie:swapinplace-frame1} and \ref{cl:boogie:swapinplace-frame2}: the output \<b> is the same as the input \<a> for indexes outside the range \<[low..high)>.
\item A \<while> loop in Boogie renders the semantics of the \lstinline[language=PseudoEiffel]|until...do| loop in pseudo-code.
\end{itemize}

The correctness proof goes through without additional annotations.
In fact, unlike the case of \<reverse>, Boogie's translation seems much more robust with respect to inessential changes such as variable names or equivalent orderings of declarations.

\subsubsection{Lemmas about swapping}  \label{sec:boogie:lemma-procedures-swapping}
The most elaborate component for a Boogie proof of rotation by swapping is a translation of  Lemma~\ref{lm:rot-swap}, which is used to justify the recursive calls in the main algorithm.
We provide two distinct lemma procedures with symmetric structures, one for each of the \emph{left is smaller} case \eqref{eq:lemma-rot-swap-left} and \emph{right is smaller} case \eqref{eq:lemma-rot-swap-right}.

\begin{figure}[!tb]
\begin{boogie}
procedure lemma_left_smaller(a: [int]int, al, ah: int, 
                             c: [int]int, cl, ch: int, 
                             b: [int]int, bl, bh: int, d: int)
   requires al < ah;
   requires ah - al == bh - bl && ah - al == ch - cl;
   requires 0 < d && d < ah - al - d;
   requires (forall i: int :: ah - (d + al) <= i && i < ah - al ==>  (*\label{cl:lemma-left:p1}*)
     rot(a, al, ah, d)[i] == seq(b, bl, bh)[i]);
   requires (forall i: int :: 0 <= i && i < ah - (d + al) ==> (*\label{cl:lemma-left:p2}*)
     rot(c, cl, ch - d, d)[i] == seq(b, bl, bh - d)[i]);
   requires (forall i: int :: 0 <= i && i < d ==> (*\label{cl:lemma-left:p3}*)
     seq(c, cl, ch)[i] == seq(a, al, ah)[i + (ah - d - al)]);
   requires (forall i: int :: d <= i && i < ah - (d + al) ==> (*\label{cl:lemma-left:p4}*)
     seq(c, cl, ch)[i] == seq(a, al, ah)[i]);
   ensures (forall i: int :: 0 <= i && i < ah - al ==> (*\label{cl:lemma-left:post}*)
     rot(a, al, ah, d)[i] == seq(b, bl, bh)[i]);
{
   assert (forall i: int :: 0 <= i && i < ah - (d + al) ==> (*\label{cl:lemma-left:as1}*)
     rot(c, cl, ch - d, d)[i] == seq(b, bl, bh)[i]);
   assert (forall i: int :: 0 <= i && i < ch - cl - d - d ==> (*\label{cl:lemma-left:as2}*)
     rot(c, cl, ch - d, d)[i] == seq(c, cl, ch - d)[i + d]);
   assert (forall i: int :: 0 <= i && i < ch - cl - d - d ==> (*\label{cl:lemma-left:as3}*)
     rot(c, cl, ch - d, d)[i] == seq(c, cl, ch)[i + d]);
   assert (forall i: int :: ch - cl - d - d <= i && i < ch - cl - d ==> (*\label{cl:lemma-left:as4}*)
     rot(c, cl, ch - d, d)[i] == 
     seq(c, cl, ch - d)[i - (ch - cl - d - d)]);
   assert (forall i: int :: ch - cl - d - d <= i && i < ch - cl - d ==> (*\label{cl:lemma-left:as5}*)
     rot(c, cl, ch - d, d)[i] == 
     seq(c, cl, ch)[i - (ch - cl - d - d)]);
   assert (forall i: int :: 0 <= i && i < ah - (d + al) ==> (*\label{cl:lemma-left:as6}*)
     rot(a, al, ah, d)[i] == seq(b, bl, bh)[i]);
   assert (forall i: int :: ah - (d + al) <= i && i < ah - al ==> (*\label{cl:lemma-left:as7}*)
     rot(a, al, ah, d)[i] == seq(b, bl, bh)[i]);
}
\end{boogie}
\caption{Lemma~\ref{lm:rot-swap} for case \emph{left is smaller} \eqref{eq:lemma-rot-swap-left} as a Boogie lemma procedure, with $d = \@d@$, $N = \@ah@ - \@al@ = \@bh@ - \@bl@ = \@ch@ - \@cl@$,  $\@a[al..ah)@ = \cat{\xq}{\cat{\yq}{\zq}}$, $\@c[cl..ch - d)@ = \cat{\zq}{\yq}$, and $\@b[bl..bh)@ = \cat{\rot{(\@c[cl..ch - d)@)}{d}}{\xq}$.}
\label{fig:boogie:lemma-left-is-smaller}
\end{figure}

\paragraph{Lemma for \emph{left is smaller} case.}
Figure~\ref{fig:boogie:lemma-left-is-smaller} shows the Boogie translation of Lemma~\ref{lm:rot-swap} in case \eqref{eq:lemma-rot-swap-left}.
Recall how the lemma justifies the main algorithm, demonstrated in Figure~\ref{fig:rotation-by-swap-example}: to compute $\rot{(\cat{\xq}{\cat{\yq}{\zq}})}{d}$, first swap equal-length sequences $\xq$ and $\zq$, and then recur on $\cat{\zq}{\yq}$.

Procedure \<lemma_left_smaller> traces these two macro steps through an additional input map \<c>.
Then, \<a[al..ah)> represents the input consisting of $\cat{\xq}{\cat{\yq}{\zq}}$, with \<d> the length of both leftmost $\xq$ and rightmost $\zq$ segments, as in the top-left picture of Figure~\ref{fig:rotation-by-swap-example}.
Slice \<c[cl..ch - d)> represents $\cat{\zq}{\yq}$: the initial part of the processed array after swapping $\xq$ and $\zq$, as in the top-right picture of Figure~\ref{fig:rotation-by-swap-example}.
Slice \<b[bl..bh)> represents the final output after recursively rotating \<c[cl..ch - d)> by \<d>, as in the bottom-right picture of Figure~\ref{fig:rotation-by-swap-example}.
The preconditions of \<lemma_left_smaller> encode these assumptions: line~\ref{cl:lemma-left:p1} describes the right-most slice \<b[bh - d..bh)> as $\xq$, which is in place in the rotation of \<a[al..ah)>; line~\ref{cl:lemma-left:p2} describes the other slice \<b[bl..bh - d)> as $\rot{(\@c[cl..ch - d)@)}{\@d@}$; lines~\ref{cl:lemma-left:p3} and~\ref{cl:lemma-left:p3} respectively describe \<c[cl..cl + d)> as $\zq$ and \<c[cl + d..ch - d)> as $\yq$.
The postcondition on line~\ref{cl:lemma-left:post} concludes that the \<b[bl..bh)> described in the precondition is indeed the rotation of \<a[al..ah)> by \<d>.

The procedure body consists of a series of seven \<assert> that guide Boogie through the proof of the postcondition from the preconditions.
As usual, there is room for variations, but this particular sequence of \<assert> is fairly natural and produces a fast proof; to illustrate, this is an informal explanation of what each \<assert> establishes:
\begin{description}
\item[Line~\ref{cl:lemma-left:as1}] relaxes the right bound of \<seq(b)> in the precondition on line~\ref{cl:lemma-left:p2} from \<bh - d> to \<bh>, since indexes beyond \<bh - d> are out of the quantification range.
\item[Line~\ref{cl:lemma-left:as2}] recalls the definition of \<rot> for \<c[cl + d..ch - d - d)> or $\yq$.
\item[Line~\ref{cl:lemma-left:as3}] relaxes the right bound of \<seq(c)> in the previous \<assert> from \<ch - d> to \<ch>, since indexes beyond \<ch - d> are out of the quantification range.
\item[Line~\ref{cl:lemma-left:as4}] recalls the definition of \<rot> for \<c[cl..cl + d)> or $\zq$.
\item[Line~\ref{cl:lemma-left:as5}] relaxes the right bound of \<seq(c)> in the previous \<assert> from \<ch - d> to \<ch>, since indexes beyond \<ch - d> are out of the quantification range.
\item[Line~\ref{cl:lemma-left:as6}] concludes that \<b[bl..bh - d)> is $\rot{(\@a[al..ah))@}{\@d@}$ between \<[al..ah - d)>, using the facts about \<c> proved so far (specifically, lines~\ref{cl:lemma-left:as3} and \ref{cl:lemma-left:as5}), and the relations between \<b> and \<c> and between \<b> and \<a> in the preconditions.
\item[Line~\ref{cl:lemma-left:as7}] recalls that \<b[bh - d..bh)> coincides with $\rot{(\@a[al..ah))@}{\@d@}[\@ah - d..ah@)$; even if this \<assert> is the very same as the precondition on line~\ref{cl:lemma-left:p1}, it is necessary to recall it explicitly in the body so that Boogie uses it to close the proof.
\end{description}

\begin{figure}[!bt]
\begin{boogie}
procedure lemma_right_smaller(a: [int]int, al, ah: int, 
                              c: [int]int, cl, ch: int, 
                              b: [int]int, bl, bh: int, d: int)
 requires al < ah;
 requires ah - al == bh - bl && ah - al == ch - cl;
 requires 0 < d && d < ah - al - d;
 requires (forall i: int :: 0 <= i && i < d ==> (*\label{cl:lemma-right:p1}*)
   rot(a, al, ah, ah - al - d)[i] == seq(b, bl, bh)[i]);
 requires (forall i: int :: 0 <= i && i < ah - (d + al) ==> (*\label{cl:lemma-right:p2}*)
   rot(c, cl + d, ch, ah - al - d - d)[i] == seq(b, bl + d, bh)[i]);
 requires (forall i: int :: ah - (d + al) <= i && i < ah - al ==> (*\label{cl:lemma-right:p3}*)
   seq(c, cl, ch)[i] == seq(a, al, ah)[i - (ah - d - al)]);
 requires (forall i: int :: d <= i && i < ah - (d + al) ==> (*\label{cl:lemma-right:p4}*)
   seq(c, cl, ch)[i] == seq(a, al, ah)[i]);
 ensures (forall i: int :: 0 <= i && i < ah - al ==> (*\label{cl:lemma-right:post}*)
   rot(a, al, ah, ah - al - d)[i] == seq(b, bl, bh)[i]);
{
 assert (forall i: int :: 0 <= i && i < ah - (d + al) ==> (*\label{cl:lemma-right:as1}*)
   rot(c, cl + d, ch, ah - al - d - d)[i] == seq(b, bl, bh)[i + d]);
 assert (forall i: int :: d <= i && i < ch - cl - d ==>  (*\label{cl:lemma-right:as2}*)
   rot(c, cl + d, ch, ah - al - d - d)[i] == 
   seq(c, cl + d, ch)[i - d]);
 assert (forall i: int :: 0 <= i && i < ch - cl - d - d ==> (*\label{cl:lemma-right:as3}*)
   rot(c, cl + d, ch, ah - al - d - d)[i + d] == 
   seq(c, cl + d, ch)[i]);
 assert (forall i: int :: d <= i && i < ch - cl - d ==> (*\label{cl:lemma-right:as4}*)
   rot(c, cl + d, ch, ah - al - d - d)[i] == seq(c, cl, ch)[i]);
 assert (forall i: int :: 0 <= i && i < d ==> (*\label{cl:lemma-right:as5}*)
   rot(c, cl + d, ch, ah - al - d - d)[i] == 
   seq(c, cl + d, ch)[i + (ch - cl - d - d)]);
 assert (forall i: int :: 0 <= i && i < d ==> (*\label{cl:lemma-right:as6}*)
   rot(c, cl + d, ch, ah - al - d - d)[i] == 
   seq(c, cl, ch)[i + (ch - cl - d)]);
 assert (forall i: int :: d <= i && i < ah - al ==> (*\label{cl:lemma-right:as7}*)
   rot(a, al, ah, ah - al - d)[i] == seq(b, bl, bh)[i]);
 assert (forall i: int :: { seq(b, bl, bh)[i] }  // trigger  (*\label{cl:lemma-right:as8}*)
   0 <= i && i < d ==>
   rot(a, al, ah, ah - al - d)[i] == seq(b, bl, bh)[i]);
}
\end{boogie}
\caption{Lemma~\ref{lm:rot-swap} for case \emph{right is smaller} \eqref{eq:lemma-rot-swap-right} as a Boogie lemma procedure, with $d = \@d@$, $N = \@ah@ - \@al@ = \@bh@ - \@bl@ = \@ch@ - \@cl@$,  $\@a[al..ah)@ = \cat{\xq}{\cat{\yq}{\zq}}$, $\@c[cl + d..ch)@ = \cat{\yq}{\xq}$, and $\@b[bl..bh)@ = \cat{\zq}{\rot{(\@c[cl + d..ch)@)}{N - 2d}}$.}
\label{fig:boogie:lemma-right-is-smaller}
\end{figure}

Once we have understood the rationale behind \<lemma_left_smaller>, it is not difficult to derive the dual \<lemma_right_smaller> shown in Figure~\ref{fig:boogie:lemma-right-is-smaller} and corresponding to case~\eqref{eq:lemma-rot-swap-right} of Lemma~\ref{lm:rot-swap}.
Following the example of Figure~\ref{fig:rotation-by-swap-example-right} going from \<a> (top-left picture) through \<c> (top-right picture) to \<b> (bottom-right picture) helps understand the lemma procedure.
Compared to \<lemma_left_smaller>, there now is one more \<assert> due to an additional index rescaling (the second \<assert> refers to \<i - d>, which becomes \<i> in the third \<assert>).

The real twist is, however, the need for a \emph{trigger} in the last \<assert> on line~\ref{cl:lemma-right:as8}:
\[
\@\{ seq(b, bl, bh)[i] \}@\,.
\]
Even if the \<assert> is just a repetition of the precondition on line~\ref{cl:lemma-right:p1}, Boogie needs help to pick the relevant facts among the many instantiated terms that are available.
The trigger directs the SMT solver\footnote{The description of the SMT solver Simplify~\cite{Simplify} discusses how triggers work; see also~\cite{RR-triggers,ALR-computingSMT} and \cite[Sec.~11.2]{BoogieManual}.} to only instantiate the universal quantifier in the \<assert> for those \<i>'s such that \<seq(b, bl, bh)[i]> is a term in the current proof context.
In this particular case, using the trigger makes a dramatic difference in terms of performance when proving the whole lemma procedure.

\begin{figure}[!htb]
\begin{boogie}
procedure rotate_swap(a: [int]int, N: int, r: int) 
          returns(b: [int]int)
   requires 0 < r && r < N;
   ensures (forall i: int :: 
     0 <= i && i < N ==> seq(b, 0, N)[i] == rot(a, 0, N, r)[i]);
{
   call b := rotate_swap_helper(a, r, 0, N);
}

procedure rotate_swap_helper(a: [int]int, p: int, low, high: int) 
   returns(b: [int]int)
   requires low <= p && p < high;
   ensures (forall i: int :: 0 <= i && i < high - low ==>
     rot(a, low, high, p - low)[i] == seq(b, low, high)[i]);
   ensures (forall i: int :: i < low ==> b[i] == a[i]);
   ensures (forall i: int :: high <= i ==> b[i] == a[i]);
{
   var c: [int]int;  // ghost: value of b before recursive call
   if (p == low) { b := a; return; }
   if (p $-$ low $\leq$ high $-$ p) {
      // swap (*a[low..p)*) and (*a[high $-$ (p $-$ low)..high)*)
      call b := swap_sections(a, low, high, p - low);  (*\label{cl:boogie:swaphelper:swap-left}*)
      if (p - low == high - p) {
         // now the whole (*b[low..high)*) is in place
         return;
      } else {
         // now (*b[high $-$ (p $-$ low)..high)*) is in place
         c := b;   // ghost
         call b := rotate_swap_helper(b, p, low, high - (p - low));
         call lemma_left_smaller(a, low, high, c, low, high, 
                                 b, low, high, p - low);
      }
   } else { 
      assert p - low > high - p;
      assert 0 <= high - p && high - p <= high - low;
      // swap (*a[low..low $+$ (high $-$ p))*) and (*a[p..high)*)
      call b := swap_sections(a, low, high, high - p);
      // now (*b[low..low $+$ (high $-$ p))*) is in place
      c := b;   // ghost
      call b := rotate_swap_helper(b, p, low + (high - p), high);
      call lemma_right_smaller(a, low, high, c, low, high, 
                               b, low, high, high - p);
   }
}
\end{boogie}
\caption{Verified Boogie annotated implementation of the rotation by swapping recursive algorithm of Figure~\ref{fig:rotation-swap-algo}.}
\label{fig:boogie:rotation-swap}
\end{figure}

\subsubsection{Mechanized proof of rotation by swapping: recursive version} \label{sec:swap:proof:recursive}
Presenting annotated versions of \<rotate_swap> and \<rotate_swap_helper>, Figure~\ref{fig:boogie:rotation-swap} is the Boogie counterpart to Figure~\ref{fig:rotation-swap-algo}.

With respect to the pseudo-code version of Figure~\ref{fig:rotation-swap-algo}, \<rotate_swap_helper> in Boogie has some structural differences that are worth discussing.
The most pronounced one is a different conditional structures.
The pseudo-code algorithm clearly distinguishes between three cases (equal length of slices to be swapped, left is smaller, right is smaller), and each case has a call to \<swap_sections> followed, in the last two cases, by a recursive call to the helper; the trivial base case \lstinline[language=PseudoEiffel]|low = p| 
is handled by an enclosing \lstinline[language=PseudoEiffel]|if|.
By contrast, the Boogie procedure handles the trivial base case initially introducing abrupt termination (i.e., a \<return>).
Then, the call to \<swap_sections> on line~\ref{cl:boogie:swaphelper:swap-left} applies to two cases: ``equal length'' and ``left is smaller''.
This structure helps reduce repetitions in reasoning along different conditional branches, and in fact it makes for quicker verification.
Having one fewer call to \<swap_sections> with respect to the pseudo-code version avoids checking \<swap_sections>'s precondition twice with the same arguments in different contexts; and the \<return> in the ``equal length'' case drives a direct proof of the helper's postcondition from \<swap_sections>'s postcondition and the few other facts available at that location, instead of having to consider many other inapplicable facts in a conditional reasoning at the unique exit point of the structured pseudo-code version.
Of course, other solutions are possible in Boogie with some trial and error, but it should be clear that the two versions are semantically equivalent.
To help unravel the branching structure with more clarity, we have two \<assert> in the ``right is smaller'' branch; they also are crucial for performance.

The usage of a \emph{ghost} variable \<c> is another novelty of Figure~\ref{fig:boogie:rotation-swap} compared to the previous Boogie examples.
It is no coincidence that the name \<c> is also used for one argument of the lemma procedures presented in Section~\ref{sec:boogie:lemma-procedures-swapping}.
In \<rotate_swap_helper>, \<c> keeps track of the value of \<b> after the first macro-step (call to \<swap_sections>) and before the second one (recursive call to the helper).
Thanks to \<c>, we conclude the proof of each recursive case by calling the corresponding lemma procedure, which relates the input \<a> to the final output \<b> through \<c> to establish \<rotate_swap_helper>'s postcondition.
As discussed in the upcoming Section~\ref{sec:boogie:modulo}, the rotation by modular visit algorithm contemplates a much richer usage of ghost code, but the idea is already clear here: ghost code keeps track of program state beyond what is explicit in the non-ghost program variables (that is, variables used in the actual computation), capturing information that is readily useful for proofs.

A final aspect of modularization leveraged in the proof of rotation by swapping is not apparent in the presentation on paper.
We split the proof of the procedures in separate files. 
Each file contains only one procedure with implementation (for example, \<rotate_swap_helper>) together with only the signature and specification of other procedures called in the single implementation (for example, \<lemma_left_smaller>, \<lemma_right_smaller>, and \<swap_sections>).
We invoke Boogie separately on each file.
Even if Boogie works modularly, there clearly is interference between different proofs originating in the same file; having only one procedure to prove per invocation significantly reduces the possible problems---ultimately causing slower proofs or timeouts due to unfruitful proof search heuristics being applied.
The bottom line is that how code and annotations are structured can make a significant different when mechanizing verification of algorithms.

\subsubsection{Mechanized proof of rotation by swapping: iterative version}

A careful organization of code and annotations is also central to the proof of the \emph{iterative} version of rotation by swapping.
A formalization of Lemma~\ref{lm:rot-swap} is still at the core of the correctness argument; but we now proceed using a different approach than in the recursive version: since lemmas and imperative code are both encoded as \<procedure>s in Boogie, we combine them in the same procedure.

%% After setting the current page, flush all pending floats
\afterpage{\clearpage}

To this end, we introduce three variants of \<swap_sections>, one for each of the by-now familiar cases: ``equal length'' sections, ``left is smaller'', and ``right is smaller''.
We name the three variants accordingly: \<swap_equal>, \<swap_left>, and \<swap_right>.
The operational part of the variants is identical, and simply consists of a suitable call to \<swap_sections> of Figure~\ref{fig:boogie:swap-algo}.
What is different is their \emph{specification}: besides describing output in terms of input, it also relates the output to the original reversal problem as per Lemma~\ref{lm:rot-swap}.
Take for example \<swap_left> in Figure~\ref{fig:boogie:swap-left-withlemma}, which swaps \<c[l..p)> and \<c[h - (p - l)..h)> under the assumption \<p - l \< h - p>.
Its input arguments also include the original input \<a[low..high)> to be rotated.
Its precondition assumes that \<c[low..l)> and \<c[h..high)> correspond to already rotated slices of \<a[low..high)>.
Its postcondition guarantees that output \<b[low..l)> and \<b[h - (p - l)..high)> will consist of rotated slices of \<a[low..high)>, thus ensuring progress; and that rotating \<b[l..h - (p - l))> by \<p - l> will complete the rotation of \<a[low..high)>.
Of course, the names \<a>, \<b>, and \<c> correspond to the three macro-step also underlying the recursive version and the running example in Figure~\ref{fig:rotation-by-swap-example}.
Similar comments apply to the augmented specification of \<swap_equal> and \<swap_right> shown in Figure~\ref{fig:boogie:swap-equal-right-withlemma}.

\begin{figure}[!htb]
\begin{boogie}
procedure swap_left(a: [int]int, c: [int]int, low, high: int, 
                    l, h: int, p: int) returns(b: [int]int)
   requires low <= l && l < p && p < h && h <= high;
   requires p - l < h - p;      // left is smaller
   requires (forall i: int :: 0 <= i && i < l - low ==>
     rot(a, low, high, p - low)[i] == seq(c, low, high)[i]);
   requires (forall i: int :: l - low <= i && i < h - low ==> 
     rot(a, low, high, p - low)[i] == 
     rot(c, l, h, p - l)[i - (l - low)]);
   requires (forall i: int :: h - low <= i && i < high - low ==>
     rot(a, low, high, p - low)[i] == seq(c, low, high)[i]);
   requires (forall i: int :: i < low ==> c[i] == a[i]);
   requires (forall i: int :: high <= i ==> c[i] == a[i]);
   ensures (forall i: int :: 0 <= i && i < l - low ==>
     rot(a, low, high, p - low)[i] == seq(b, low, high)[i]);
   ensures (forall i: int :: l - low <= i && i < h - low - (p - l) ==> 
     rot(a, low, high, p - low)[i] == 
     rot(b, l, h - (p - l), p - l)[i - (l - low)]);
   ensures (forall i: int :: h - low - (p - l) <= i && i < high - low ==>
     rot(a, low, high, p - low)[i] == seq(b, low, high)[i]);
   ensures (forall i: int :: i < low ==> b[i] == a[i]);
   ensures (forall i: int :: high <= i ==> b[i] == a[i]);
{
   call b := swap_sections(c, l, h, p - l);
   // Asserts to prove post from pre and swap_sections's post
   // ...
} 
\end{boogie}
\caption{Verified Boogie annotated implementation of the in-place slice swapping algorithm Figure~\ref{fig:swap-algo}: postcondition augmented with Lemma~\ref{lm:rot-swap} for case \emph{left is smaller}~\eqref{eq:lemma-rot-swap-left}.}
\label{fig:boogie:swap-left-withlemma}
\end{figure}

\begin{figure}[!p]
\begin{boogie}
procedure swap_equal(a: [int]int, c: [int]int, low, high: int, 
                     l, h: int, p: int) returns(b: [int]int);
   requires low <= l && l < p && p < h && h <= high;
   requires p - l == h - p;     // left same size as right
   requires (forall i: int :: 0 <= i && i < l - low ==>
     rot(a, low, high, p - low)[i] == seq(c, low, high)[i]);
   requires (forall i: int :: l - low <= i && i < h - low ==> 
     rot(a, low, high, p - low)[i] == 
     rot(c, l, h, p - l)[i - (l - low)]);
   requires (forall i: int :: h - low <= i && i < high - low ==>
     rot(a, low, high, p - low)[i] == seq(c, low, high)[i]);
   requires (forall i: int :: i < low ==> c[i] == a[i]);
   requires (forall i: int :: high <= i ==> c[i] == a[i]);
   ensures (forall i: int :: 0 <= i && i < p - low ==>
     rot(a, low, high, p - low)[i] == seq(b, low, high)[i]);
   ensures (forall i: int :: p - low <= i && i < high - low ==>
     rot(a, low, high, p - low)[i] == seq(b, low, high)[i]);
   ensures (forall i: int :: i < low ==> b[i] == a[i]);
   ensures (forall i: int :: high <= i ==> b[i] == a[i]);

procedure swap_right(a: [int]int, c: [int]int, low, high: int, 
                     l, h: int, p: int) returns(b: [int]int);
   requires low <= l && l < p && p < h && h <= high;
   requires p - l > h - p;  // right is smaller
   requires (forall i: int :: 0 <= i && i < l - low ==>
     rot(a, low, high, p - low)[i] == seq(c, low, high)[i]);
   requires (forall i: int :: h - low <= i && i < high - low ==>
     rot(a, low, high, p - low)[i] == seq(c, low, high)[i]);
   requires (forall i: int :: l - low <= i && i < h - low ==> 
     rot(a, low, high, p - low)[i] == 
     rot(c, l, h, p - l)[i - (l - low)]); 
   requires (forall i: int :: i < low ==> c[i] == a[i]);
   requires (forall i: int :: high <= i ==> c[i] == a[i]);
   ensures (forall i: int :: 0 <= i && i < l - low + (h - p) ==>
     rot(a, low, high, p - low)[i] == seq(b, low, high)[i]);
   ensures (forall i: int :: l - low + (h - p) <= i && i < h - low ==> 
     rot(a, low, high, p - low)[i] == 
     rot(b, l + (h - p), h, p - (l + (h - p)))[i - (l - low + h - p)]);
   ensures (forall i: int :: h - low <= i && i < high - low ==>
     rot(a, low, high, p - low)[i] == seq(b, low, high)[i]);
   ensures (forall i: int :: i < low ==> b[i] == a[i]);
   ensures (forall i: int :: high <= i ==> b[i] == a[i]);
\end{boogie}
\caption{Boogie specifications of the in-place slice swapping algorithm Figure~\ref{fig:swap-algo}. The postcondition of \<swap_equal> is augmented with the property that, when \<p - l = h - p>, swapping \<c[l..p)> and \<c[p..h)> in place is tantamount to rotating \<c[l..h)> by \<p - l>. The postcondition of \<swap_right> is augmented with Lemma~\ref{lm:rot-swap} for case \emph{right is smaller}~\eqref{eq:lemma-rot-swap-right}.}
\label{fig:boogie:swap-equal-right-withlemma}
\end{figure}

\begin{figure}[!htb]
\begin{boogie}
procedure rotate_swap_iterative(a: [int]int, N: int, r: int) 
          returns(b: [int]int)
   requires 0 < r && r < N;
   ensures (forall i: int :: 
     0 <= i && i < N ==> seq(b, 0, N)[i] == rot(a, 0, N, r)[i]);
{
   var low, p, high: int;
   low, p, high := 0, r, N;
   b := a;
   while (low < p && p < high)
   invariant 0 <= low && low <= p && p <= high && high <= N;
   invariant low == p <==> p == high;
   invariant (forall i: int :: 0 <= i && i < low ==>
              rot(a, 0, N, p)[i] == seq(b, 0, N)[i]);
   invariant p - low < high - low ==>
             (forall i: int :: low <= i && i < high ==> 
              rot(a, 0, N, p)[i] == 
              rot(b, low, high, p - low)[i - low]);
   invariant (forall i: int :: high <= i && i < N ==>
              rot(a, 0, N, p)[i] == seq(b, 0, N)[i]);
   invariant (forall i: int :: i < 0 ==> b[i] == a[i]);
   invariant (forall i: int :: N <= i ==> b[i] == a[i]);
   {
      goto equal_length, left_smaller, right_smaller;
      equal_length:
         assume p - low == high - p;
         call b := swap_equal(a, b, 0, N, low, high, p);
         low, high := low + (p - low), high - (high - p);
         goto continue;
      left_smaller:
         assume p - low < high - p;
         call b := swap_left(a, b, 0, N, low, high, p);
         high := high - (p - low);
         goto continue;
      right_smaller:
         assume p - low > high - p;
         call b := swap_right(a, b, 0, N, low, high, p);
         low := low + (high - p);
         goto continue;
      continue:
   }
}
\end{boogie}
\caption{Verified Boogie annotated implementation of the rotation by swapping iterative algorithm of Figure~\ref{fig:rotation-swap-algo-iter}.}
\label{fig:boogie:rotation-swap-iter}
\end{figure}

The advantage of this approach is that we can reason about special properties of swapping separately in each case.
The call to \<swap_sections> in the bodies of \<swap_left>, \<swap_right>, and \<swap_equal> is followed by a sequence of \<assert> that proves the special properties of the swapping declared in the augmented postconditions.
For brevity, we omit the proofs; suffice it to say that \<swap_equal> has a simple proof, whereas \<swap_left> and \<swap_right>'s proofs are quite involved and require elaborate assertions and careful usage of triggers.

With this organization, Boogie can prove \<rotate_swap_iterative> in Figure~\ref{fig:boogie:rotation-swap-iter} with the same invariants as the pseudo code in Figure~\ref{fig:rotation-swap-algo-iter} without additional annotations (the only exception being the straightforward framing invariants to keep track of the unchanged parts of the map domain before \<0> and after \<N>).
While we could have used nested \<if>s to replicate the three-case structure in the loop body of Figure~\ref{fig:rotation-swap-algo-iter}, we demonstrate another construct, nondeterministic \<goto>, which emphasizes the three-way case split.
Embedding the proof of Lemma~\ref{lm:rot-swap} in separate procedures makes for a simple and efficient high-level proof that reflects the argument on paper.

\subsection{Rotation by modular visit: ghost code and framing} \label{sec:boogie:modulo}

Underlying the proof of the rotation by modular visit algorithm discussed informally 
in Section~\ref{sec:modulo:proof} were properties of modular arithmetic and cyclic decompositions 
of permutations.
Mechanizing the proofs of those properties all the way down to fundamental arithmetic would be exceedingly complicated and out of the scope of the present discussion; instead, we capture the fundamental mathematical properties as axioms whose correctness is intuitively clear, and build the main correctness proof atop them.

This approach has the additional advantage that it lets us focus on other aspects central to mechanizing the proof of rotation by modular visit, and in particular on keeping track of implicit information in the program state by means of \emph{ghost code}.
We already encountered ghost code among the annotations of rotation by swapping (Section~\ref{sec:swap:proof:recursive}), but proving rotation by modular visit will require more complex usage, especially to detail \emph{framing} of the result array \<b>.

\subsubsection{Axioms about cycles}

Figure~\ref{fig:boogie:mp-gcd-ic-def} shows declarations and axiomatic definitions of three fundamental quantities featuring in the proof of rotation by modular visit: $\mpos{N}{M}{s}{k}$, $\gcd(N, M)$, and $\ic{N}{M}$ corresponding to $\@mp@(N, M, s, k)$, $\@gcd@(N, M)$, and $\@ic@(N, M)$ in Boogie.

\begin{figure}[!htb]
\begin{boogie}
function mp(N: int, M: int, s: int, p: int) returns(int);
axiom (forall N: int, M: int, s: int :: 
  0 < M && M < N && 0 <= s && s < N ==> mp(N, M, s, 0) == s);
axiom (forall N: int, M: int, s: int, k: int :: 
  0 < M && M < N && 0 <= s && s < N && 0 < k ==> 
  mp(N, M, s, k) == wrap(mp(N, M, s, k - 1) + M, N));

function gcd(N: int, M: int) returns(int);
axiom (forall N, M: int :: 0 < N && 0 < M ==>   (*\label{ln:ax:gcd-0}*)
  0 < gcd(N, M) && gcd(N, M) <= N && gcd(N, M) <= M);

function ic(N: int, M: int) returns(int);  
axiom (forall N: int, M: int ::  (*\label{ln:ax:ic-0}*)
  0 <= M && M < N ==> 0 < ic(N, M) && ic(N, M) <= N);

axiom (forall N: int, M: int :: 0 < M && M < N ==> (*\label{ln:ax:gcd-ic-N}*)
  gcd(N, M) * ic(N, M) == N);
axiom (forall N: int, M: int, s: int :: mp(N, M, s, ic(N, M)) == s); (*\label{ln:ax:mp-ic-start}*)
axiom (forall N: int, M: int, s: int, p: int :: 
  0 < p && p < ic(N, M) ==> mp(N, M, s, p) != s);
axiom (forall N: int, M: int, s: int, p, q: int :: (*\label{ax:mp-ic-1}*)
  0 <= p && p < ic(N, M) && 0 <= q && q < ic(N, M) && p != q ==> 
  mp(N, M, s, p) != mp(N, M, s, q));  (*\label{ln:ax:mp-ic-end}*)
axiom (forall N: int, M: int, s, t: int, p, q: int :: (*\label{ax:mp-ic-2}*)
  0 <= p && p < ic(N, M) && 0 <= q && q < ic(N, M) && 
  0 <= s && s < t && t < s + gcd(N, M) && t < N ==> 
  mp(N, M, s, p) != mp(N, M, t, q));
\end{boogie}
\caption{Boogie declarations and axiomatic definitions of \<mp>, \<gcd>, and \<ic>.}
\label{fig:boogie:mp-gcd-ic-def}
\end{figure}

The first two axioms characterize $\mpos{N}{M}{s}{k}$---which gives the $k$-th index in a cycle starting at $s$ with step $M$ wrapping over $N$---inductively as
\[
\mpos{N}{M}{s}{k} = 
\begin{cases}
s  &  k = 0 \,, \\
\wrap{(M + \mpos{N}{M}{s}{k - 1})}{N}
& k > 0 \,.
\end{cases}
\]
This definition and the one in~\eqref{eq:mp-def} are equivalent (a fact which could be proved from a suitable axiomatization of modular arithmetic), but the inductive definition has the advantage of directly matching the program's logic: each iteration of the inner loop moves \<v> to the ``next'' value in the modular visit.
In contrast, \eqref{eq:mp-def} is inductive only indirectly through definition \eqref{eq:wrap-def} of `$\wrap{}{}$'.

The following two axioms, lines \ref{ln:ax:gcd-0} and \ref{ln:ax:ic-0} in Figure~\ref{fig:boogie:mp-gcd-ic-def}, define how $\gcd$ and $\ic{}{}$ are bounded by their arguments.

The remaining axioms in Figure~\ref{fig:boogie:mp-gcd-ic-def} complete the characterization of $\mpos{}{}{}{}$, $\gcd$, and $\ic{}{}$ in terms of mutual properties.
The axiom on line~\ref{ln:ax:gcd-ic-N} is equivalent to \eqref{eq:ic-def}; in the proof, it is necessary to conclude that the inner and outer loops combined visit all the \<N> elements of the input.
The three axioms on lines~\ref{ln:ax:mp-ic-start}--\ref{ln:ax:mp-ic-end} define the \<ic(N, M)> elements in the same cycle starting at a generic \<s>; in the proof, these characterize the elements visited by the inner loop (executed to completion for a given value of \<s>).
By contrast, the last axiom in Figure~\ref{fig:boogie:mp-gcd-ic-def} declares disjointness between elements of the cycles a rotation can be decomposed into; in the proof, it is necessary to combine the effect of each iteration of the outer loop (in fact, the axiom covers exactly \<ic(N, M)> different consecutive values of \<s>).

% It should be clear that the axioms in Figure~\ref{fig:boogie:mp-gcd-ic-def} capture fundamental properties of $\mpos{}{}{}{}$, $\gcd$, and $\ic{}{}$ that follow from their definitions and were implicit in the presentation of Section~\ref{sec:modulo:proof}.
% The mechanized proof builds upon this level of abstraction.

\begin{figure}[!htb]
\begin{boogie}
procedure rotate_modulo(a: [int]int, N: int, r: int) 
          returns(b: [int]int)
  requires 0 < r && r < N;
  ensures (forall i: int :: 
    0 <= i && i < N ==> seq(b, 0, N)[i] == rot(a, 0, N, r)[i]);
{
   var start, v, displaced: int;
   // ghost:
   var k: int;  // index (*mp(N, N $-$ r, start, k)*) currently visited
   var set: [int]bool; // (*set[i]*) iff (*b[i]*) has been assigned to

   b := a;
   assume (forall i: int :: !set[i]);  // ghost: initialize (*b*)  (*\label{ln:boogie:modulo:outer:assume}*)

   start := 0;
   while (start < gcd(N, N - r))  
   invariant (0 <= start && start <= gcd(N, N - r));
   invariant (forall i: int ::   (*\label{ln:boogie:modulo:outer:framing-same}*)
     0 <= i && i < N && !set[i] ==> b[i] == a[i]);
   invariant (forall i: int, s: int ::  (*\label{ln:boogie:modulo:outer:framing-notset}*)
     0 <= i && i < ic(N, N - r) && start <= s && s < gcd(N, N - r) 
     ==> !set[mp(N, N - r, s, i)]);
   invariant (forall i: int, s: int :: (*\label{ln:boogie:modulo:outer:framing-essential}*)
     0 <= i && i < ic(N, N - r) && 0 <= s && s < start 
     ==> set[mp(N, N - r, s, i)]);
   invariant (forall i: int :: 0 <= i && i < N && set[i] ==> (*\label{ln:boogie:modulo:outer:essential}*)
     seq(b, 0, N)[i] == rot(a, 0, N, r)[i]);
   {
      v, displaced := start, b[start];
      k := 0; // ghost

      // one unconditional iteration of the inner loop
      k := k + 1; // ghost
      v := v + N - r;
      if (v >= N) { v := v - N; }
      b[v], displaced := displaced, b[v];
      set[v] := true;  // ghost

      // Inner loop here: see (*Figure~\ref{fig:boogie:rotation-modulo-inner}*)

      assert k == ic(N, N - r);
      start := start + 1;
   }
   // Concluding assertions here: see (*Figure~\ref{fig:boogie:rotation-modulo-closing}*)
}
\end{boogie}
\caption{Verified Boogie annotated implementation of the rotation by modular visit algorithm of Figure~\ref{fig:rotation-by-modulo-algo}.}
\label{fig:boogie:rotation-modulo-outer}
\end{figure}

\subsubsection{Outer loop}

The Boogie version of the algorithm is a mouthful;
we begin looking at the annotated outer \<while> loop, whose Boogie version is shown in Figure~\ref{fig:boogie:rotation-modulo-outer} (with references to the parts presented later).

%% After setting the current page, flush all pending floats
\afterpage{\clearpage}

\paragraph{Framing using a ghost map.}
A fundamental difficulty we encounter trying to translate the annotated algorithm of Figure~\ref{fig:rotation-by-modulo-algo} into Boogie is the lack of readily available framing annotations.
To prove that the essential outer loop invariant \eqref{eq:modular:outer:essential} is inductive, we have to establish that each new iteration works on new elements of \<b> or, equivalently, that it does not touch the elements set by previous iterations.
This non-interference property ultimately boils down to the fact that the cycle visited by each outer loop iteration is disjoint from the other cycles.
To put this fact to use in the mechanized proof, we introduce ghost state that keeps track precisely of the visited locations.
The axioms in Figure~\ref{fig:boogie:mp-gcd-ic-def} can then be used to prove that the ghost state changes following invariants that reflect progress as in the original loop invariant \eqref{eq:modular:outer:essential}.

Concretely, we introduce a Boolean map \<set> as ghost state: $\@set@[k]$ is \<true> iff the imperative code has changed the value of $\@b@[k]$ from its initial input value $\@a@[k]$ to its correct value in the rotation underway.
This convention makes it possible to decouple framing (``what elements the algorithm modifies'') from functional properties (``how the algorithm modifies the elements''), which simplifies the life of the theorem prover by bringing the annotations closer in form to the axioms used to verify them, and hence also simplifies the task of checking each of them individually.

The invariant on line~\ref{ln:boogie:modulo:outer:essential} in Figure~\ref{fig:boogie:rotation-modulo-outer} restates the essential outer loop invariant \eqref{eq:modular:outer:essential} in terms if \<set>: if \<set[i]> is \<true> then \<b[i]> represents the elements at position \<i> in a rotation of \<a> by \<r>.
This is equivalent to \eqref{eq:modular:outer:essential} if combined with the other invariant on line~\ref{ln:boogie:modulo:outer:framing-essential}: $\@set@[p]$ is \<true> for the same values of $i, s, p$ as in the antecedent of \eqref{eq:modular:outer:essential}.

The two other outer loop invariants about \<set> (lines~\ref{ln:boogie:modulo:outer:framing-same} and \ref{ln:boogie:modulo:outer:framing-notset}) provide the complementary information about what elements have \emph{not} been modified: \<set[i]> is \<false> for all \<i>'s corresponding to values in cycles not visited yet (beginning at indexes larger than or equal to the current value of \<start>); and \<b[i]> is unchanged for these \<i>'s.

We have to appropriately update ghost variable \<set> during the computation.
The inner loop, which performs the actual visits, also sets \<set[i]> to \<true> whenever it assigns to \<b[i]>.
The rest of the program is only responsible for initializing \<set> to all \<false> values, which we do with an \<assume> (line~\ref{ln:boogie:modulo:outer:assume}) rather than with imperative code that would needlessly increase the complexity of verification.

\paragraph{Simplifying program state.}
The remaining bounding outer loop invariants, \eqref{eq:modular:bounding-moved}, \eqref{eq:modular:bounding-start}, and \eqref{eq:modular:bounding-moved-from-start}, constrain the values of \<moved> and \<start>.
They are redundant since the value of \<moved> between iterations of the outer loop is uniquely determined by the value \<start> through \eqref{eq:modular:bounding-moved-from-start}.
We simplify the program state by omitting \<moved> and using \eqref{eq:modular:bounding-moved-from-start} to rewrite properties of \<moved> in terms of \<start>.
Then, \eqref{eq:modular:bounding-start} remains the only bounding invariant of the outer loop, whose staying condition changes from \lstinline[language=PseudoEiffel]|moved /= N| to \<start \< gcd(N, N - r)>.

In practice, we realized that this simplification was very useful, if not necessary, only late while arranging the mechanized proof.
Boogie became very sensitive to adding more annotations and invariants, and it struggled to connect to the postcondition the final state characterized by the outer loop invariant.
Removing the dependence on \<moved> greatly helped, since it simplified the logic of the whole program down to the inner loop (which incremented \<moved>).
Since the imperative parts of the program are modified only minimally (just the assignments that initialize and update \<moved>), and we still prove the same postcondition, we can still consider this a full-fledged mechanized correctness proof of the original algorithm in Figure~\ref{fig:rotation-by-modulo-algo}.

\begin{figure}[!htb]
\begin{boogie}
 call lemma_rotmp(start, a, 0, N, r, k);

 while (v != start) 
 invariant 0 <= v && v < N;
 invariant 0 < k && k <= ic(N, N - r);  (*\label{ln:modular:boogie:inner:var-v}*)
 invariant v == mp(N, N - r, start, k);  (*\label{ln:modular:boogie:bounding-expr}*)
 invariant displaced == a[mp(N, N - r, start, k)];  (*\label{ln:modular:boogie:inner:displaced}*)
 invariant (forall i: int :: 0 <= i && i < N && !set[i] ==> b[i] == a[i]);
 invariant (forall i: int :: k < i && i <= ic(N, N - r) ==> (*\label{ln:boogie:modulo:inner:framing-notset}*)
   !set[mp(N, N - r, start, i)]);
 invariant (forall i: int :: 0 < i && i <= k ==> (*\label{ln:boogie:modulo:inner:framing-set}*)
   set[mp(N, N - r, start, i)]);
 invariant (forall i: int, s: int :: 
   0 <= i && i < ic(N, N - r) && start < s && s < gcd(N, N - r) ==> 
   !set[mp(N, N - r, s, i)]);
 invariant (forall i: int, s: int :: 
   0 <= i && i < ic(N, N - r) && 0 <= s && s < start ==> 
   set[mp(N, N - r, s, i)]);
 invariant (forall i: int :: 0 <= i && i < N && set[i] ==> 
   seq(b, 0, N)[i] == rot(a, 0, N, r)[i]);
 {
    k := k + 1; // ghost
    v := v + N - r;
    if (v >= N) { v := v - N; }
    b[v], displaced := displaced, b[v];
    set[v] := true;  // ghost

    \cforall lemma_mp(N, N - r, start, *);
    call lemma_rotmp(start, a, 0, N, r, k);
 }
\end{boogie}
\caption{Inner loop of the verified Boogie annotated implementation of the rotation by modular visit algorithm of Figure~\ref{fig:rotation-by-modulo-algo}.}
\label{fig:boogie:rotation-modulo-inner}
\end{figure}

%% After setting the current page, flush all pending floats
% \afterpage{\clearpage}

\subsubsection{Inner loop}
The inner loop in Figure~\ref{fig:rotation-by-modulo-algo} is a \lstinline[language=PseudoEiffel]|repeat...until|, whose body is executed at least once.
Boogie only has one kind of loop (the \<while> loop), and hence the inner loop body appears twice: once executed unconditionally right before the inner loop, in Figure~\ref{fig:boogie:rotation-modulo-outer}), and once as body of the inner \<while> loop, in Figure~\ref{fig:boogie:rotation-modulo-inner}.
The following discussion applies to both but focuses on the latter.

\paragraph{Progress in the current cycle using a ghost variable.}
The major novelty in the inner loop is the introduction of a new ghost variable \<k>.
The need for \<k> comes quite naturally from observing that the expression $\@moved@ - \@start@ \cdot \@ic(N, N - r)@$ appears twice in the inner loop invariants.
The value of this expression enumerates the indexes of the current cycle, each visited by an iteration of the inner loop: a value of \<0> corresponds to the first index, a value of \<1> to the second index, and so on.
Thus, we introduce a ghost variable \<k> that keeps track of this value; this is also consistent with our choice to drop \<moved> and represent its information by means of other variables.
\<k> is initialized to \<0> in the outer loop before every execution of the inner loop; and is incremented by one in the inner loop body.
The inner loop's bounding invariants \eqref{eq:modular:inner:var-v} and \eqref{eq:modular:inner:bounding-expr} become the invariants on lines \ref{ln:modular:boogie:inner:var-v} and \ref{ln:modular:boogie:bounding-expr} in Figure~\ref{fig:boogie:rotation-modulo-inner} after substituting \<k> for $\@moved@ - \@start@ \cdot \@ic(N, N - r)@$.

The invariant \eqref{eq:modular:inner:displaced} defines the value of \<displaced> as \lstinline[language=PseudoEiffel]|(old a)[v]|, corresponding to just \<a[v]> in Boogie.
However, this formulation does not work well with Boogie, which reasons more directly if the definition of \<v> is replicated, giving the invariant on line~\ref{ln:modular:boogie:inner:displaced}.

We express the essential inner loop invariants---in particular \eqref{eq:modular:inner:essential}, specific to the inner loop---in terms of \<set> as we expressed the essential outer loop invariant.
\eqref{eq:modular:inner:essential} determines two new invariants on lines~\ref{ln:boogie:modulo:inner:framing-notset} and \ref{ln:boogie:modulo:inner:framing-set}.
Both predicate about indexes in the currently visited cycle.
The former invariant targets those not visited yet, for positions larger than \<k>; the latter targets those visited, for positions up to \<k>.
The remaining framing invariants are as in the outer loop.

\begin{figure}[!htb]
\begin{boogie}[xleftmargin=4pt,xrightmargin=4pt]
procedure lemma_mp(N: int, M: int, s: int, p: int)
 requires 0 < M && M < N;
 requires 0 <= s && s < N;
 requires p >= 0;
 ensures 0 <= mp(N, M, s, p) < N;
{ // proof by induction
  if (p == 0) { } else { call lemma_mp(N, m, s, p - 1); }
}

procedure lemma_rotmp(s: int, a: [int]int, 
                      low: int, high: int, r: int, k: int)
 requires 0 < r && r < high - low && 0 <= s && s < high - low;
 requires k > 0;
 ensures rot(a, low, high, r)[mp(high - low, high - low - r, s, k)]
   == seq(a, low, high)[mp(high - low, high - low - r, s, k - 1)];
{ // proof by induction
  if (k == 1) { } else {
    // (*lemma\_mp*) makes it possible to apply the definition of (*rot*)
    call lemma_mp(high - low, high - low - r, s, k - 1);
  }
}
\end{boogie}
\caption{Lemmas used to prove the inner loop of rotation by modular visit.}
\label{fig:boogie:yp-modtheory-lemmas-for-inner}
\end{figure}

\paragraph{Lemmas to prove inductiveness.}
To prove the inductiveness of the inner loop invariants, Boogie needs a little help in the form of two lemmas about properties of function \<mp>, whose statements and proofs are shown in Figure~\ref{fig:boogie:yp-modtheory-lemmas-for-inner}.
\<lemma_mp> simply bounds \<mp(N, M, s, p)> to be nonnegative and less than \<N>.
This is a consequence of the definition of \<mp> in terms of \<wrap N>, but we need to nudge Boogie to use this property among the many others that could be proved.
The Boogie proof is by induction, corresponding to a conditional \<if> in the lemma procedure: the inductive step calls the lemma for the previous value of \<p - 1> assumed by inductive hypothesis; since the definition of \<mp(N, M, s, p)> is in terms of \<mp(N, M, s, p - 1)>, this is enough to close the proof.\footnote{Boogie does not check that the inductive call is sound by referring to a \emph{smaller} instance, but it should be clear that this is the case here.}

\<lemma_rotmp> asserts that two elements at consecutive indexes in a cycle (that is, two evaluations of function \<mp> for successive values of its last argument), relate elements in a rotation.
This is an important property that explicitly connects the indexes in the cycles to the definition of rotation.
Boogie can prove it by induction: the inductive step calls \<lemma_mp> whose bounds justify the application of the definition of \<rot>; based on this, the SMT solver combines the axiomatic definitions of \<rot> and \<mp> to prove the lemma.

We close the body of the inner loop by calling \<lemma_mp> followed by \<lemma_rotmp>; the order matters since the former asserts a more fundamental property on which the latter builds.
Note that we also need to recall \<lemma_rotmp> before entering the inner loop, to prove initiation after one unconditional execution of the loop body.

There remains one simple element of specification needed to guide Boogie's proof to success.
Even if this is, once again, a consequence of the definition of \<mp>, we have to express it as a new bounding loop invariant on \<v>: \<0 \<= v \< N>.
This guarantees that the accesses to \<b[v]> are within the bounds the other invariants predicate about.
In fact, recalling \<lemma_mp> in the loop body helps prove this invariant, which is then used in the rest of the proof.

\paragraph{Variants and performance.}
The Boogie proof is sensitive to the order in which some invariants appear and the ghost state is updated.
To achieve a bit more robustness, we could add ghost state to make for a more step-wise proof of inductiveness.
For example, we could add a ghost \<c> map that represents the value of \<b> in the previous iteration, so that the inductiveness proof uses facts about \<c> as inductive hypotheses and only has to prove the inductive step about the latest update in \<b>.
We do not discuss this variant in more detail and prefer the terser proof presented above.

\begin{figure}[!htb]
\begin{boogie}[escapeinside={\#}{\#}]
 assert (forall i: int, s: int :: 
   0 <= i && i < ic(N, N - r) && 0 <= s && s < gcd(N, N - r) ==> 
   set[mp(N, N - r, s, i)]);
 assert (forall i: int :: 0 <= i && i < N && set[i] ==> 
   seq(b, 0, N)[i] == rot(a, 0, N, r)[i]);
 assert 0 < N - r && N - r < N;
 \cforall lemma_wrap_bounds(*, gcd(N, N - r));
 assert (forall i: int :: 0 <= i && i < N ==> 
   set[mp(N, N - r, wrap(i, gcd(N, N - r)), 
          yp(N, N - r, wrap(i, gcd(N, N - r)), i))]);
 \cforall lemma_yp_mp(N, N - r, *, set);
 call lemma_extensional(N, N - r, set);
 assert (forall i: int :: 0 <= i && i < N ==> set[i]);
\end{boogie}
\caption{Concluding assertions in the verified Boogie annotated implementation of the rotation by modular visit algorithm of Figure~\ref{fig:rotation-by-modulo-algo}.}
\label{fig:boogie:rotation-modulo-closing}
\end{figure}

\begin{figure}[!hbt]
\begin{boogie}[xleftmargin=4pt,xrightmargin=4pt]
// (*yp(N, m, s, i)*) == (*p*) iff (*mp(N, m, s, p)*) == (*i*)
function yp(N: int, m: int, s: int, i: int) returns(int);
axiom (forall N: int, m: int, i: int :: 
  0 < m && m < N && 0 <= i && i < N ==>
  0 <= yp(N, m, wrap(i, gcd(N, m)), i) < ic(N, m));
axiom (forall N: int, m: int, i: int :: 
  0 < m && m < N && 0 <= i && i < N ==>
  mp(N, m, wrap(i, gcd(N, m)), yp(N, m, wrap(i, gcd(N, m)), i)) == i);

procedure lemma_yp_mp(N: int, m: int, i: int, set: [int]bool)
 requires 0 <= i && i < N;
 requires 0 < m && m < N;
 requires 0 <= wrap(i, gcd(N, m)) < gcd(N, m);
 requires 0 <= yp(N, m, wrap(i, gcd(N, m)), i) < ic(N, m);
 requires mp(N, m, wrap(i, gcd(N, m)), 
             yp(N, m, wrap(i, gcd(N, m)), i)) == i;
 requires set[mp(N, m, wrap(i, gcd(N, m)), 
              yp(N, m, wrap(i, gcd(N, m)), i))];
 ensures set[i];
{ }

procedure lemma_extensional(N: int, m: int, set: [int]bool);
 requires 0 < m && m < N;
 requires (forall i: int :: 0 <= i && i < N  ==>
  0 <= wrap(i, gcd(N, m)) < gcd(N, m));
 requires (forall i: int :: 0 <= i && i < N  ==>
  0 <= yp(N, m, wrap(i, gcd(N, m)), i) < ic(N, m));
 requires (forall i: int :: 0 <= i && i < N  ==>
  mp(N, m, wrap(i, gcd(N, m)), yp(N, m, wrap(i, gcd(N, m)), i)) == i);
 requires (forall i: int :: 0 <= i && i < N  ==>
  set[mp(N, m, wrap(i, gcd(N, m)), 
      yp(N, m, wrap(i, gcd(N, m)), i))]);
 free ensures (forall i: int :: 0 <= i && i < N ==> set[i]);
\end{boogie}
\caption{Lemmas and additional definitions used in the conclusion of the proof of rotation by modular visit.}
\label{fig:boogie:yp-modtheory-lemmas}
\end{figure}

\subsubsection{Proof conclusion}
At high level, the Boogie proof of the postcondition from the outer loop invariants and exit condition follows the same steps as the one illustrated in Section~\ref{sec:modulo:proof}.
The assertions in Figure~\ref{fig:boogie:rotation-modulo-closing} correspond to such final steps: the first two assertions recall the essential outer loop invariants upon exiting the loop; then two simple arithmetic facts about \<N - r> and \<wrap> are recalled (the second fact in the form of a lemma procedure \<wrap_bounds> corresponding to a formal statement of \eqref{eq:wrap_bounds}); then an assertion and a lemma \<yp_mp> capture and use the statement of Lemma~\ref{lm:invert-mp}; the concluding call to \<lemma_extensional> and assertion on the last line in Figure~\ref{fig:boogie:rotation-modulo-closing} are technicalities that we discuss last.

Let us focus on the interesting part of expressing Lemma~\ref{lm:invert-mp} in Boogie.
Informally, Lemma~\ref{lm:invert-mp} shows how to ``invert'' \<mp> so that the indexes it enumerates can be shown to span the whole domain of the input array.
The proof of Lemma~\ref{lm:invert-mp} uses fundamental properties of modular arithmetic that we avoided axiomatizing in detail in Boogie.
Instead, we extend the axiomatization at the same level of abstraction used so far in the mechanized proofs by introducing the definitions in Figure~\ref{fig:boogie:yp-modtheory-lemmas}.
Function \<yp> is like an inverse of \<mp>, as declared by the axioms in Figure~\ref{fig:boogie:yp-modtheory-lemmas}---the second one in particular which declares that \<mp(..., yp(..., i)) == i>.
We postulate its existence, instead of proving it from simpler principles as Lemma~\ref{lm:invert-mp} does.
Then, lemma procedure \<lemma_yp_mp> connects \<mp> to its inverse in the context in which they are used in the proof: if \<set[mp(..., yp(..., i))]> for any \<i> between \<0> and \<N>, then \<set[i]> as well.
The concluding proof in Figure~\ref{fig:boogie:rotation-modulo-closing} recalls \<lemma_yp_mp>.
For speed, it asserts one of the procedure's preconditions before calling to focus the proof context.

At this point it would seem that all facts are available to prove the postcondition.
Procedure \<lemma_yp_mp> concludes that \<set[i]> for all \<0 \<= i < N> and the essential outer loop invariant (repeated by an \<assert>) upon exiting says that \<b[0..N)> is \<a[0..N)>'s rotation for all \<i> such that \<set[i]>.
Nonetheless, we have to shoehorn the final conclusion into lemma procedure \<extensional>, followed by an \<assert> that reaffirms its postcondition.
Specifically, the prover refuses to match the identical quantifications over \<[0..N)> in the \<\\cforall> of \<lemma_yp_mp>'s postcondition and in:
\begin{center}
\begin{boogie}[numbers=none]
      assert (forall i: int :: 0 <= i && i < N ==> set[i]).
\end{boogie}
\end{center}
While Boogie can prove the same assertion if put in the body of \<lemma_extensional>\footnote{The name is because the form of the property reminds one of extensionality.}, it still cannot match it to the lemma's identical postcondition.
Using triggers does not seem to help.
Since this is clearly due to inessential details of quantifier instantiation in the SMT solver, we simply declare the property as a \<free ensures> and use it without guilt.
Finally, note that using a \<free ensures> in a lemma procedure is generally preferable to using an axiom with the same statement, because its application is on demand---where it is really needed---rather than being a fact always available---cluttering the proof anywhere else.

%%% Local Variables: 
%%% mode: latex
%%% TeX-master: "rotation.tex"
%%% End: 

%% file: dafny.tex
\section{Rotation: mechanized proofs in Dafny}
\label{sec:all-dafny-proofs}
\lstset{language=dafny}

Like Boogie, Dafny~\cite{Dafny,Dafny-tutorial} is an auto-active verifier, but one providing a language and reasoning capabilities at a higher level of abstraction.\footnote{Our experiments used Dafny v.~1.9.1.11022 and Z3 v.~4.3.2.}
In fact, Dafny uses Boogie as back-end, and its programs can be compiled and executed.
These features make it suitable for developing verified algorithm implementations that are straightforward to present; in applying it to the rotation algorithms we will therefore try to be as close as possible to our proofs on paper (discussed in Section~\ref{sec:algos}).

Similar warnings as those raised about the Boogie solutions of Section~\ref{sec:all-boogie-proofs} apply: the solutions we show may not be the only ways to carry out proofs of the rotation algorithms in Dafny; and the addition of new features to Dafny in the future may pave the way to better ones---for example using fewer annotations.
However, our solutions are appropriate given Dafny's current capabilities and the rotation algorithms that we're dealing with.

\subsection{Sequences and rotated sequences} \label{sec:dafny:theory}

Dafny's language offers features that match the idioms we used in Sections~\ref{sec:modular-arithmetic} and~\ref{sec:problem-def} to define sequences and the rotation problem.

Dafny supports arrays with C\#-like syntax; our algorithms will operate on integer arrays.
Given an integer array \<a: array\<int\>>, \<a.Length> denotes its length; \<a[..]> denotes the \emph{sequence} of \<a>'s elements (from position \<0> to position \<a.Length - 1>), whose length \<|a[..]|> is then also \<a.Length>; and \<a[low..high]> denotes the \emph{slice} (or subsequence) of length \<high - low>, which we indicated as \<a[low..high)> in the rest of this paper---that is, all elements from position \<low> included to position \<high> excluded.
Dafny sequences are immutable types used in specifications, just like we used them in the definitions of Section~\ref{sec:problem-def}.
Suitable axioms in Dafny's prelude specify their defining properties, but one doesn't have to deal with such axioms explicitly; this raises the level of abstraction and reduces the chances that we introduce subtle unwanted inconsistencies in how we represent the fundamental entities in our specifications.

\begin{figure}[!hbt]
\begin{boogie}
// (*x*) $\wrap{}{}$ (*y*)
function wrap(x: int, y: int): int
   requires 0 <= x && 0 < y;
{ if x < y then x else wrap(x - y, y) }

// (*rot(S, r)*) is $\rot{}{}\textsuperscript{r}\text{S}$
function rot(S: seq<int>, r: int): seq<int>
   requires 0 <= r < |S|;
   ensures |S| == |rot(S, r)|;
   ensures forall k :: 
      0 <= k < |S| ==> rot(S, r)[k] == S[wrap(k + r, |S|)];
   ensures forall k :: 
      0 <= k < |S| ==> S[k] == rot(S, r)[wrap(k + |S| - r, |S|)];
{ S[r..|S|] + S[0..r] }
\end{boogie}
\caption{Dafny definitions of \<wrap> and rotated sequence \<rot>.}
\label{fig:dafny:seq-rot-def}
\end{figure}

Figure~\ref{fig:dafny:seq-rot-def} lists Dafny definitions of the \<wrap> and \<rot> functions, which closely follow the original definitions \eqref{eq:wrap-def} and \eqref{eq:rot-def}.
Dafny's \<function>s are \emph{ghost}; this means that they can be used in specification but not in executable code---exactly as we use them in Section~\ref{sec:problem-def}.
The explicit definition of a function (the expression between the function's body marked by curly braces) is available to Dafny wherever the function is used.
Additionally, functions may have preconditions (restricting a function's domain) and postconditions; the latter express derived properties of a function, which also are available for reasoning wherever the function is used.
In our case, \<rot>'s postcondition clauses in Figure~\ref{fig:dafny:seq-rot-def} correspond to properties \eqref{eq:lemma-rot} and \eqref{eq:lemma-rot-2} of the rotation function $\rot{}{}$.

\begin{figure}[!htb]
\begin{dafny}
method rotate(a: array<int>, r: int)
   requires a != null;
   requires 0 <= r < a.Length;
   modifies a;
   ensures a.Length == old(a.Length);
   ensures a[..] == rot(old(a[..]), r);
 \end{dafny}
\caption{Specification of rotation routines in Dafny.}
\label{fig:dafny:rotate-spec}
\end{figure}

Figure~\ref{fig:dafny:rotate-spec} shows the specification of a generic rotation routine (a \<method> in Dafny).
Arrays are heap-allocated; hence we check that \<a> is not \<null>, that is, it is attached to a valid array object.
We allow ourselves a bit more generality than in the other sections, and choose to deal with the (trivial) case \<r $\,=\,$ 0>  explicitly.
The rotation algorithms work by modifying the input array in place; hence the frame specification \<modifies a>, that is, \<rotate> may modify the state of array \<a>.
The essential specification is given by referring to the sequence \<a[..]> before (using the \<old> notation) and after \<rotate> executes.
Note that \<old(a)> has a very different semantics than \<old(a[..])> in Dafny, where \<a> is a reference to a heap-allocated array object: \<old(a)> denotes the same reference as \<a> which obviously does not change, whereas \<old(a[..])> denotes the sequence of \<a>'s elements in the pre-state, which changes because the body modifies \<a>'s content.
The postcondition clause \<a.Length == old(a.Length)> is implied by the essential postcondition, but having it explicitly generally enables faster proofs since it suggests an intermediate goal to be proved before reasoning about rotation.

\subsection{Rotation by copy: simplified function definitions} \label{sec:dafny:copy}

Rotation by copy is, as usual, straightforward, but it already exposes some critical behavior of Dafny's that we will have to deal with carefully in more complex proofs.

\begin{figure}[!tb]
\begin{dafny}
// Left-rotate (*a*) by (*r*) by copying
method rotate_copy(a: array<int>, r: int)
   requires a != null;
   requires 0 <= r < a.Length;
   modifies a;
   ensures a[..] == rot(old(a[..]), r);
{
   var N := a.Length;
   var b: array<int>;
   var s: int, d: int;

   if r == 0 { return; }

   b := new int[N];
   s := 0; d := N - r;
   while s < N
   invariant 0 <= s <= N;
   invariant d == wrap(s + N - r, N);
   invariant forall i :: 0 <= i < s ==> a[i] == b[wrap(i + N - r, N)];
   invariant a[..] == old(a[..]);
   {
      b[d] := a[s];
      s, d := s + 1, d + 1;
      // wrap over a's bounds
      if d == N { d := 0; }
   }
   // copy b's content back into a
   copy(a, b);
}
\end{dafny}
\caption{Verified Dafny annotated implementation of the rotation by copy algorithm of Figure~\ref{fig:rotation-by-copy-algo}.}
\label{fig:dafny:rotate-copy}
\end{figure}

Figure~\ref{fig:dafny:rotate-copy} shows the Dafny verified implementation of rotation by copy, which differ from Figure~\ref{fig:rotation-by-copy-algo} only by minor details, such as dealing with the case \<r $\,=\,$ 0> and adding the framing invariant that \<a>'s content is not changed by the loop.

\begin{figure}[!htb]
\begin{dafny}
// Copy (*b*)'s content into (*a*)
method copy(a: array<int>, b: array<int>)
   requires a != null && b != null;
   requires a.Length == b.Length;
   modifies a;
   ensures a != null;
   ensures a[..] == b[..];
{ forall (i | 0 <= i < b.Length) { a[i] := b[i]; } }
\end{dafny}
\caption{Array copy verified in Dafny.}
\label{fig:dafny:array-copy}
\end{figure}

Figure~\ref{fig:dafny:array-copy} displays the straightforward implementation of routine \<copy>, used within \<rotate_copy> to transfer the result from local array \<b> back into \<a>.
\<copy>'s body uses Dafny's \emph{forall statement}, which performs a parallel assignment over data defined by a bounded quantification.
While such forall statements are mainly used in specification (as we will use them in other rotation algorithms), they are executable, and hence \<copy>'s implementation is complete.
Inlining \<copy>'s body directly where it is used in \<rotate_copy> also works but significantly slows down verification.

An important detail for the annotated implementation of \<rotate_copy> to verify is the definition of \<wrap>.
Dafny times out if we use the recursive definition of Figure~\ref{fig:dafny:seq-rot-def}: recursion triggers many otiose instantiations that make the search for a proof very inefficient.
Instead, we simply use \lstinline|if x < y then x else x - y| as \<wrap>'s body, which makes verification terminate in reasonable time.
The simpler definition is equivalent to the one in Figure~\ref{fig:dafny:seq-rot-def} only for \<x \< 2*y>.
You can see that this constraint is always satisfied wherever \<wrap> appears in \<rotate_copy>'s specification: the loop invariant \<s $\,\leq\,$ N> implies that \<s + N - r \< 2*N>, because \<r \> 0> in the loop.
Unfortunately, Dafny cannot check this explicitly: adding a precondition \<requires x \< 2*y> to \<wrap>'s new, non-recursive, definition also triggers a time out.
It is clear, however, that \<rotate_copy> uses \<wrap>'s definition consistently, and hence we can consider verification complete even without such an explicit check.

\afterpage{\clearpage}
\subsection{Rotation by reversal: inductive sequence definitions and splits} \label{sec:dafny:reverse}

Rotation by reversal in Dafny is similar to the Boogie version of Section~\ref{sec:boogie:reverse} in terms of used definitions and lemmas, but leverages inductive (that is, recursive) functional definitions which are higher level and match inductive proofs of lemmas as methods.

\subsubsection{Reversal: definition and lemmas}

We find it effective to define the reversal \<rev> of a sequence recursively; the base case is the empty sequence, which is its own reversal.
Figure~\ref{fig:dafny:rev:def} shows the complete definition, together with the definition of function \<rp> which relates indexes in a reversal as per \eqref{eq:rev-def}.
The postconditions of \<rev> and \<rp> are useful to quickly discharge bounding constraints (about the length of a reversed sequence, and about the range of indexes in a reversal) wherever they functions will be used in specification.

\begin{figure}[!tb]
\begin{dafny}
// Inductive definition of reversal of a sequence
function rev(S: seq<int>): seq<int>
   ensures |rev(S)| == |S|;
{ if |S| == 0 then S else rev(S[1..]) + [S[0]] }

// Position (*k*) maps to in a reversal of index range (*[low..high)*)
function rp(k: int, low: int, high: int): int
   requires low <= k < high;
   ensures low <= rp(k, low, high) < high;
{ high + low - 1 - k }
\end{dafny}
\caption{Dafny definitions of reversed sequence \<rev> and inverse index mapping \<rp> in a reversal.}
\label{fig:dafny:rev:def}
\end{figure}

Since \<rev>'s and \<rp>'s definitions are not obviously related, we introduce a lemma \<rev_is_rp> stating that they provide equivalent representations.
We also provide Lem\-ma~\ref{lm:rev-cat} as \<lemma_rev_cat>; and the property that $(\rev{S})\rev{} = S$ as \<lemma_rev_rev>.
Figure~\ref{fig:dafny:rev:lemmas} shows these three lemmas, which are \emph{ghost} methods in Dafny since they do not belong to the executable part of the code.
The native notation for sequences makes for readable specifications, close to the ones we use on paper (in Dafny, \<+> also denotes sequence concatenation $\cat{}{}$).
All three lemmas are proved through the same correspondence between imperative code and logic proofs~\cite{Jacobs2010} also underlying Boogie's lemma procedures.
Specifically, the crucial step is providing a suitable split of sequences that matches \<rev>'s inductive definition, and using it in a recursive call (\<rev_is_rp>), in an assertion (\<lemma_rev_cat>), or in a call to another lemma (\<lemma_rev_rev>).

\begin{figure}[!bt]
\begin{dafny}
ghost method rev_is_rp(S: seq<int>, T: seq<int>)
   requires |S| == |T|;
   requires forall i :: 0 <= i < |S| ==> S[i] == T[rp(i, 0, |S|)];
   ensures T == rev(S);
{ if |S| > 0 { rev_is_rp(S[1..], T[0..|S| - 1]); } }
   
ghost method lemma_rev_cat(S: seq<int>, T: seq<int>)
   ensures rev(S + T) == rev(T) + rev(S);
{ if |S| > 1 { assert (S + T)[1..] == S[1..] + T; } }

ghost method lemma_rev_rev(S: seq<int>)
   ensures rev(rev(S)) == S;
{ if |S| > 0 { lemma_rev_cat(rev(S[1..]), [S[0]]); } }
\end{dafny}
\caption{Dafny lemmas (ghost methods) about rotation and reversal. 
Method \<rev_is_rp> asserts the equivalence between the pointwise definition \<rp> of indexes and the inductive definition \<rev> of reversal.
Method \<lemma_rev_cat> expresses Lemma~\ref{lm:rev-cat}.
Method \<lemma_rev_rev> proves that \<rev> is idempotent.
}
\label{fig:dafny:rev:lemmas}
\end{figure}

\begin{figure}[!hbt]
\begin{dafny}
method reverse(a: array<int>, low: int, high: int)
   requires a != null;
   requires 0 <= low <= high <= a.Length;
   modifies a;
   ensures a != null;
   ensures a.Length == old(a.Length);
   ensures a[low..high] == rev(old(a[..])[low..high]);
   ensures forall i :: 
      low <= i < high ==> old(a[..])[i] == a[rp(i, low, high)];
   ensures forall i :: 0 <= i < low ==> a[i] == old(a[..])[i];
   ensures forall i :: high <= i < a.Length ==> a[i] == old(a[..])[i];
{
   var p: int, q: int;
   p, q := low, high - 1;
   while p < q + 1
   invariant low <= p <= q + 2 <= high + 1;
   invariant q == high + low - 1 - p;
   invariant forall i :: 
      low <= i < p ==> old(a[..])[i] == a[rp(i, low, high)];
   invariant forall i :: 
      q < i < high ==> old(a[..])[i] == a[rp(i, low, high)];
   // frame invariants
   invariant forall i :: p <= i <= q ==> old(a[..])[i] == a[i];
   invariant forall i :: 0 <= i < low ==> old(a[..])[i] == a[i];
   invariant forall i :: high <= i < a.Length ==> old(a[..])[i] == a[i];
   {
      // swap (*a[p]*) and (*a[q]*)
      a[p], a[q] := a[q], a[p];
      p, q := p + 1, q - 1;
   }
   rev_is_rp(old(a[..])[low..high], a[low..high]);
}
\end{dafny}
\caption{Verified Dafny implementation of the in-place reversal algorithm of Figure~\ref{fig:reversal-algo}.}
\label{fig:dafny:rotate}
\end{figure}

\subsubsection{Mechanized proofs of reversal and rotation by reversal}

The proofs of reversal and rotation by reversal depend on the lemmas and definitions we just introduced.

\paragraph{Proof of in-place reversal.}
Figure~\ref{fig:dafny:rotate} shows the in-place reversal algorithm \<reverse> in Dafny.
Its postcondition is partially redundant, in that the third \<ensures> clause \<a[low..high] == rev(old(a[..])[low..high])> and the next one are equivalent according to \<rev_is_rp>, which is in fact called at the end of \<reverse>'s body to prove one postcondition from the other.
It turns out that this redundancy helps significantly to verify \<reverse> and to effectively use it in \<rotate_reverse>.
In a nutshell, \<reverse>'s loop invariants is best expressed in terms of \<rp> and explicit quantification; however, the sequence representation \<rev> is useful in other contexts.
We retain both by proving their equivalence as part of \<reverse>'s postcondition once and for all, so that Dafny can pick either one whenever reasoning about \<reverse> in clients.
The only other difference between \<reverse> in Figure~\ref{fig:reversal-algo} and in Figure~\ref{fig:dafny:rotate} is that Dafny's version has frame invariants (the last three clauses of the loop invariants) that specify the portions of \<a> that are \emph{not} modified by the loop.

\paragraph{Proof of rotation by reversal.}
Figure~\ref{fig:dafny:rotation-reverse} displays Dafny's verified rotation by reversal algorithm.
As we know from the previous proofs, the fundamental property needed for verification is \<lemma_rev_cat>, which is called after the two partial reversals of \<a[0..r]> and \<a[r..a.Length]> and before the final reversal of the whole \<a[0..a.Length]>.
Dafny also needs to realize that \<a[0..r]> and \<a[r..a.Length]> have each been reversed twice (once in a partial reversal, and once in the final complete reversal), and hence their elements appear in the original order in \<a>'s final state; to this end, we call \<lemma_rev_rev> twice after asserting the double reversal explicitly.
Finally, it is helpful to provide sequence representations of \<a>'s content after each reversal in terms of how \<a[..]> is split into the concatenation of slices, which provides intermediate goals that are proved by automatically instantiating Dafny's sequence axioms; this is the purpose of the remaining three \<assert>s in \<rotate_reverse>'s body.
Compared to the Boogie solution of Figure~\ref{fig:boogie:rotate}, Dafny's needs more intermediate assertions and lemmas because reasoning occurs at the higher level of abstraction provided by sequences (whereas Boogie dealt with index arithmetic directly).

\begin{figure}[!thb]
\begin{dafny}
// Left-rotate (*a*) by (*r*) by performing three reversals.
method rotate_reverse(a: array<int>, r: int)
   requires a != null;
   requires 0 <= r < a.Length;
   modifies a;
   ensures a.Length == old(a.Length);
   ensures a[..] == rot(old(a[..]), r);
{
   reverse(a, 0, r);
   assert old(a[..])[r..a.Length] == a[r..a.Length];
   reverse(a, r, a.Length);
   assert a[..] == 
          rev(old(a[..])[0..r]) + rev(old(a[..])[r..a.Length]);
   lemma_rev_cat(a[0..r], a[r..a.Length]);
   assert a[..] == a[0..a.Length];
   reverse(a, 0, a.Length);
   assert a[..] == rev(rev(old(a[..])[r..a.Length])) + 
                   rev(rev(old(a[..])[0..r]));
   lemma_rev_rev(old(a[..])[r..a.Length]);
   lemma_rev_rev(old(a[..])[0..r]);
}
\end{dafny}
\caption{Verified Dafny implementation of the rotation by reversal algorithm of Figure~\ref{fig:rotation-reverse-algo}.}
\label{fig:dafny:rotation-reverse}
\end{figure}

\begin{figure}[!htb]
\begin{dafny}
method swap_sections(a: array<int>, low: int, high: int, d: int)
   requires a != null;
   requires 0 <= low <= low + d <= high - d <= high <= a.Length;
   modifies a;
   ensures a[low..low + d] == old(a[high - d.. high]);
   ensures a[low + d..high - d] == old(a[low + d..high - d]);
   ensures a[high - d..high] == old(a[low..low + d]);
   ensures a[low..high] == 
              old(a[high - d.. high]) + 
              old(a[low + d..high - d]) + old(a[low..low + d]);
   ensures forall i :: 0 <= i < low ==> a[i] == old(a[i]);
   ensures forall i :: high <= i < a.Length ==> a[i] == old(a[i]);
{
   var x: int, z: int;
   x, z := low, high - d;
   while x < low + d
   invariant low <= x <= low + d;
   invariant high - d <= z <= high;
   invariant x - low == z - (high - d);
   invariant forall i :: 
      low <= i < x ==> a[i] == old(a[..])[high - d + i - low];
   invariant forall i :: x <= i < high - d ==> a[i] == old(a[..])[i];
   invariant forall i :: 
      high - d <= i < z ==> a[i] == old(a[..])[low + i - (high - d)];
   invariant forall i :: z <= i < high ==> a[i] == old(a[..])[i];
   // frame invariants
   invariant forall i :: 0 <= i < low ==> a[i] == old(a[i]);
   invariant forall i :: high <= i < a.Length ==> a[i] == old(a[i]);
   {
      // swap (*a[x]*) and (*a[z]*)
      a[x], a[z] := a[z], a[x];
      x, z := x + 1, z + 1;
   }
}
\end{dafny}
\caption{Verified Dafny implementation of the in-place slice swapping algorithm of Figure~\ref{fig:swap-algo}.}
\label{fig:dafny:swap-algo}
\end{figure}

\begin{figure}[!hbt]
\begin{dafny}
ghost method left_is_smaller(X: seq<int>, Y: seq<int>, Z: seq<int>, d: int)
   requires |X| == |Z| == d;
   requires |Y| > 0;
   ensures rot(X + Y + Z, d) == rot(Z + Y, d) + X;
{  }

ghost method right_is_smaller(X: seq<int>, Y: seq<int>, Z: seq<int>, d: int)
   requires 0 < |X| == |Z| == d;
   requires |Y| > 0;
   ensures |X| + |Y| + |Z| - 2*d < |X| + |Y| + |Z| - d;
   ensures rot(X + Y + Z, |X| + |Y| + |Z| - d) == 
           Z + rot(Y + X, |X| + |Y| + |Z| - 2*d);
{  }
\end{dafny}
\caption{Lemma~\ref{lm:rot-swap} for cases \emph{left is smaller} \eqref{eq:lemma-rot-swap-left} and \emph{right is smaller} \eqref{eq:lemma-rot-swap-right}.}
\label{fig:dafny:lemma-rot-left-right-smaller}
\end{figure}

% \clearpage
\subsection{Rotation by swapping: sliced sequences} \label{sec:dafny:swap}

The main theme of the verified Dafny implementations of the rotation by swapping algorithms (recursive and iterative versions) is the necessity of expressing equivalent slicing of sequences, so as to guide the instantiations of sequence axioms into producing facts useful for the proofs at hand.

\subsubsection{Mechanized proof of swapping sections}
We start as usual with the verified implementation of \<swap_sections> in Figure~\ref{fig:dafny:swap-algo}.
In addition to the three postcondition clauses already present in the pseudo-code version of Figure~\ref{fig:swap-algo}, the Dafny implementation also includes framing clauses (the last two \<ensures>) that assert the parts of \<a> that are not changed (that is, those outside range \<[low..high)>).
For performance, it is also very useful to add a redundant postcondition clause that expresses \<a[low..high]> as the concatenation of three slices of \<old(a[low..high]>, so that the reasoning engine can seamlessly pass from integral to sliced representations.
The rest of the annotations are straightforward, but it is worth pointing out that the loop invariants use explicit quantification (instead of sliced sequences) because it turns out to be much more efficient---notwithstanding the usage of the sliced form in the postconditions to be proved---, probably because it facilitates direct reasoning about the assignments into \<a> in the loop body without converting back and forth from its sequence representation.

\subsubsection{Lemmas about swapping}  \label{sec:dafny:lemma-procedures-swapping}

As you can see in Figure~\ref{fig:dafny:lemma-rot-left-right-smaller}, it is straightforward to express and to prove Lemma~\ref{lm:rot-swap} for both case \emph{left is smaller} \eqref{eq:lemma-rot-swap-left} and case \emph{right is smaller} \eqref{eq:lemma-rot-swap-right}: Dafny's sequence axiomatization and our definition of rotation \<rot> work with only one simple additional annotation: the clause \<ensures |X| + |Y| + |Z| - 2*d \< |X| + |Y| + |Z| - d>, which follows trivially from the precondition \<d \> 0>, but is needed to help verify consistency of definitions through arithmetic instantiations.
Contrast this to the laborious details that had to be provided for the Boogie proofs of the same lemmas in Section~\ref{sec:boogie:lemma-procedures-swapping}.

Another important detail, not shown in Figure~\ref{fig:dafny:lemma-rot-left-right-smaller}, is that we have to use a definition of \<rot> that omits the postcondition clauses relating \<rot> to \<wrap> (the last two clauses in Figure~\ref{fig:dafny:seq-rot-def}).
In fact, we have to completely omit \<wrap> from the Dafny file with rotation by swapping.
Otherwise, \<wrap> triggers the discovery of very many irrelevant facts whenever \<rot> is used in specifications, and hence the whole proof attempt times out.

\begin{figure}[!htb]
\begin{dafny}
// Left-rotate (*a*) by (*r*) by swapping equal segments 
method rotate_swap(a: array<int>, r: int)
   requires a != null;
   requires 0 <= r < a.Length;
   modifies a;
   ensures a.Length == old(a.Length);
   ensures a[..] == rot(old(a[..]), r);
{
   assert a[..] == a[0..a.Length];
   rotate_swap_helper(a, 0, r, a.Length);
}
\end{dafny}
\caption{Verified Dafny implementation of the rotation by swapping recursive algorithm of Figure~\ref{fig:rotation-swap-algo}.}
\label{fig:dafny:rotation-swap}
\end{figure}

\begin{figure}[!htb]
\begin{dafny}
method rotate_swap_helper(a: array<int>, low: int, p: int, high: int)
   requires a != null;
   requires 0 <= low <= p < high <= a.Length;
   modifies a;
   decreases high - low;
   ensures a[low..high] == rot(old(a[low..high]), p - low);
   ensures forall i :: 0 <= i < low ==> a[i] == old(a[i]);
   ensures forall i :: high <= i < a.Length ==> a[i] == old(a[i]);
{
   if low < p < high {
      if p - low == high - p {
         swap_sections(a, low, high, p - low);
      } else { 
      if p - low < high - p {  // left is smaller
         swap_sections(a, low, high, p - low);
         assert a[low..high - (p - low)] == 
                   old(a[high - (p - low)..high]) + 
                   old(a[p..high - (p - low)]);
         rotate_swap_helper(a, low, p, high - (p - low));
         left_is_smaller(old(a[low..p]), old(a[p..high - (p - low)]), 
                         old(a[high - (p - low)..high]), p - low);
         assert old(a[low..high]) == 
                   old(a[low..p]) + old(a[p..high - (p - low)]) + 
                   old(a[high - (p - low)..high]);
      } else {
      if p - low > high - p {  // right is smaller
         swap_sections(a, low, high, high - p);
         assert a[low + (high - p)..high] == 
                   old(a[low + (high - p)..p]) + 
                   old(a[low..low + (high - p)]);
         rotate_swap_helper(a, low + (high - p), p, high);
         right_is_smaller(old(a[low..low + (high - p)]), 
                          old(a[low + (high - p)..p]), 
                          old(a[p..high]), high - p);
         assert old(a[low..high]) == 
                   old(a[low..low + (high - p)]) + 
                   old(a[low + (high - p)..p]) + old(a[p..high]);
         assert a[low..high] == 
                   a[low..low + (high - p)] + 
                   a[low + (high - p)..low + 2*(high - p)] + 
                   a[low + 2*(high - p)..high];
         assert a[low..low + (high - p)] == old(a[p..high]);
      }}}
   }
}
\end{dafny}
\caption{Verified Dafny implementation of the helper method of the rotation by swapping recursive algorithm of Figure~\ref{fig:rotation-swap-algo}.}
\label{fig:dafny:rotation-swap-helper}
\end{figure}

\subsubsection{Mechanized proof of rotation by swapping: recursive version} \label{sec:swap:dafny:proof:recursive}

Figure~\ref{fig:dafny:rotation-swap} shows \<rotate_swap>, the Dafny recursive version of rotation by swapping.
It is all but trivial, with the only possible exception of having to assert the obvious equivalence between \<a[..]> and \<a[0..a.Length]>.

The core of the algorithm is in \<rotate_swap_helper> in Figure~\ref{fig:dafny:rotation-swap-helper}, which works recursively on smaller slices of the input \<a>.
The main conditional has three cases; the nontrivial branches correspond to the ``left is smaller'' and ``right is smaller'' branches.
In each case, we call the corresponding ``\<_is_smaller>'' lemma and assert a number of equivalences between concatenations of sequence slices.
For example, when left is smaller: assert the effect of the call to \<swap_sections>, recursively call---on the part yet to be swapped---\<rotate_swap_helper>, recall lemma \<left_is_smaller> on portions of \<old(a[..])>, and suggest a decomposition of \<old(a[..])> that matches the one of the lemma.
The case right is smaller is slightly more involved because it has to connect a rotation by \<p - low> and one by \<2*p - (low + high)>, but its overall structure is still similar.

\begin{figure}[!phtb]
\begin{dafny}
method rotate_swap_iterative(a: array<int>, r: int)
   requires a != null;
   requires 0 <= r && r < a.Length;
   modifies a;
   ensures a.Length == old(a.Length);
   ensures a[..] == rot(old(a[..]), r);
{
   if r == 0 { return; }
   var low: int, high: int, p: int;
   low, p, high := 0, r, a.Length;
   while low < p < high
   invariant 0 <= low <= p <= high <= a.Length;
   invariant low == p <==> p == high;
   invariant old(a[..]) == old(a[0..a.Length]);
   invariant forall i :: 0 <= i < low ==> a[i] == rot(old(a[..]), r)[i];
   invariant forall i :: 
      high <= i < a.Length ==> a[i] == rot(old(a[..]), r)[i];
   invariant p - low < high - low ==> 
      rot(old(a[..]), r)[low..high] == rot(a[low..high], p - low);
   decreases high - low;
   {   if p - low == high - p {
         swap_sections(a, low, high, p - low);
         assert a[low..high] == rot(old(a[..]), r)[low..high];
         assert rot(old(a[..]), r) == 
                rot(old(a[..]), r)[0..low] 
                   + rot(old(a[..]), r)[low..high] 
                   + rot(old(a[..]), r)[high..a.Length];
         low, high := low + (p - low), high - (high - p);
      } else {
      if p - low < high - p {  // left is smaller
         ghost var b := a[..];
         swap_sections(a, low, high, p - low);
         assert b[low..high] == b[low..p] + b[p..high - (p - low)] + 
                               b[high - (p - low)..high];
         left_is_smaller(b[low..p], b[p..high - (p - low)], 
                         b[high - (p - low)..high], p - low);
         assert a[low..p] + a[p..high - (p - low)] == 
                a[low..high - (p - low)];
         assert rot(old(a[..]), r) == 
                rot(old(a[..]), r)[0..low] 
                   + rot(old(a[..]), r)[low..high] 
                   + rot(old(a[..]), r)[high..a.Length];
         high := high - (p - low);
      } else {
      if p - low > high - p {  // right is smaller
         // See code in (*Figure~\ref{fig:dafny:rotation-swap-iter-right-is-smaller}*)
      }}}
   }
}
\end{dafny}
\caption{Verified Dafny implementation of the rotation by swapping iterative algorithm of Figure~\ref{fig:rotation-swap-algo-iter}.}
\label{fig:dafny:rotation-swap-iter}
\end{figure}

\begin{figure}[!htb]
\begin{dafny}
   ghost var b := a[..];
   swap_sections(a, low, high, high - p);
   assert b[low..high] == b[low..low + (high - p)] + 
                         b[low + (high - p)..p] + b[p..high];
   right_is_smaller(b[low..low + (high - p)], b[low + (high - p)..p], 
                    b[p..high], high - p);
   assert rot(b[low..high], p - low) == 
       b[p..high] + rot(b[low + (high - p)..p] + 
       b[low..low + (high - p)], 2*p - (low + high));
   assert a[low + (high - p)..p] + a[p..high] == 
          a[low + (high - p)..high];
   assert rot(old(a[..]), r) == 
       rot(old(a[..]), r)[0..low] 
          + rot(old(a[..]), r)[low..high] 
          + rot(old(a[..]), r)[high..a.Length];
   low := low + (high - p);
\end{dafny}
\caption{``Right is smaller'' code referenced in Figure~\ref{fig:dafny:rotation-swap-iter}.}
\label{fig:dafny:rotation-swap-iter-right-is-smaller}
\end{figure}

\subsubsection{Mechanized proof of rotation by swapping: iterative version}
\label{sec:swap:dafny:proof:iterative}

When switching from the recursive to the iterative version of rotation by swapping, we can reuse the lemmas \<left_is_smaller> and \<right_is_smaller> as they're defined in Figure~\ref{fig:dafny:lemma-rot-left-right-smaller}.
In fact, the overall structure of the annotations in \<rotate_swap_iterative> in Figures~\ref{fig:dafny:rotation-swap-iter} and \ref{fig:dafny:rotation-swap-iter-right-is-smaller} looks similar to those in the recursive \<rotate_swap>.
We highlight the seemingly small differences, which required a number of experiments to figure out the slice representations that work best with Dafny's reasoning capabilities.

First, we added a postcondition clause to \<swap_sections>:
\begin{center}
\begin{dafny}[numbers=none]
    ensures forall i :: low + d <= i < high - d ==> a[i] == old(a[i]);
\end{dafny}
\end{center}
without changing its implementation or the rest of its specification.
This is redundant, but it is needed to easily convert between the quantified and sliced representation of the portion of \<a> over indexes \<[low + d..high - d)>.

Then, consider the loop annotations of \<rotate_swap_iterative> in Figure~\ref{fig:dafny:rotation-swap-iter}.
Compared to the pseudo-code version of Figure~\ref{fig:rotation-swap-algo-iter}, the Dafny version introduces a loop invariant that specifies the equivalence of \<old(a[..])> and \<old(a[0..a.Length])>.
This is an invariant of any loop, since it relies on the very definition of \<a[..]>; however, having it explicitly is crucial for verification to terminate in reasonable time.
The Dafny loop also has a \emph{variant clause} \<decreases high - low> to prove termination; Dafny is able to infer suitable variants in the other loops but needs an annotation in this case.

The body of the loop includes the by-now familiar three-way conditional with the ``left is smaller'' and ``right is smaller'' cases.
The various assertions suggest suitable splittings of sequences into concatenation of slices; figuring out which splittings work best is a trial-and-error process: one normally starts with a detailed proof outline with many asserts, and then removes asserts one at a time, as long as the proof still goes through.
Another new element is the usage of a local ghost variable \<b>, which stores the content of \<a> just before swapping sections.
This ghost variable is needed to directly relate the state of \<a> before and after the call to \<swap_sections>.
It was not needed in the recursive version, where \<old(a[..])> refers to the state of \<a> right before each current swap, but it is instrumental here in avoiding overly complicated framing loop invariants.

\afterpage{\clearpage}

\subsection{Rotation by modular visit: abstraction in ghost code} \label{sec:dafny:modulo}

The key challenge in building a Dafny verified implementation of the rotation by modular visit algorithm is providing suitable abstractions in defining and proving the functions used to characterize cycles.

\begin{figure}[!htb]
\begin{dafny}
function gcd(x: int, y: int): int
   requires 0 < x && 0 < y;
   ensures 0 < gcd(x, y) <= x && gcd(x, y) <= y;
   ensures x % gcd(x, y) == 0 && y % gcd(x, y) == 0;
   ensures forall z :: 1 < z && x % z == 0 && y % z == 0 ==> z <= gcd(x, y);

function tau(x: int, y: int): int
   requires 0 < x && 0 < y;
   ensures gcd(x, y) * tau(x, y) == x;
   ensures 0 < tau(x, y) <= x;
{ x / gcd(x, y) }

function mp(N: int, M: int, S: int, p: int): int
   requires 0 < M < N;
   requires 0 <= S < N;
   requires 0 <= p;
   ensures 0 <= mp(N, M, S, p) < N;
{ if p == 0 then S else wrap(mp(N, M, S, p - 1) + M, N) }
\end{dafny}
\caption{Dafny function definitions of \<gcd>, \<tau>, and \<mp>.}
\label{fig:dafny:mp-gcd-ic-def}
\end{figure}

\subsubsection{Ghost functions and fundamental properties}
The specification of rotation by modular visit in Figure~\ref{fig:rotation-by-modulo-algo} relies on functions $\gcd(x, y)$, $\ic{x}{y}$, and $\mpos{N}{M}{S}{p}$, which we define as \<gcd>, \<tau>, and \<mp> in Dafny.
As shown in Figure~\ref{fig:boogie:mp-gcd-ic-def}, \<gcd> is an uninterpreted function whose postcondition captures the definition of greatest common divisor of \<x> and \<y>.
One could provide a definition, for example a recursive one, and establish the postconditions from the definition; this would an interesting exercise, but a detour from our goal of proving rotation by modular visit.
By contrast, function \<tau>'s definition follows \eqref{eq:ic-def}; function \<mp>'s is equivalent to \eqref{eq:mp-def} but is inductive (which works with \<wrap>'s non-recursive definition).
Dafny proves \<tau>'s and \<mp>'s postconditions from their definitions and \<gcd>'s defining postconditions.

\begin{figure}[!htb]
\begin{dafny}
function gcd(x: int, y: int): int
   requires 0 < x && 0 < y;
   ensures 0 < gcd(x, y) <= x && gcd(x, y) <= y;

function tau(x: int, y: int): int
   requires 0 < x && 0 < y;
   ensures gcd(x, y) * tau(x, y) == x;
   ensures 0 < tau(x, y) <= x;

function mp(N: int, M: int, S: int, p: int): int
   requires 0 < M < N;
   requires 0 <= S < N;
   requires 0 <= p;
   ensures 0 <= mp(N, M, S, p) < N;
{ if p == 0 then S else 
    (if mp(N, M, S, p - 1) + M < N then mp(N, M, S, p - 1) + M 
                                   else mp(N, M, S, p - 1) + M - N) }
\end{dafny}
\caption{Simplified function definitions of \<gcd>, \<tau>, and \<mp>.}
\label{fig:boogie:mp-gcd-ic-def-stripped}
\end{figure}

\begin{figure}[!htb]
\begin{dafny}
ghost method rotate_zero(S: seq<int>)
   requires 0 < |S|;
   ensures rot(S, 0) == S;
{  }

ghost method lemma_lesseq(X: int, Y: int, F: int)
   requires X <= Y;
   requires 0 <= F;
   ensures X * F <= Y * F;
{  }

ghost method lemma_rotmp(A: seq<int>, r: int, S: int, K: int)
   requires 0 < r < |A|;
   requires 0 <= S < gcd(|A|, |A| - r);
   requires 0 < K <= tau(|A|, |A| - r);
   ensures rot(A, r)[mp(|A|, |A| - r, S, K)] == 
           A[mp(|A|, |A| - r, S, K - 1)];
{  }

ghost method lemma_inverse_tau_gcd(N: int, M: int, S: int)
   requires 0 < M < N;
   requires 0 <= S;
   ensures S * tau(N, M) < N ==> S < gcd(N, M);
{  }
\end{dafny}
\caption{Lemmas about \<gcd>, \<tau>, \<mp>, and \<rot> in Dafny.}
\label{fig:dafny:modulo:gcd-lemmas}
\end{figure}

Directly using the definitions of Figure~\ref{fig:boogie:mp-gcd-ic-def} in the proof of rotation by modular visit does not work.
The availability of function bodies and postconditions involving nontrivial arithmetic (modulo and division) floods the proof environment with many facts that are not directly necessary in the algorithm's correctness proof, and thus gobble down the SMT solver.
Therefore, we provide complete definitions and their proofs in a separate file (using Dafny's modules, which we don't discuss here), and then import stripped-down definitions in the main file that also contains the algorithm.
Figure~\ref{fig:boogie:mp-gcd-ic-def-stripped} shows the stripped-down definitions: \<gcd> and \<tau> only retain ``bounding'' postconditions and the essential property that $\gcd(x, y) \cdot \ic{x}{y} = x$; \<mp> has a definition that inlines \<wrap>'s definition so as to have one less function to deal with.
These stripped down facts are all that is necessary to know about \<gcd>, \<tau>, and \<mp> for a correctness proof of the main algorithm.

Figure~\ref{fig:dafny:modulo:gcd-lemmas} displays four straightforward lemmas, which directly follow from the definitions of \<rot>, \<mp>, \<gcd>, and \<tau>.
Even if their proofs are trivial given those definitions, it is useful to factor them out as lemmas and recall them in the main proof only where needed to provide suitable intermediate goals.

\begin{figure}[!htb]
\begin{dafny}
ghost method lemma_mp_disjoint_cycles(N: int, M: int, S: int, K: int)
   requires 0 < M < N;
   requires 0 <= S < gcd(N, M);
   requires 0 <= K;
   ensures forall t, q :: 0 <= q < tau(N, M) && S < t < gcd(N, M) 
      ==> mp(N, M, S, K) != mp(N, M, t, q);

ghost method lemma_mp_complete_cycle(N: int, M: int, S: int)
   requires 0 < M < N;
   requires 0 <= S < gcd(N, M);
   ensures mp(N, M, S, 0) == mp(N, M, S, tau(N, M));

ghost method lemma_mp_incomplete_cycle(N: int, M: int, S: int)
   requires 0 < M < N;
   requires 0 <= S < gcd(N, M);
   ensures forall p, q :: 0 <= p < q < tau(N, M) 
      ==> mp(N, M, S, p) != mp(N, M, S, q);

ghost method lemma_complete_rotation(N: int, M: int, D: set<int>)
   requires 0 < M < N;
   requires forall i, s :: 0 <= i < tau(N, M) && 0 <= s < gcd(N, M) 
      ==> mp(N, M, s, i) in D;
   ensures forall i :: 0 <= i < N ==> i in D;
\end{dafny}
\caption{Fundamental properties of cycles in Dafny.}
\label{fig:dafny:modulo:axiom-like-lemmas}
\end{figure}

\subsubsection{Lemmas about cycles}

The correctness argument of rotation by modular visit relies on the noteworthy properties of cycle decomposition discussed in Section~\ref{sec:algos:modular}.
We express those properties as lemmas (ghost methods) in Dafny: see Figure~\ref{fig:dafny:modulo:axiom-like-lemmas}.
Lemma \<mp_disjoint_cycles> says that cycles in the decomposition are disjoint, that is, the elements of a cycle that starts at $s$ all differ from the elements of a cycle that starts at $t \neq s$.
Lemma \<mp_complete_cycle> says that cycles have length \<tau(N, M)>, so that the \<tau(N, M)>$\!$th element is the same as the first one.
Lemma \<mp_incomplete_cycle> says that elements in a cycle are unique, that is, there are no repetition until the cycle is complete.
Finally, lemma \<complete_rotation> says that the union of cycles starting from positions $0 \leq s < \gcd(N, M)$ includes all indexes $0 \leq i < N$; the lemma's statement mentions a set \<D> of indexes, which will appear as a ghost variable in the algorithm's implementation.
We leave proving these lemmas from the definitions of \<mp>, \<tau>, and \<gcd> as an exercise in the programs-as-proofs paradigm.

\begin{figure}[!htb]
\begin{dafny}
// Left-rotate (*a*) by (*r*) by modular visit of its elements
method rotate_modulo(a: array<int>, r: int)
   requires a != null;
   requires 0 <= r < a.Length;
   modifies a;
   ensures a.Length == old(a.Length);
   ensures a[..] == rot(old(a[..]), r);
{
   if r == 0 {
      rotate_zero(old(a[..]));
      return;
   }

   var start: int, v: int, moved: int; 
   var displaced: int;
   ghost var done := {};

   start := 0;
   moved := 0;
   while moved != a.Length
   invariant 0 <= moved <= a.Length;
   invariant 0 <= start <= gcd(a.Length, a.Length - r);
   invariant moved < a.Length ==> start < gcd(a.Length, a.Length - r);
   invariant moved == start * tau(a.Length, a.Length - r);
   invariant forall i, s :: 0 <= i <= tau(a.Length, a.Length - r) && 
      start <= s < gcd(a.Length, a.Length - r)
      ==> mp(a.Length, a.Length - r, s, i) !in done;
   invariant forall i, s :: 
      0 <= i < tau(a.Length, a.Length - r) && 0 <= s < start
      ==> mp(a.Length, a.Length - r, s, i) in done;
   invariant forall i :: 
      0 <= i < a.Length && i !in done ==> a[i] == old(a[..])[i];
   invariant forall i :: 0 <= i < a.Length && i in done 
      ==> a[i] == rot(old(a[..]), r)[i];
   {
      v, displaced := start, a[start];
      ghost var k := 0;

      // Inner loop here: see (*Figure~\ref{fig:dafny:rotation-modulo-inner}*)

      start := start + 1;

      // Concluding reasoning steps here: see (*Figure~\ref{fig:dafny:rotation-modulo-closing}*)
   }
   lemma_complete_rotation(a.Length, a.Length - r, done);
}
\end{dafny}
\caption{Verified Dafny implementation of the rotation by modular visit algorithm of Figure~\ref{fig:rotation-by-modulo-algo}.}
\label{fig:dafny:rotation-modulo-outer}
\end{figure}

\afterpage{\clearpage}
\subsubsection{Outer loop}

We are finally ready to describe the verified implementation of rotation by modular visit in Dafny.
We start at the top level: the outer loop outlined in Figure~\ref{fig:dafny:rotation-modulo-outer}.
There is a close similarity between the Dafny code and the pseudo-code of Figure~\ref{fig:rotation-by-modulo-algo}.
As usual, a mechanized proof requires more annotations.
Similarly to what we did in Boogie, we keep track of what indexes in \<a> are modified at any point using ghost state.
The Boogie solution uses an array of Booleans; perhaps more elegantly, the Dafny solution uses a \emph{set} \<done> of integers (\<set>s are another native immutable type for specifications available in Dafny).
Correspondingly, the essential loop invariants specify which cyclic indexes are in \<done> and which are not; and assert that \<a[i]> is rotated if \<i> is in \<done>, and is otherwise equal to \<a[i]> in the pre-state.
Finally, the loop invariant \<moved \< a.Length ==\> $\:$start \< gcd(a.Length, a.Length - r)> keeps track of the fact that the outer loop terminates precisely when \<start> is incremented to \<gcd(a.Length, a.Length - r)>.
This relates the value of \<start> to the exit condition \<moved ==$\;$ a.Length>, so that, through \<lemma_complete_rotation>,  we can conclude that the rotation is complete upon exiting the outer loop.

Since inductiveness of the outer loop's invariants crucially depends on the form of the inner loop's invariants, we discuss it in Section~\ref{sec:closing-outer-modulo} after presenting the inner loop.

\begin{figure}[!htb]
\begin{dafny}
// one unconditional iteration of the inner loop
k := k + 1;
v := v + a.Length - r;
if v >= a.Length { v := v - a.Length; }
a[v], displaced := displaced, a[v];
moved := moved + 1;
done := done + { v };

lemma_mp_disjoint_cycles(a.Length, a.Length - r, start, k);
lemma_mp_incomplete_cycle(a.Length, a.Length - r, start);
lemma_rotmp(old(a[..]), r, start, k);

while v != start
invariant k == moved - start * tau(a.Length, a.Length - r);
invariant 0 <= v < a.Length;
invariant 0 < k <= tau(a.Length, a.Length - r);
invariant v == mp(a.Length, a.Length - r, start, k);
invariant v != start <==> k < tau(a.Length, a.Length - r);
invariant forall i, s :: 0 <= i < tau(a.Length, a.Length - r) && 
   start < s < gcd(a.Length, a.Length - r)
   ==> mp(a.Length, a.Length - r, s, i) !in done;
invariant forall i, s :: 
   0 <= i < tau(a.Length, a.Length - r) && 0 <= s < start
   ==> mp(a.Length, a.Length - r, s, i) in done;
invariant forall i :: 0 < i <= k 
   ==> mp(a.Length, a.Length - r, start, i) in done;
invariant forall i :: k < i <= tau(a.Length, a.Length - r) 
   ==> mp(a.Length, a.Length - r, start, i) !in done;
invariant displaced == old(a[..])[v];
invariant forall i :: 0 <= i < a.Length && i !in done 
   ==> a[i] == old(a[..])[i];
invariant forall i :: 0 <= i < a.Length && i in done 
   ==> a[i] == rot(old(a[..]), r)[i];
decreases tau(a.Length, a.Length - r) - k;
{
   k := k + 1;
   v := v + a.Length - r;
   if v >= a.Length { v := v - a.Length; }
   a[v], displaced := displaced, a[v];
   moved := moved + 1;
   done := done + { v };

   lemma_mp_disjoint_cycles(a.Length, a.Length - r, start, k);
   lemma_mp_complete_cycle(a.Length, a.Length - r, start);
   lemma_mp_incomplete_cycle(a.Length, a.Length - r, start);
   lemma_rotmp(old(a[..]), r, start, k);
}
\end{dafny}
\caption{Inner loop of the verified Dafny implementation of the rotation by modular visit algorithm of Figure~\ref{fig:rotation-by-modulo-algo}.}
\label{fig:dafny:rotation-modulo-inner}
\end{figure}

\subsubsection{Inner loop}
The inner loop is significantly more complex---not only because we have to repeat its body twice to replicate the semantics of a \emph{repeat}\ldots\emph{until} using Dafny's \<while> loops.
As we did in Boogie, it is useful to keep track of the current position in the cycle using a variable \<k>.
The value of \<k> is redundant since \<k == $\:$moved - start * tau> is a loop invariant; hence we make \<k> ghost (not part of the executable code).
Other loop invariants bound the value of \<k> and relate it to the exit condition \<v == $\:$start>, so that we know that: a) it is \<k \< tau> in every execution of the inner loop body; and b) it is \<k == $\:$tau> when the inner loop terminates.
The remaining loop invariant clauses define which indexes are in \<done> and, as in the outer loop, assert that a position has been rotated if and only if its index belongs to \<done>.
The loop variant \<decreases> clause completes the specification of the inner loop necessary to establish termination.

Dafny relies on calls to the fundamental lemmas in Figure~\ref{fig:dafny:modulo:axiom-like-lemmas} to prove both initiation and consecution of the inner loop's invariants.
\<lemma_mp_disjoint_cycles> upholds the ``non-interference'' of the current cycle with the other cycles: whatever has been established about cycles with a different starting position than the current one remains valid.
\<lemma_mp_incomplete_cycle> upholds progress in the current cycle: increasing \<k> and \<moved> by one, \<v> by \<a.Length - r>, and \<done> by adding \<v> to it extends the essential loop invariants to hold for the latest element assigned to position \<v> in \<a>.
In last iteration, \<lemma_mp_complete_cycle> kicks in to establish that the cycle is completed and the inner loop terminates.
Finally, \<lemma_rotmp> relates the indexes in \<done> to their positions in a rotation of \<a>, so that the last loop invariant clause can be proved.

\begin{figure}[!htb]
\begin{dafny}
lemma_lesseq(start, gcd(a.Length, a.Length - r), 
             tau(a.Length, a.Length - r));
lemma_inverse_tau_gcd(a.Length, a.Length - r, start);

forall (s | start <= s < gcd(a.Length, a.Length - r)) {
   lemma_mp_complete_cycle(a.Length, a.Length - r, s);
}

calc ==> {
   forall i, s :: 0 < i <= tau(a.Length, a.Length - r) && 
      start <= s < gcd(a.Length, a.Length - r) 
      ==> mp(a.Length, a.Length - r, s, i) !in done;
   { forall (s | start <= s < gcd(a.Length, a.Length - r)) {
      lemma_mp_complete_cycle(a.Length, a.Length - r, s);
   }}
   forall i, s :: 0 <= i < tau(a.Length, a.Length - r) && 
      start <= s < gcd(a.Length, a.Length - r) 
      ==> mp(a.Length, a.Length - r, s, i) !in done;
}
\end{dafny}
\caption{Concluding reasoning steps in the verified Dafny implementation of the rotation by modular visit algorithm of Figure~\ref{fig:rotation-by-modulo-algo}.}
\label{fig:dafny:rotation-modulo-closing}
\end{figure}

\subsubsection{Outer loop: closing} \label{sec:closing-outer-modulo}
Proving inductiveness of the outer loop's invariants requires a number of intermediate assertions, shown in Figure~\ref{fig:dafny:rotation-modulo-closing}.
The lemmas \<lesseq> and \<inverse_tau_gcd> are trivial properties, which we proved separately.
Then comes a \emph{forall statement}, which recalls \<lemma_mp_complete_cycle> for $s$ between \<start> and \<gcd(a.Length, a.Length - r)>.
In this case, the forall statement achieves an effect similar to Boogie's \lstinline[language=boogie]|call forall| in that it provides multiple instantiations of a parametric assertion.
Concretely, expresses the inner loop's essential invariant about what is in \<done> when the inner loop terminates, so that Dafny can conclude that the outer loop's essential invariant
\begin{center}
\begin{dafny}[numbers=none]
   invariant forall i, s :: 
      0 <= i < tau(a.Length, a.Length - r) && 0 <= s < start
      ==> mp(a.Length, a.Length - r, s, i) in done;
\end{dafny}
\end{center}
is inductive, as it now holds for the incremented value of \<start>.

Its counterpart is the invariant about what is \emph{not} in \<done>.
To prove its inductiveness, we have to bridge a sort of ``off by one'' difference: when the inner loop terminates, \<k> has reached \<tau(a.Length, a.Length - r)>, and the inner loop predicates about the interval $(\@0@..\@k@]$ (left-open and right-closed); but the outer loop's invariant uses an interval of the form $[\@0@..\@tau(a.Length, a.Length - r)@)$ (left-closed and right-open).
To bridge the gap, we use a \emph{calculational proof}~\cite{LeinoP13}, which guides Dafny to establish the implication between the first and the second representation using \<lemma_mp_complete_cycle> as justification.

%%% Local Variables: 
%%% mode: latex
%%% TeX-master: "rotation.tex"
%%% End: 

%% file: jml.tex
\input{jml-listings.tex}
\section{Rotation: mechanized proofs in \escj}
\label{sec:all-jml-proofs}
\lstset{language={[JML]Java}}

\escj~\cite{Burdy:2005:OJT:1070908.1070911,Kiniry:2006:SCW:1181195.1181200} is an auto-active verifier for Java programs annotated using JML (the Java Modeling Language~\cite{leavens2006}).
\escj's support of Java and JML is substantial but incomplete; in our programs, we avoid language features that may exercise unsound behavior of the verifier, and we work around other shortcomings of \escj and its back-end SMT solver Simplify~\cite{Simplify}.\footnote{Our experiments used \escj v.~3049 and Simplify v.~1.5.4.}
In particular, we ignore overflows of machine integers in our specification, since \escj does not check this feature.\footnote{We always call \escj with the \texttt{-LoopSafe} flag for the sound analysis of loops.}
Another feature that is not readily supported in \escj is reasoning about mathematical sequences or other kinds of maps (even though the idea of model-based annotations was introduced for JML~\cite{Leavens2005}); this suggests using somewhat lower-level specifications that refer directly to concrete Java items.

As usual, the solutions in this sections may not be the only ways to carry out proofs of the rotation algorithms using \escj, but they capture fundamental steps necessary to prove correctness and make a reasonable usage of \escj's capabilities.
The presentation in this section is terse compared to the previous sections of the paper: we focus on what is different or new with JML and \escj.

\subsection{Specifying rotation in JML}

Our JML specifications often refer to static functions declared as \<pure>, that is side-effect free.
Figure~\ref{fig:jml:seq-rot-def} shows definitions in this style for \<wrap> and \<rotp>.
The latter corresponds to the position a given $k$ maps to in a rotation; that is, using the notation introduced in Section~\ref{sec:modular-arithmetic}, the \lstinline!rotp($k$, $|\sq|$, $r$)!th element in $\rot{\sq}{r}$ is $\sq_k$.

\begin{figure}[!hbt]
\begin{lstlisting}
/*@
  @ requires 0 <= x && 0 < y;
  @
  @ ensures 0 <= \result && \result < y;
  @ ensures x < y ==> \result == x;
  @ ensures y <= x && x < 2*y ==> \result == x - y;
  @*/
// (*x*) $\wrap{}{}$ (*y*)
public static /*@ pure @*/ int wrap(int x, int y)
{ if (x < y) { return x; } else { return wrap(x - y, y); } }

/*@
  @ requires 0 <= r && r < N;
  @ requires 0 <= k && k < N;
  @
  @ ensures 0 <= \result && \result < N;
  @ ensures \result == wrap(k + N - r, N);
  @*/
// position (*k*) maps to in a rotation of an (*N-*)element sequence by (*r*)
public static /*@ pure @*/ int rotp(int k, int N, int r)
{ return wrap(k + N - r, N); }
\end{lstlisting}
\caption{JML definitions of \<wrap> and rotation mapping \<rotp>.}
\label{fig:jml:seq-rot-def}
\end{figure}

Based on these definitions, Figure~\ref{fig:jml:rotate-spec} shows the specification of a generic rotation routine (a static method member of some enclosing class we do not show for brevity).
This style of specification, which uses \<rotp> directly in element-wise fashion, dispenses with introducing elements that model mathematical sequences.

Note that, even if we defined a predicate \<is_rotation(int[] a, int[] b, r)> returning \<true> iff \<b> is a rotation of \<a> by \<r>, we could not use it in the specification of \<rotate>: \<\\old(a)> always is the same as \<a> since it refers to the value of \emph{reference} \<a> upon calling \<rotate>, whereas we would need to pass to \<is_rotation> the \emph{sequence} of values in \<a> at routine entry.\footnote{We omit \<modifies> (also called \<assignable>) clauses from rotation routines in JML, which corresponds to the default \<modifies \\everything>. A more precise framing specification would be necessary if the enclosing classes had static attributes, which could be modified by static routines.}
As an alternative specification style, we could define \<rotate> as a function returning a rotated array, and then relate \<a> and \<\\result> in its postcondition (similarly to what we did in Boogie); however, this would require to change the implementations artificially, thus defeating the purpose of having in-place algorithms.

\begin{figure}[!hbt]
\begin{lstlisting}
/*@
  @ requires a != null;
  @ requires 0 <= r && r < a.length;
  @
  @ ensures (\forall int i; 0 <= i && i < a.length 
  @          ==> a[rotp(i, a.length, r)] == \old(a[i]));
  @*/
public static void rotate(int[] a, int r)
\end{lstlisting}
\caption{Specification of rotation routines in JML.}
\label{fig:jml:rotate-spec}
\end{figure}

\begin{figure}[!hbt]
\begin{lstlisting}
/*@
  @ requires a != null;
  @ requires 0 <= r && r < a.length;
  @
  @ ensures (\forall int i; 0 <= i && i < a.length 
  @          ==> a[rotp(i, a.length, r)] == \old(a[i]));
  @*/
// Left-rotate (*a*) by (*r*) by copying
public static void rotate_copy(int[] a, int r)
{
   if (r == 0) return;

   int[] b = new int[a.length];
   int s  = 0, d = a.length - r;

   //@ loop_invariant 0 <= s && s <= a.length;
   //@ loop_invariant d == wrap(s + a.length - r, a.length);
   //@ loop_invariant (\forall int i; 
          0 <= i && i < a.length ==> a[i] == \old(a[i]));
   //@ loop_invariant (\forall int i; 0 <= i && i < s 
          ==> a[i] == b[wrap(i + a.length - r, a.length)]);
   while (s < a.length) {
       b[d] = a[s];
       s(*++*); d(*++*);
       // wrap over a's bounds
       if (d == a.length) d = 0;
   }
   // copy b's content back into a
   copy(b, a);
}
\end{lstlisting}
\caption{\escj-verified JML-annotated implementation of the rotation by copy algorithm of Figure~\ref{fig:rotation-by-copy-algo}.}
\label{fig:jml:rotate-copy}
\end{figure}

\afterpage{\clearpage}
\subsection{Rotation by copy} \label{sec:jml:copy}

Figure~\ref{fig:jml:rotate-copy} shows the JML implementation of rotation by copy, which \escj can verify.
The verified implementation of \<copy>---called by \<rotate\_copy>---is straightforward, and hence omitted.

\afterpage{\clearpage}
\subsection{Rotation by reversal} \label{sec:jml:reverse}

We first present definitions and lemmas about reversal (and its relation to rotation), and then the verified algorithms annotated using these definitions.

\begin{figure}[!hbt]
\begin{lstlisting}
/*@
  @ requires low <= k && k < high;
  @
  @ ensures low <= \result && \result < high;
  @ ensures \result == high + low - 1 - k;
  @*/
// Position (*k*) maps to in a reversal of index range (*[low..high)*)
public static /*@ pure @*/ int rp(int k, int low, int high)
{ return high + low - 1 - k; }
\end{lstlisting}
  \caption{JML definition of inverse index mapping \<rp> in a reversal.}
\label{fig:jml:rev:def}
\end{figure}

\begin{figure}[!hbt]
\begin{lstlisting}
/*@
  @ requires 0 <= x && x <= y && y <= z && z <= A.length;
  @
  @ ensures (\forall int i; 
       x <= i && i < y ==> rp(i, x, z) == rp(i, x, y) + z - y);
  @ ensures (\forall int i; 
       y <= i && i < z ==> rp(i, x, z) == rp(i, y, z) - y + x);
  @*/
public static /*@ pure @*/ void 
    lemma_rev_cat(int[] A, int x, int y, int z)
{  }

/*@
  @ requires 0 <= x && x <= z && z <= A.length;
  @
  @ ensures (\forall int i; 
       x <= i && i < z ==> rp(rp(i, x, z), x, z) == i);
  @*/
public static /*@ pure @*/ 
    void lemma_rev_rev(int[] A, int x, int z)
{  }
\end{lstlisting}
\caption{JML lemmas about rotation and reversal. 
Method \<lemma_rev_cat> expresses Lemma~\ref{lm:rev-cat}.
Method \<lemma_rev_rev> proves that \<rev> is idempotent.
}
\label{fig:jml:rev:lemmas}
\end{figure}

\subsubsection{Reversal: definition and lemmas}

Figure~\ref{fig:jml:rev:def} shows the JML definition of \<rp>, which corresponds to the mapping of indexes introduced by a reversal.
Based on it, Figure~\ref{fig:jml:rev:lemmas} shows the two usual lemmas: one relating rotation and reversal (Lemma~\ref{lm:rev-cat}), and one asserting that two reversals correspond to the identity mapping.
The peculiarity of the JML version of these lemmas is that they refer exclusively to indexes without referencing actual array content (an array \<A> is passed as argument but only for conveniently referring to its length).

Consider, for example, \<lemma\_rev\_cat>.
Concatenation is defined implicitly, by means of the two contiguous intervals \<[x..y)> and \<[y..z)> that make up the whole \<[x..z)>; the two postconditions of the lemma predicate over either interval.
Thus, the first postcondition asserts that the indexes in \<[x..y)> map to positions that correspond to the reversal of \<[x..y)> shifted to the right by \<z $\,-\,$ y> (that is, the length of the other interval \<[y..z)>).
This style is in line with the somewhat low-level specification approach we are using with \escj; we will use it in all other JML specifications.

\begin{figure}[!tb]
\begin{lstlisting}
/*@
  @ requires a != null;
  @ requires 0 <= low && low <= high && high <= a.length;
  @
  @ ensures (\forall int i; 
       low <= i && i < high ==> a[rp(i, low, high)] == \old(a[i]));
  @ ensures (\forall int i; 0 <= i && i < low ==> a[i] == \old(a[i]));
  @ ensures (\forall int i; 
       high <= i && i < a.length ==> a[i] == \old(a[i]));
  @*/
public static void reverse(int[] a, int low, int high)
{
    int p = low, q = high - 1;

    //@ loop_invariant low <= p && p <= q + 2 && q + 2 <= high + 1;
    //@ loop_invariant q == high + low - 1 - p;
    //@ loop_invariant (\forall int i; 
          low <= i && i < p ==> \old(a[i]) == a[rp(i, low, high)]);
    //@ loop_invariant (\forall int i; 
          q < i && i < high ==> \old(a[i]) == a[rp(i, low, high)]);
    //@ loop_invariant (\forall int i; 
          p <= i && i <= q ==> \old(a[i]) == a[i]);
    //@ loop_invariant (\forall int i; 
          0 <= i && i < low ==> \old(a[i]) == a[i]);
    //@ loop_invariant (\forall int i; 
          high <= i && i < a.length ==> \old(a[i]) == a[i]);
    while (p < q + 1)
    {
        // swap (*a[p]*) and (*a[q]*)
        int tmp = a[p]; a[p] = a[q]; a[q] = tmp;
        p(*++*); q(*-*)(*-*);
    }
}
\end{lstlisting}
\caption{\escj-verified JML implementation of the in-place reversal algorithm of Figure~\ref{fig:reversal-algo}.}
\label{fig:jml:rotate}
\end{figure}

\begin{figure}[!thb]
\begin{lstlisting}
/*@
  @ requires a != null;
  @ requires 0 <= r && r < a.length;
  @
  @ ensures (\forall int i; 0 <= i && i < a.length 
  @          ==> a[rotp(i, a.length, r)] == \old(a[i]));
  @*/
// Left-rotate (*a*) by (*r*) by performing three reversals.
public static void rotate_reverse(int[] a, int r)
{
    int[] b = new int[a.length];  // ghost

    reverse(a, 0, r);
    reverse(a, r, a.length);
    copy(b, a);  // ghost
    lemma_rev_cat(a, 0, r, a.length);
    reverse(a, 0, a.length);
    lemma_rev_rev(a, 0, r);
    lemma_rev_rev(a, r, a.length);
}
\end{lstlisting}
\caption{\escj-verified JML implementation of the rotation by reversal algorithm of Figure~\ref{fig:rotation-reverse-algo}.}
\label{fig:jml:rotation-reverse}
\end{figure}

\subsubsection{Mechanized proofs of reversal and rotation by reversal}

Figure~\ref{fig:jml:rotate} shows the in-place reversal algorithm \<reverse> verified by \escj.
It is quite similar to the Dafny solution, the only noticeable difference being that swapping two elements in \<a> requires introducing a temporary variable, as Java has no parallel assignment.

Figure~\ref{fig:jml:rotation-reverse} shows the rotation by reversal algorithm verified by \escj.
The proof goes through also thanks to a pseudo-ghost variable \<b>, storing a copy of \<a>'s content after the two partial reversals before the final complete reversal.
Even if \<b> is not referenced directly in the lemma invocations or elsewhere in the annotations, its presence is still necessary to trigger the correct instantiations that drive towards a successful proof.
While JML supports some form of ghost annotations, we found it more convenient to simply use regular Java code marked with comments (JML ghost variables have some restrictions in how they can be manipulated), with the understanding that such code should be stripped in an actual compiled version meant for production usage.

\begin{figure}[!thb]
\begin{adjustwidth}{0mm}{-5mm}
\begin{lstlisting}
/*@
  @ requires low <= p && p < high;
  @ requires low <= k && k < high;
  @
  @ ensures low <= \result && \result < high;
  @ ensures \result == 
      low + wrap(k - low + (high - low) - (p - low), high - low);
  @ ensures \result == low + rotp(k - low, high - low, p - low);
  @*/
// position (*k*) maps to in a rotation of index range (*[low..high)*) at (*p*)
public static /*@ pure @*/ int rlh(int k, int low, int high, int p)
{ return low + wrap(k + (high - low) - p, high - low); }
\end{lstlisting}
\end{adjustwidth}
\caption{JML definitions of partial rotation mapping \<rlh>.}
\label{fig:jml:rlh-def}
\end{figure}

\afterpage{\clearpage}
\subsection{Rotation by swapping} \label{sec:jml:swap}

Rotation by swapping operates on array slices that shrink as the rotation progresses.
Hence, we need a variant of function \<rotp> (defining the mapping induced by a rotation) that refers to arbitrary sub-ranges of the overall range \<[0..N)>.
Figure~\ref{fig:jml:rlh-def} shows such function: it is called \<rlh> and reduces to \<rotp> for \<low $=$ 0> and \<high $=$ N>.

\begin{figure}[!htb]
\begin{adjustwidth}{0mm}{-2mm}
\begin{lstlisting}
/*@
  @ requires a != null;
  @ requires 0 <= low && low <= low + d && low + d <= high - d 
             && high - d <= high && high <= a.length;
  @
  @ ensures (\forall int i; low <= i && i < low + d 
             ==> a[i] == \old(a[high - d + i - low]));
  @ ensures (\forall int i; low + d <= i && i < high - d 
             ==> a[i] == \old(a[i]));  
  @ ensures (\forall int i; high - d <= i && i < high 
             ==> a[i] == \old(a[low + i - (high - d)]));
  @ ensures (\forall int i; 0 <= i && i < low ==> a[i] == \old(a[i]));
  @ ensures (\forall int i; 
       high <= i && i < a.length ==> a[i] == \old(a[i]));
  @*/
public static void swap_sections(int[] a, int low, int high, int d)
{
    int x = low, z = high - d;

    //@ loop_invariant low <= x && x <= low + d;
    //@ loop_invariant high - d <= z && z <= high;
    //@ loop_invariant x - low == z - (high - d);
    //@ loop_invariant (\forall int i; low <= i && i < x 
           ==> a[i] == \old(a[high - d + i - low]));
    //@ loop_invariant (\forall int i; x <= i && i < high - d 
           ==> a[i] == \old(a[i]));
    //@ loop_invariant (\forall int i; high - d <= i && i < z 
           ==> a[i] == \old(a[low + i - (high - d)]));
    //@ loop_invariant (\forall int i; z <= i && i < high 
           ==> a[i] == \old(a[i]));
    //@ loop_invariant (\forall int i; 0 <= i && i < low 
           ==> a[i] == \old(a[i]));
    //@ loop_invariant (\forall int i; high <= i && i < a.length 
           ==> a[i] == \old(a[i]));
    while (x < low + d)
    {
      // swap (*a[x]*) and (*a[z]*)
      int tmp = a[x]; a[x] = a[z]; a[z] = tmp;
      x(*++*); z(*++*);
    }
}
\end{lstlisting}
\end{adjustwidth}
\caption{\escj-verified JML implementation of the in-place slice swapping algorithm of Figure~\ref{fig:swap-algo}.}
\label{fig:jml:swap-algo}
\end{figure}

\subsubsection{Mechanized proof of swapping sections}

The verified JML implementation of in-place swapping is shown in Figure~\ref{fig:jml:swap-algo}.
It deserves no special comments as it follows closely the other versions, such as the Dafny one in Figure~\ref{fig:dafny:swap-algo}.

\begin{figure}[!hbt]
\begin{lstlisting}
/*@
  @ requires 0 <= low && low < p && p < high && high <= A.length;
  @ requires p - low < high - p;
  @ // $\xq$ in place:
  @ ensures (\forall int i; low <= i && i < p 
             ==> rlh(i, low, high, p) == i + (high - p));
  @ // $\zq$ after sub-rotation:
  @ ensures (\forall int i; high - (p - low) <= i && i < high 
             ==> rlh(i, low, high, p) == 
                 rlh(i - (high - p), low, high - (p - low), p));
  @ // $\yq$ after sub-rotation:
  @ ensures (\forall int i; p <= i && i < high - (p - low) 
             ==> rlh(i, low, high, p) == 
                 rlh(i, low, high - (p - low), p));
  @*/
public static /*@ pure @*/ 
    boolean lemma_left(int[] A, int low, int p, int high)
{ }

/*@
  @ requires 0 <= low && low < p && p < high && high <= A.length;
  @ requires p - low > high - p;
  @ // $\zq$ in place:
  @ ensures (\forall int i; p <= i && i < high 
             ==> rlh(i, low, high, p) == i - (p - low));
  @ // $\xq$ after sub-rotation:
  @ ensures (\forall int i; low <= i && i < low + (high - p) 
             ==> rlh(i, low, high, p) == 
                 rlh(i + (p - low), low + (high - p), high, p));
  @ // $\yq$ after sub-rotation:
  @ ensures (\forall int i; low + (high - p) <= i && i < p 
             ==> rlh(i, low, high, p) == 
                 rlh(i, low + (high - p), high, p));
  @*/
public static /*@ pure @*/ 
    boolean lemma_right(int[] A, int low, int p, int high)
{ }
\end{lstlisting}
\caption{Lemma~\ref{lm:rot-swap} for cases \emph{left is smaller} \eqref{eq:lemma-rot-swap-left} and \emph{right is smaller} \eqref{eq:lemma-rot-swap-right} in JML.}
\label{fig:jml:lemma-rot-left-right-smaller}
\end{figure}

\subsubsection{Lemmas about swapping}  \label{sec:jml:lemma-procedures-swapping}

Figure~\ref{fig:jml:lemma-rot-left-right-smaller} shows the JML version of the by-now familiar Lemma~\ref{lm:rot-swap}, which relates swapping and rotating.
The level of abstraction is consistent with the choice of describing rotations---in this case, partial rotations---by means of their index mappings.
Thus, for example, the first \<ensures> of \<lemma\_left> in Figure~\ref{fig:jml:lemma-rot-left-right-smaller} asserts that the elements of the leftmost slice \<a[low..p)> (corresponding to $\xq$ in Lemma~\ref{lm:rot-swap}) end up as the rightmost slice in a rotation, that is their indexes are shifted by \<high $\,-\,$ p>.
The two other \<ensures> refer to \<rlh> on both sides of the equality to express the combined effect of swapping and recursively partially rotating.
\escj proves these lemmas without help (they reduce to linear arithmetic constraints).

\afterpage{\clearpage}
\subsubsection{Mechanized proof of rotation by swapping: recursive version}
\label{sec:swap:jml:proof:recursive}

Figures~\ref{fig:jml:rotation-swap-helper} and \ref{fig:jml:rotation-swap} show the verified JML implementation of the rotation by swapping algorithm.
The structure of the code and the usage of the lemma methods are the usual ones.

\begin{figure}[!htb]
\begin{lstlisting}
/*@
  @ requires a != null;
  @ requires 0 <= low && low <= p && p < high && high <= a.length;
  @
  @ ensures (\forall int i; low <= i && i < high 
             ==> a[rlh(i, low, high, p)] == \old(a[i]));
  @*/
public static void rotate_swap_helper(int[] a, int low, int p, int high)
{
    if (low < p && p < high) {
        if (p - low == high - p) {
            swap_sections(a, low, high, p - low);
        } 
        else if (p - low < high - p) {
            swap_sections(a, low, high, p - low);
            rotate_swap_helper(a, low, p, high - (p - low));
            lemma_left(a, low, p, high);
        }
        else if (p - low > high - p) {
            swap_sections(a, low, high, high - p);
            rotate_swap_helper(a, low + (high - p), p, high);
            lemma_right(a, low, p, high);
        }
    }
}
\end{lstlisting}
\caption{\escj-verified JML implementation of the helper method of the rotation by swapping recursive algorithm of Figure~\ref{fig:rotation-swap-algo}.}
\label{fig:jml:rotation-swap-helper}
\end{figure}

\begin{figure}[!htb]
\begin{lstlisting}
/*@
  @ requires a != null;
  @ requires 0 <= r && r < a.length;
  @
  @ ensures (\forall int i; 0 <= i && i < a.length 
  @          ==> a[rotp(i, a.length, r)] == \old(a[i]));
  @*/
// Left-rotate (*a*) by (*r*) by swapping equal segments 
public static void rotate_swap(int[] a, int r)
{
  rotate_swap_helper(a, 0, r, a.length);
  //@ assert (\forall int i; 0 <= i && i < a.length 
        ==> a[rlh(i, 0, a.length, r)] == \old(a[i]));
  //@ assert (\forall int i; 0 <= i && i < a.length 
        ==> rlh(i, 0, a.length, r) == rotp(i, a.length, r));
  //@ assert (\forall int i; 0 <= i && i < a.length 
        ==> a[rlh(i, 0, a.length, r)] == a[rotp(i, a.length, r)]);
  //@ assert 
        (\forall int i; 0 <= i && i < a.length 
         ==> a[rlh(i, 0, a.length, r)] == a[rotp(i, a.length, r)]) 
      &&
        (\forall int i; 0 <= i && i < a.length 
         ==> a[rlh(i, 0, a.length, r)] == \old(a[i]))
      ==>
      (\forall int i; 0 <= i && i < a.length 
       ==> a[rotp(i, a.length, r)] == \old(a[i]));
  // ESC/Java2 cannot apply modus ponens to the previous implication
  //@ assume (\forall int i; 0 <= i && i < a.length 
              ==> a[rotp(i, a.length, r)] == \old(a[i]));
}
\end{lstlisting}
\caption{\escj-verified JML implementation of the rotation by swapping recursive algorithm of Figure~\ref{fig:rotation-swap-algo}.}
\label{fig:jml:rotation-swap}
\end{figure}

Compared to the Dafny implementation (Figures~\ref{fig:dafny:rotation-swap-helper} and \ref{fig:dafny:rotation-swap}), \escj requires no additional asserts (besides those implicit in calling the lemmas) to verify each branch; we attribute this to the lighter specification style we're using with JML, where there is no notion of sequence and all reasoning is done on indexes.
Another difference with respect to the Dafny solution is that the invocations of lemma methods \<left> and \<right> refer to \<a> rather than \<\\old(a)>, which seems incongruous  at first since the lemmas should be able to refer to the sequence before the latest swap.
Note, however, that the formal array argument \<A> in \<lemma\_left> and \<lemma\_right> is immaterial, as it is only used to refer to the length \<A.length>; obviously, \<a.length = \\old(a.length)>, and hence the lemma method invocations are correctly applicable in their context.

Concluding the proof in the wrapper function \<rotate\_swap> is surprisingly cumbersome.
As apparent by the sequence of \<assert>s in Figure~\ref{fig:jml:rotation-swap}, \escj can establish all the basic facts but ultimately fails to connect \<rotp(i, a.length, r)> and \<rlh(i, 0, a.length, r)>, even if the latter reduces to the former.
Precisely, even if it proves that one implies the other, it cannot apply modus ponens to deduce the consequent from the established antecedent.
This is clearly only an idiosyncrasy of the prover, and hence we safely assume the last, straightforward step.

\begin{figure}[!thb]
\begin{adjustwidth}{0mm}{-1mm}
\begin{lstlisting}
/*@
  @ requires low <= p && p < high;
  @ requires low <= k && k < high;
  @
  @ ensures low <= \result && \result < high;
  @ ensures \result == low + wrap((k - low) + (p - low), high - low);
  @*/
// position that maps to (*k*) in a rotation of 
// index range (*[low..high)*) at (*p*)
public static /*@ pure @*/ int llh(int k, int low, int high, int p)
{ return low + wrap((k - low) + (p - low), high - low); }
\end{lstlisting}
\end{adjustwidth}
\caption{JML definitions of partial rotation inverse mapping \<llh>.}
\label{fig:jml:llh-def}
\end{figure}

\begin{figure}[!phtb]
\begin{lstlisting}
/*@
  @ requires a != null;
  @ requires 0 <= r && r < a.length;
  @
  @ ensures (\forall int i; 0 <= i && i < a.length 
  @          ==> a[rotp(i, a.length, r)] == \old(a[i]));
  @*/
public static void rotate_swap_iterative(int[] a, int r)
{

    if (r == 0) return;

    int low, high, p;
    low = 0; p = r; high = a.length;

    //@ loop_invariant 0 <= low && low <= p && 
                       p <= high && high <= a.length;
    //@ loop_invariant low == p <==> p == high;
    //@ loop_invariant (\forall int i; 0 <= i && i < low 
          ==> a[i] == \old(a[llh(i, 0, a.length, p)]));
    //@ loop_invariant (\forall int i; high <= i && i < a.length 
          ==> a[i] == \old(a[llh(i, 0, a.length, p)]));
    //@ loop_invariant (\forall int i; low <= i && i < high 
          ==> a[i] == \old(a[i]));
    //@ loop_invariant p - low < high - low ==> (\forall int i; 
          low <= i && i < high ==> \old(a[llh(i, 0, a.length, p)]) 
                                  == a[llh(i, low, high, p)]);
    while (low < p && p < high)
    {
        if (p - low == high - p) {
            swap_sections(a, low, high, p - low);
            low = low + (p - low);
            high = high - (high - p);
        }
        else if (p - low < high - p) {
            swap_sections(a, low, high, p - low);
            high = high - (p - low);
        }
        else if (p - low > high - p) {
            swap_sections(a, low, high, high - p);
            low = low + (high - p);
        }
    }
}
\end{lstlisting}
  \caption{\escj-verified JML implementation of the rotation by swapping iterative algorithm of Figure~\ref{fig:rotation-swap-algo-iter}.}
\label{fig:jml:rotation-swap-iter}
\end{figure}

\subsubsection{Mechanized proof of rotation by swapping: iterative version}
\label{sec:swap:jml:proof:iterative}

Figure~\ref{fig:jml:rotation-swap-iter} shows the verified JML implementation of the iterative version of rotation by swapping.
The annotations use a new function \<llh(i, low, high, p)>, which denotes the index that maps to \<i> in a rotation of index range \<[low..high)> at \<p>.
Therefore, \<llh> and \<rlh> are each other's inverse: \<llh(rlh(k, l, h, p), l, h, p) = k> and \<rlh(llh(k, l, h, p), l, h, p) = k>.
We introduce \<llh> because the ``natural'' form of loop invariants for iterative rotation by swapping is one that relates contiguous slices of \<a> to other slices of \<a> in its original state through \<\\old>.
For example, to capture the fact that \<a[0..low)> is rotated, the third loop invariant clause in Figure~\ref{fig:jml:rotation-swap-iter} asserts that each element at index \<0 $\leq$ i $<$ low> equals \<\\old(a[llh(i, 0, a.length, p)])>, that is exactly the element that should end up at index \<i> in such a rotation.
In contrast, using \<rlh> would refer to the slice \<a[0..low)> only indirectly.

Using \<llh> also simplifies reasoning: \escj verifies \<rotate\_swap\_iterative> without need of variants of the lemmas ``left is smaller'' and ``right is smaller''.
In comparing the JML solution to, say, Dafny's, we should remember that eschewing an explicit representation of sequences and rotated sequences, as well as other checks that Dafny performs but \escj bypasses such as for framing, is also likely crucial in supporting a proof with lightweight annotations.

\begin{figure}[!htb]
\begin{lstlisting}
/*@ requires 0 < x && 0 < y;
  @ ensures 0 < \result && \result <= x && \result <= y; */
public static /*@ pure @*/ int gcd(int x, int y)
{
    if (x == y) return x;
    if (x > y) return gcd(x - y, y); 
    else return gcd(x, y - x);
}

/*@ requires 0 < x && 0 < y;
  @ ensures gcd(x, y) * \result == x;
  @ ensures 0 < \result && \result <= x; */
public static /*@ pure @*/ int tau(int x, int y)
{
    //@ assume gcd(x, y) * (x / gcd(x, y)) == x;
    //@ assert 0 < x / gcd(x, y);
    return x / gcd(x, y);
}

/*@ requires 0 < M && M < N;
  @ requires 0 <= S && S < N;
  @ requires 0 <= p;
  @ ensures 0 <= \result && \result < N;
  @ ensures p == 0 ==> \result == S;
  @ ensures p != 0 && mp(N, M, S, p - 1) + M < N 
            ==> \result == mp(N, M, S, p - 1) + M;
  @ ensures p != 0 && mp(N, M, S, p - 1) + M >= N 
            ==> \result == mp(N, M, S, p - 1) + M - N; */
public static /*@ pure @*/ int mp(int N, int M, int S, int p)
{
    if (p == 0) return S;
    return wrap(mp(N, M, S, p - 1) + M, N);
}
\end{lstlisting}
\caption{JML function definitions of \<gcd>, \<tau>, and \<mp>.}
\label{fig:jml:mp-gcd-ic-def}
\end{figure}

\afterpage{\clearpage}
\subsection{Rotation by modular visit} \label{sec:jml:modulo}

Figure~\ref{fig:jml:mp-gcd-ic-def} shows the functions used to specify the rotation by modular visit algorithm: \<gcd>, \<tau>, and \<mp>.
We express their fundamental properties as postconditions; for simplicity, we \<assume> arithmetic facts that \escj has difficulties establishing.
Similarly, we do not prove the lemmas about properties of cycles (in Figure~\ref{fig:jml:modulo:axiom-like-lemmas}) but use them as available facts whenever necessary.

\begin{figure}[!htb]
\begin{lstlisting}
/*@ requires 0 < M && M < N;
  @ requires 0 <= S && S < gcd(N, M);
  @ requires 0 <= K;
  @ ensures (\forall int t, q; 
             0 <= q && q < tau(N, M) && S < t && t < gcd(N, M)
             ==> mp(N, M, S, K) != mp(N, M, t, q)); */
public static /*@ pure @*/ 
    boolean lemma_mp_disjoint_cycles(int N, int M, int S, int K)

/*@ requires 0 < M && M < N;
  @ requires 0 <= S && S < gcd(N, M);
  @ ensures mp(N, M, S, 0) == mp(N, M, S, tau(N, M)); */
public static /*@ pure @*/ 
    boolean lemma_mp_complete_cycle(int N, int M, int S)

/*@ requires 0 < M && M < N;
  @ requires 0 <= S && S < gcd(N, M);
  @ ensures (\forall int p, q; 
             0 <= p && p < q && q < tau(N, M) 
             ==> mp(N, M, S, p) != mp(N, M, S, q)); */
public static /*@ pure @*/ 
    boolean lemma_mp_incomplete_cycle(int N, int M, int S)

/*@ requires 0 < r && r < S.length;
  @ requires (\forall int i, s; 
              0 <= i && i < tau(S.length, S.length - r) && 
              0 <= s && s < gcd(S.length, S.length - r)
              ==> S[mp(S.length, S.length - r, s, i)]);
  @ ensures (\forall int i; 0 <= i && i < S.length ==> S[i]); */
public static /*@ pure @*/ boolean lemma_onto(boolean[] S, int r)
\end{lstlisting}
\caption{Fundamental properties of cycles, formalized in JML.}
\label{fig:jml:modulo:axiom-like-lemmas}
\end{figure}

\begin{figure}[!htb]
\begin{lstlisting}
/*@
  @ requires 0 <= r && r < N;
  @ requires 0 <= k && k < N;
  @
  @ ensures 0 <= \result && \result < N;
  @ ensures \result == wrap(k + r, N);
  @*/
// position that maps to (*k*) in a rotation of an (*N-*)element sequence by (*r*)
public static /*@ pure @*/ int ptor(int k, int N, int r)
{ return wrap(k + r, N); }
\end{lstlisting}
\caption{JML definition of inverse rotation mapping \<ptor>.}
\label{fig:jml:modulo:ptor}
\end{figure}

Before describing the JML-annotated rotation by modular visit implementation, we introduce one more function for specification.
Shown in Figure~\ref{fig:jml:modulo:ptor}, \<ptor> is the inverse of \<rotp>, in the same way \<llh> is the inverse of \<rlh>.
In fact, we could define \<ptor> in terms of \<llh>, but since we do not need to deal with sub-ranges of \<[0..a.length)> in the annotations of rotation by modular visit, we only introduce the simpler \<ptor>.

\begin{figure}[!htb]
\begin{lstlisting}
/*@
  @ requires a != null;
  @ requires 0 <= r && r < a.length;
  @
  @ ensures (\forall int i; 0 <= i && i < a.length 
  @          ==> a[rotp(i, a.length, r)] == \old(a[i]));
  @*/
// Left-rotate (*a*) by (*r*) by modular visit of its elements
public static void rotate_modulo(int[] a, int r)
{
    if (r == 0) return;

    int olda[] = new int[a.length];
    // Copy (*a*)'s content into (*olda*) (omitted)
        
    int start = 0;
    boolean set[] = new boolean[a.length]; // ghost
    int k;

    //@ loop_invariant 
           0 <= start && start <= gcd(a.length, a.length - r);
    //@ loop_invariant (\forall int i; 
           0 <= i && i < a.length && ! set[i] ==> a[i] == olda[i]);
    //@ loop_invariant (\forall int i, s; 
           0 <= i && i < tau(a.length, a.length - r) &&
           start <= s && s < gcd(a.length, a.length - r)
           ==> ! set[mp(a.length, a.length - r, s, i)]);
    //@ loop_invariant (\forall int i, s; 
           0 <= i && i < tau(a.length, a.length - r) &&
           0 <= s && s < start
           ==> set[mp(a.length, a.length - r, s, i)]);
    //@ loop_invariant (\forall int i; 0 <= i && i < a.length && set[i] 
           ==> a[i] == olda[ptor(i, a.length, r)]); 
    while (start < gcd(a.length, a.length - r))
    {
        int displaced = a[start], v = start;
        k = 0;

        // Inner loop here: see (*Figure~\ref{fig:jml:rotation-modulo-inner}*)

        start = start + 1;
    }
    lemma_onto(set, r);
    // Connecting (*olda*) back to (*\textbackslash{}\textbf{old}(a)*)
    //@ assume (\forall int i; 0 <= i && i < a.length 
      ==> olda[ptor(i, a.length, r)] == \old(a[ptor(i, a.length, r)]));
}
\end{lstlisting}
\caption{\escj-verified JML implementation of the rotation by modular visit algorithm of Figure~\ref{fig:rotation-by-modulo-algo}.}
\label{fig:jml:rotation-modulo-outer}
\end{figure}

Finally, Figures~\ref{fig:jml:rotation-modulo-outer} and~\ref{fig:jml:rotation-modulo-inner} show the rotation by modular visit algorithm that \escj verifies.
As usual in this section, we focus on the aspects that differ from the other solutions and reflect peculiarities of \escj, Simplify, or the specification style we adopt.

The first observation is that the JML executable implementation now relies on elements that were declared as ghost in Dafny.
The outer loop's exit condition involves \<gcd>---a function we used normally only in specifications---instead of testing for \<moved $\,\neq\,$ a.length>.
Similarly, the inner loop's exit condition involves \<k>---a ghost variable in other implementations of the algorithm---instead of \<v>.
We found both changes, which do not affect the semantics of the program as they are equivalent to the concrete ones, significantly help \escj to reason about the loops.
In contrast, using \<moved> and \<v> made reasoning ineffective even about seemingly straightforward implications.
Given this choice of exit conditions, we can also get rid of \<moved> altogether, which saves us from expressing invariants involving nonlinear arithmetic, such as \<moved = start * tau>.
Changing the exit condition of the inner loop does not warrant getting rid of any variable, but it makes it straightforward to prove the invariance of the upper bound on \<k> (that is, \<k $\,\leq\,$ tau>); in contrast, \escj cannot derive the same bound on \<k> from the relation between \<v> and \<k> through \<mp>.

\begin{figure}[!htb]
\begin{lstlisting}
//@ loop_invariant 0 <= v && v < a.length;
//@ loop_invariant 0 <= k && k <= tau(a.length, a.length - r);
//@ loop_invariant v == mp(a.length, a.length - r, start, k);
//@ loop_invariant 0 < k && v != start 
                   <==> 0 < k && k < tau(a.length, a.length - r);
//@ loop_invariant (\forall int i; 0 <= i && i < a.length && ! set[i] 
                    ==> a[i] == olda[i]);
//@ loop_invariant displaced == olda[v];
//@ loop_invariant (\forall int i; 
       k < i && i < tau(a.length, a.length - r)
       ==> ! set[mp(a.length, a.length - r, start, i)]);
//@ loop_invariant (\forall int i; 0 < i && i <= k
       ==> set[mp(a.length, a.length - r, start, i)]);
//@ loop_invariant (\forall int i, s; 
       0 <= i && i < tau(a.length, a.length - r) &&
       start < s && s < gcd(a.length, a.length - r)
       ==> ! set[mp(a.length, a.length - r, s, i)]);
//@ loop_invariant (\forall int i, s; 
       0 <= i && i < tau(a.length, a.length - r) &&
       0 <= s && s < start
       ==> set[mp(a.length, a.length - r, s, i)]);
//@ loop_invariant (\forall int i; 0 <= i && i < a.length && set[i] 
                    ==> a[i] == olda[ptor(i, a.length, r)]);
do
{
    k = k + 1;
    v = v + a.length - r;
    if (v >= a.length) { v = v - a.length; }
    int tmp = a[v]; a[v] = displaced; displaced = tmp;
    set[v] = true; // ghost
    lemma_mp_complete_cycle(a.length, a.length - r, start);
    lemma_mp_incomplete_cycle(a.length, a.length- r, start);
    lemma_mp_disjoint_cycles(a.length, a.length - r, start, k);
} while (k < tau(a.length, a.length - r));
\end{lstlisting}
\caption{Inner loop of the verified JML implementation of the rotation by modular visit algorithm of Figure~\ref{fig:rotation-by-modulo-algo}.}
\label{fig:jml:rotation-modulo-inner}
\end{figure}

\afterpage{\clearpage}
For framing, we use a Boolean array \<set> as we did in Boogie.\footnote{Given the modifications to the exit conditions, \<set> is the only genuinely \emph{ghost} variable in the JML implementation.}
We notice, however, that \escj has difficulties reasoning about the invariants that relate \<set> and \<a>'s content in relation to \<\\old(a)>'s.
As a simple solution, we introduce a pseudo-ghost variable \<olda>, which stores the content of \<a> upon method entry.
With this trick, \escj can prove the loop invariants involving \<set>, \<a>, and \<olda> that express the fundamental progress in rotating.
The drawback of this trick is that it fails to eventually relate \<olda> to \<\\old(a)>'s content, even if we add invariants asserting that \<olda> does not change after it is initialized.
Since \<olda> is local, it cannot appear in the postcondition either; hence we just \<assume> the final, trivial step in the correctness proof.

The inner loop of Figure~\ref{fig:jml:rotation-modulo-inner} uses Java's \<do$\ldots$while> loop.
It turns out that \escj requires invariants of \<do$\ldots$while> loops to hold also \emph{before} the first unconditional iteration.
To exemplify, consider the loop:
\begin{center}
\begin{lstlisting}[numbers=none]
int x = 0
//@ loop_invariant 0 <= x;
do {
   x = x + 1;
} while (x < 10);
\end{lstlisting}
\end{center}
Even though \<0 $\,<\,$ x> holds after the first unconditional iteration, \escj only accepts the weaker \<0 $\,\leq\,$ x> as invariant.
This explains the lower bounds of some of the invariants in Figure~\ref{fig:jml:rotation-modulo-inner}.

%%% Local Variables: 
%%% mode: latex
%%% TeX-master: "rotation.tex"
%%% End: 

%% file: jml-listings.tex
% @(#)$Id: jml-listings.tex,v 1.7 2010/02/25 21:57:06 leavens Exp $
%
% Copyright (C) 2006 Iowa State University
%
% This file is part of JML
%
% JML is free software; you can redistribute it and/or modify
% it under the terms of the GNU General Public License as published by
% the Free Software Foundation; either version 2, or (at your option)
% any later version.
%
% JML is distributed in the hope that it will be useful,
% but WITHOUT ANY WARRANTY; without even the implied warranty of
% MERCHANTABILITY or FITNESS FOR A PARTICULAR PURPOSE.  See the
% GNU General Public License for more details.
%
% You should have received a copy of the GNU General Public License
% along with JML; see the file COPYING.  If not, write to
% the Free Software Foundation, 675 Mass Ave, Cambridge, MA 02139, USA.
%
% A JML listings environment.
%
% AUTHOR: Gary T. Leavens
%
% requires listings i.e., do \usepackage{listings} first
%
% This file is set up to be used via \input{jml-listings}.
% Typically one would use \lstset{language=[JML]Java}
% to make the JML dialect of the Java language active, although there
% are other ways to do that in the listings package.
% If you want, you could make a version that is a style file,
% but then change \lstdefinelanguage to \lst@definelanguage below.
%
\lstdefinelanguage[JML]{Java}[]{Java}%
       {% C++ style comments have to start with a blank!
        comment=[l]{//\ },
        % And C-style comments must also start with a blank or star!
        morecomment=[s]{/*\ }{*/},        
        morecomment=[s]{/**}{*/},
        % sensitive=true, % inherited
        % Add JML keywords as level 1 keywords, so can typeset differently
        classoffset=1,
        % And here are all the wonderful JML keywords
        morekeywords={abrupt_behavior,abrupt_behaviour,
         accessible,accessible_redundantly,also,assert,assert_redundantly,
         assignable,assignable_redundantly,assume,assume_redundantly,
         axiom,behavior,behaviour,breaks,breaks_redundantly,
         callable,callable_redundantly,captures,captures_redundantly,
         choose,choose_if,code,code_bigint_math,code_java_math,
         code_safe_math,constraint,constraint_redundantly,constructor,
         continues,continues_redundantly,decreases,decreases_redundantly,
         decreasing,decreasing_redundantly,diverges,diverges_redundantly,
         duration,duration_redundantly,ensures,ensures_redundantly,
         example,exceptional_behavior,exceptional_behaviour,
         exceptional_example,exsures,exsures_redundantly,extract,field,
         forall,for_example,ghost,helper,hence_by,hence_by_redundantly,
         implies_that,in,in_redundantly,initializer,initially,instance,
         invariant,invariant_redundantly,loop_invariant,
         loop_invariant_redundantly,maintaining,maintaining_redundantly,
         maps,maps_redundantly,measured_by,measured_by_redundantly,method,
         model,model_program,modifiable,modifiable_redundantly,modifies,
         modifies_redundantly,monitored,monitors_for,non_null,
         normal_behavior,normal_behaviour,normal_example,nowarn,
         nullable,nullable_by_default,old,or,post,post_redundantly,
         pre,pre_redundantly,pure,readable,refine,refines,refining,represents,
         represents_redundantly,requires,requires_redundantly,
         returns,returns_redundantly,set,signals,signals_only,
         signals_only_redundantly,signals_redundantly,spec_bigint_math,
         spec_java_math,spec_protected,spec_public,spec_safe_math,
         static_initializer,uninitialized,unreachable,weakly,
         when,when_redundantly,working_space,working_space_redundantly,
         writable
        },
        % keywords from the universe type system
        morekeywords={rep,peer,readonly},
        % typeset everything that starts with a backslash as a keyword
        % BUG: this doesn't allow typesetting these keywords differently
        keywordsprefix=\\,
        otherkeywords={<:,<-,->,..,<==,==>,<==>,<=!=>},
        classoffset=0, % restore default class for keywords
        % CAF: added escaping
        escapeinside={(*}{*)},
        % CAF: added pretty-printing, adapted from Leino's Boogie style
          literate=%
          {:}{$\colon$}1
          {!}{$\lnot$}1
          {-}{$-$}1
          {+}{$+$}1
          {==}{$=$}1
          {!=}{$\neq$}1
          {&&}{$\land$}1
          {||}{$\lor$}1
          {<}{$<$}1
          {<=}{$\le$}1
          {>}{$>$}1
          {>=}{$\ge$}1
          {==>}{$\Longrightarrow$}3
          {<==}{$\Longleftarrow$}3
          {<==>}{$\Longleftrightarrow$}4
          {\\forall}{$\forall$}1
          {\\exists}{$\exists$}1
          {lambda}{$\lambda$}1,
  mathescape=true,
  numberstyle=\tiny,
  basicstyle=\ttfamily,
  commentstyle=\color[HTML]{F10000}\itshape,
  keywordstyle={\color[HTML]{222277}\bfseries},
  ndkeywordstyle={\color[HTML]{222277}\bfseries},
}
\lstnewenvironment{jml}[1][]{%
  \lstset{language={[JML]Java},firstnumber=last, numbers=left, 
     aboveskip=-3pt, belowskip=-3pt, 
     xleftmargin=8pt,xrightmargin=8pt,basicstyle=\normalsize\ttfamily,#1}}{}